\documentclass[10pt,prd,aps,eqsecnum,superscriptaddress,nofootinbib,notitlepage,longbibliography]{revtex4-2}

\usepackage[a4paper, left=2.2cm, right=2.2cm, top=2.5cm, bottom=2.5cm]{geometry}
\usepackage{setspace}
\usepackage[english]{babel}
\addto\extrasenglish{%
}
\usepackage[utf8]{inputenc}
\usepackage[T1]{fontenc}
\usepackage{lmodern}
\usepackage{amsmath,amsthm,mathdots,amssymb,bm,bbm,mathrsfs,mathtools}
\usepackage[final]{graphicx} 
\usepackage{subcaption}
\usepackage{tikz,pgfplots}\pgfplotsset{compat=1.16}

\usepackage{verbatim}

\usepackage{thm-restate}
\newtheorem{thm}{Theorem}[section]

\newtheorem{lem}{Lemma}[section]

\newtheorem{rmk}{Remark}[section]

\newtheorem{example}{Example}[section]

\newtheorem{defn}{Definition}[section]

\usepackage{comment}

\usepackage[colorlinks=true, urlcolor=violet, linkcolor=blue, citecolor=red, hyperindex=true, linktocpage=true, pagebackref=true, draft=false]{hyperref}
\usepackage{mleftright}\mleftright 

\usepackage{enumitem}
\mathtoolsset{showonlyrefs=true}
\usepackage{microtype,color,graphicx} 
\usepackage{tikz-cd}
\usetikzlibrary{shapes.callouts, arrows.meta, positioning, calc}

\usepackage{xcolor}
\usepackage{multirow}
\usepackage{tablefootnote}

\renewcommand{\vec}{\bm}
\newcommand{\CA}{\mathcal{A}}

\newcommand{\CC}{\mathcal{C}}
\newcommand{\BC}{\mathbb{C}}

\newcommand{\CE}{\mathcal{E}}

\newcommand{\CL}{\mathcal{L}}
\newcommand{\CM}{\mathcal{M}}
\newcommand{\CN}{\mathcal{N}}

\newcommand{\CO}{\mathcal{O}}

\newcommand{\CR}{\mathcal{R}}
\newcommand{\BR}{\mathbb{R}}

\newcommand{\vA}{\bm{A}}

\newcommand{\vC}{\bm{C}}

\newcommand{\vH}{\bm{H}}
\newcommand{\vI}{\bm{I}}

\newcommand{\vM}{\bm{M}}

\newcommand{\vO}{\bm{O}}
\newcommand{\vP}{\bm{P}}

\newcommand{\vQ}{\bm{Q}}

\newcommand{\vS}{\bm{S}}

\newcommand{\vU}{\bm{U}}
\newcommand{\vV}{\bm{V}}
\newcommand{\vW}{\bm{W}}
\newcommand{\vX}{\bm{X}}
\newcommand{\vY }{\bm{Y }}
\newcommand{\vZ}{\bm{Z}}

\newcommand{\vsigma}{\bm{ \sigma}}

\newcommand{\vtau}{\bm{ \tau}}

\newcommand{\vrho}{\bm{ \rho}}
\renewcommand{\L}{\left}
\newcommand{\R}{\right}

\newcommand{\dagg}{\dagger}

\newcommand{\vertiii}[1]{{\left\vert\kern-0.25ex\left\vert\kern-0.25ex\left\vert #1 \right\vert\kern-0.25ex\right\vert\kern-0.25ex\right\vert}}
\newcommand{\norm}[1]{\Vert {#1} \Vert}

\newcommand{\normp}[2]{\norm{#1}_{#2}}
\newcommand{\lnormp}[2]{\lnorm{#1}_{#2}}

\newcommand{\labs}[1]{\left\vert {#1} \right\vert}
\newcommand{\lnorm}[1]{\left\Vert {#1} \right\Vert}
\newcommand{\e}{\mathrm{e}}

\newcommand{\ri}{\mathrm{i}}
\newcommand{\rd}{\mathrm{d}}
\newcommand*{\tr}{\mathrm{Tr}}

\newcommand{\indicator}{\mathbbm{1}}

\newcommand{\nrm}[1]{\left\| #1 \right\|}
\newcommand{\ipc}[2]{\left\langle#1,#2\right\rangle}

\newcommand{\undersetbrace}[2]{ \underset{#1}{\underbrace{#2}}}

\DeclarePairedDelimiterX{\braket}[1]{\langle}{\rangle}{#1}
\DeclarePairedDelimiterX\ketbra[2]{| }{|}{#1 \delimsize\rangle\!\delimsize\langle #2}	
\DeclarePairedDelimiterX\dotp[2]{\langle}{\rangle}{#1, #2}

\DeclareMathAlphabet{\dutchcal}{U}{dutchcal}{m}{n}
\SetMathAlphabet{\dutchcal}{bold}{U}{dutchcal}{b}{n}
\DeclareMathAlphabet{\dutchbcal} {U}{dutchcal}{b}{n}

\newcommand{\under}[2]{\underbrace{#1}_{\substack{#2}}}

\makeatletter
\DeclareRobustCommand*{\pmzerodot}{%
	\nfss@text{%
		\sbox0{$\vcenter{}$}
		\sbox2{0}%
		\sbox4{0\/}%
		\ooalign{%
			0\cr
			\hidewidth
			\kern\dimexpr\wd4-\wd2\relax 
			\raise\dimexpr(\ht2-\dp2)/2-\ht0\relax\hbox{%
				\if b\expandafter\@car\f@series\@nil\relax
				\mathversion{bold}%
				\fi
				$\cdot\m@th$%
			}%
			\hidewidth
			\cr
			\vphantom{0}
		}%
	}%
}

\usepackage{ifdraft}
\ifdraft{
	\newcommand{\authnote}[3]{{\color{#3} {\bf  #1:} #2}}	
}{
	\newcommand{\authnote}[3]{}
}

\renewcommand{\thefootnote}{\fnsymbol{footnote}}

\usetikzlibrary{quantikz2}
\makeatletter
\def\l@subsubsection#1#2{}
\makeatother
\begin{document}

\title{A Structural Theory of Quantum Metastability: \\ Markov Properties and Area Laws}

\author{Thiago Bergamaschi$^\ddagger$}
\affiliation{University of California, Berkeley, CA, USA}
\author{Chi-Fang Chen$^\ddagger$}
\email{achifchen@gmail.com}
\affiliation{University of California, Berkeley, CA, USA}
\author{Umesh Vazirani}
\affiliation{University of California, Berkeley, CA, USA}

\begin{abstract}
Statistical mechanics assumes that a quantum many-body system at low temperature can be effectively described by its Gibbs state. However, many complex quantum systems exist only as metastable states of dissipative open system dynamics, which appear stable and robust yet deviate substantially from true thermal equilibrium. In this work, we model metastable states as approximate stationary states of a quasi-local, (KMS)-detailed-balanced master equation representing Markovian system-bath interaction, and unveil a universal structural theory: \textit{all} metastable states satisfy an area law of mutual information and a Markov property. The more metastable the states are, the larger the regions to which these structural results apply. Therefore, the hallmark correlation structure and noise resilience of Gibbs states are not exclusive to true equilibrium but emerge dynamically. Behind our structural results lies a systematic framework encompassing sharp equivalences between local minima of free energy, a non-commutative Fisher information, and approximate detailed balance conditions. Our results build towards a comprehensive theory of thermal metastability and, in turn, formulate a well-defined, feasible, and repeatable target for quantum thermal simulation. 
\end{abstract}

\maketitle

\def\thefootnote{$\ddagger$}\footnotetext{Co-first authors. }\def\thefootnote{\arabic{footnote}}

\section{Introduction}

Statistical mechanics postulates that a system in contact with a thermal bath is described by a Gibbs state. Without addressing \textit{how} the system reaches thermal equilibrium, studying Gibbs states across a variety of Hamiltonian families has proven remarkably explanatory, from magnetism \cite{Ising:1925em,Onsager44} to superconductivity \cite{BCS57,KT73}. Ever since, this static approach to quantum many-body systems had served as a precise, self-contained, if minimal, starting point for low-temperature physics.

Much of the power of equilibrium theory rests on the distinctive structures of Gibbs states that are absent in generic dynamical problems. Most notably, the thermal area law for Gibbs states~\cite{wolf2008area} bounds the \textit{mutual information} between a region and its complement by the surface area (Fig \ref{fig:arealaw}), instead of its volume, significantly reducing the set of physically viable states in an exponentially large Hilbert space. Exploiting such a structure~\cite{hastingsAreaLaw} has been the driving force behind theoretical and numerical studies of one- and two-dimensional quantum systems at low temperatures. In fact, classical Gibbs distribution enjoys a closely related, even stronger~\textit{Markov property}: any set of spins, conditioned on its boundary, is independent from the rest. This \textit{conditional independence} structure is often the starting point of any analytical or algorithmic study of Gibbs states. For quantum Gibbs states, however, the search for universal Markov properties has remained an ongoing effort~\cite{kuwahara2024clustering,chen2025GibbsMarkov}.

\begin{figure}[t]
\centering

\begin{subfigure}[t]{0.35\textwidth}
  \centering
  \includegraphics[width=\linewidth]{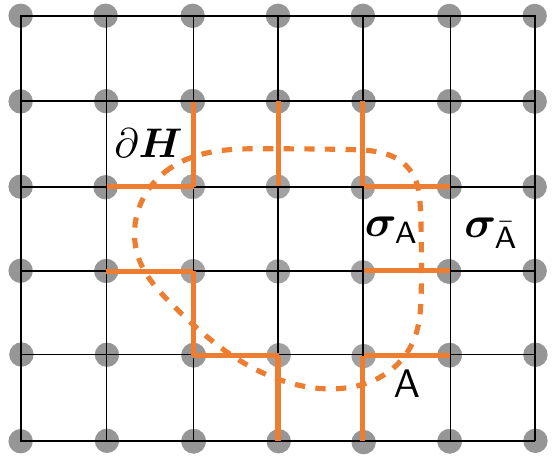}
  \caption{ An Area Law across the region $\mathsf{A}$}
  \label{fig:arealaw}
\end{subfigure}
\begin{subfigure}[t]{0.53\textwidth}
  \centering
  \includegraphics[width=1.0\linewidth]{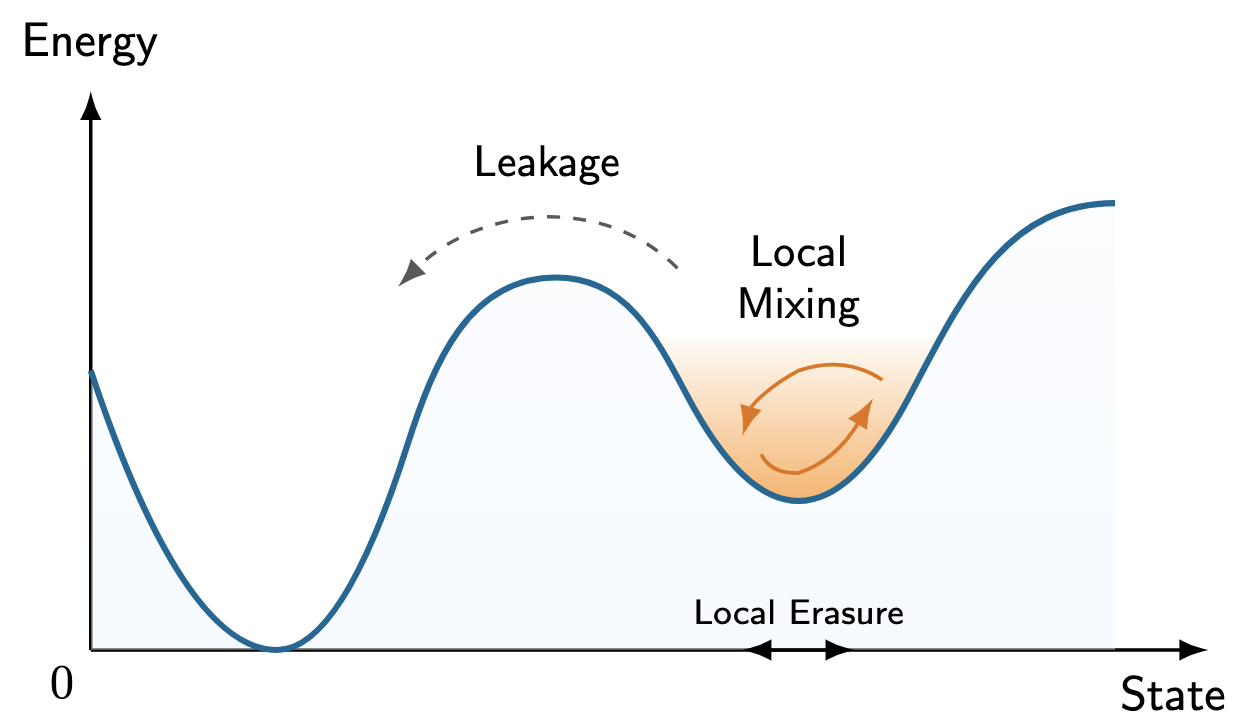}
  \caption{A metastable state in the energy landscape}
  \label{fig:twotimescales}
\end{subfigure}

\captionsetup{justification=raggedright,singlelinecheck=false}
\caption{(a) Given a metastable state $\vsigma$ and a Hamiltonian $\vH$, we prove that sufficiently small regions $\mathsf{A}$ admit an area law: the mutual information between $\mathsf{A}$ and its complement $\bar{\mathsf{A}}$ scales with the strength of the Hamiltonian on the boundary $\mathsf{I}(\mathsf{A}:\mathsf{\bar{A}})_{\vsigma} \le 2\beta\cdot \norm{\partial \vH}$. (b) Underlying this behavior is an interpretation of metastable states as local minima of the free energy, in approximate thermal equilibrium within a ``well'' in the energy landscape. Local perturbations to the state take it out-of-equilibrium, triggering a local mixing process. So long as the lifetime of the metastable state (the time to escape the surrounding energy well) exceeds the local mixing timescale, the system re-equilibrates within the well.} 

\label{fig:area_timescales}
\end{figure}

Future quantum simulators promise a new computational lens into quantum many-body physics that would otherwise be intractable. Since Feynman~\cite{feynman1982SimQPhysWithComputers}, simulation of the fundamental dynamical laws of physics — notably Hamiltonian evolution on prescribed initial states — has found an arsenal of efficient quantum algorithms (e.g., \cite{low2016HamSimQSignProc,gilyen2018QSingValTransf,childs2021theory}). Nevertheless, many-body phenomena that can be reliably observed at low temperatures -- from materials to molecules -- appear strikingly \textit{static}, robust, and fall into a handful of well-defined instances. To effectively bridge this contrast, statistical mechanics provided a succinct, well-defined equilibrium guiding framework. Ironically, the corresponding Gibbs sampling problem is provably NP-hard at the very least, which is widely believed to be intractable, even for quantum computers. What, then, really counts as a meaningful and algorithmically efficient quantum thermal simulation problem?

If we regard Nature's thermalization dynamics as a quantum computational process, genuine equilibrium may likewise be physically infeasible. Indeed, many existing complex quantum systems at finite temperatures, ranging from nuclei to molecules, are not in their global Gibbs states but appear \textit{metastable} for extraordinarily long lifetimes. In complex systems such as spin glasses \cite{AJBray_1980} or quantum memories \cite{Alicki2008OnTS}, folklore suggests that metastability arises from a rugged energy landscape, where different initializations may land in a zoo of local minima, and simple, local thermalization dynamics fail to overcome free energy barriers between them. These two additional inputs —the initial condition and the dynamical mechanism of thermalization —amend the static view of statistical mechanics and hint at phenomena absent in equilibrium, motivating a systematic formulation of quantum metastability. 
 
Even if genuine thermal equilibrium is nonphysical, Gibbs states are at least mathematically well-defined. In contrast, posing a realistic microscopic dynamic model of thermalization seemingly requires a mouthful of details. In spite of this ambiguity, in classical spin systems, \textit{Glauber dynamics}~\cite{glauber} cleared a path forward by extracting a self-contained, precise model of thermalization that updates individual spins randomly in a detailed-balanced manner. To date, Glauber dynamics has essentially served as the go-to equation of motion for the thermalization of many-body spin systems, culminating in a flourishing subject of Markov chain Monte Carlo algorithms and their mixing times \cite{Markovchain_mixing}.

In this spirit, we seek a minimal entry point to a precise quantum theory of out-of-equilibrium physics by modeling Nature's thermalization process as a time-independent master equation. Specifically, we demand that the dynamics respect locality, which adapts to many-body systems, and satisfy quantum detailed balance, which ensures eventual equilibration to the exact Gibbs state. The recently discovered master equation~\cite {chen2023efficient} precisely matches both properties. Under this model, a natural formulation of metastability that we explore is the set of \textit{approximate stationary states} of the master equation, i.e., states changing slowly in trace distance. This definition encapsulates the true Gibbs state, but allows for a richer set of out-of-equilibrium states that may arise from unknown initial conditions, providing a self-contained theoretical starting point for metastability. 

In this work, we unveil a universal structural theory that characterizes \textit{all} metastable states at any temperature, regardless of initial conditions. Strikingly, the hallmark features of Gibbs states are not exclusive to true equilibrium but arise dynamically: metastable states also satisfy an area law for mutual information, which severely constrains the viable quantum correlation structure, and a closely related local Markov property, which characterizes resilience to local perturbation. The more metastable the states are, the larger the regions to which these structural results apply. 

Why should any structure emerge in metastable states that evolved from an arbitrary, strongly correlated initial state? It is instructive to have the classical analogy in mind. Recall, for classical Gibbs states, the area law and the Markov property essentially capture the same intuition: if the statistics of a spin region are fully mediated by the boundary, then the mutual information cannot exceed the size of the boundary. The same picture nearly applies to classical metastable states, up to an approximate notion of conditional independence. The most natural and useful definition -- that will eventually inspire the quantum argument -- is a dynamical feature: if one were to erase any local patch of the metastable state (by replacing it with an arbitrary configuration), and subsequently run Glauber dynamics conditioned on its boundary, then it would approximately recover the global metastable state (Fig \ref{fig:glauber_subfigure}). Although this erasure necessarily causes a huge disturbance to the state, the robustness to such violence, a \textit{local Markov property}, is precisely the hallmark of conditional independence. Indeed, this dynamical argument implies the static statement that classical metastable states have approximately detailed-balanced conditional measures, on small enough regions.

Nevertheless, the essential workhorse of the classical argument -- conditioning on the boundary -- is simply ill-defined in the quantum case, or at least cannot be exactly true. Fundamentally, quantum measurements and operations can be destructive. In particular, running the Lindbladian dynamics with jumps targeting a subset of qubits — the counterpart to the conditional Glauber dynamics — necessarily back-reacts on the \textit{full system} with a strength that decays with distance (Fig~\ref{fig:lindblad_subfigure}). Indeed, in recent developments in approximate quantum Markov properties~\cite[]{chen2025GibbsMarkov}, the recovery map for an erased region has an effective radius far too large to derive any sensible area law for the Gibbs states. 

The resolution to these obstacles lies in another conceptual ingredient, the \textit{variational} characterization of quantum Gibbs states: the Gibbs state is the unique minimizer of the \textit{free energy}, which captures a tradeoff between energy and entropy. The original proof of the thermal area law for Gibbs states~\cite{wolf2008area} boils down to a balance between the entropy, which determines the mutual information, and the energy, which reflects the geometric locality of the Hamiltonian. The insight is that this template extends to suitable \textit{local minima} of free energy, where the Gibbs state, as the global minimum, is merely a special case. 

The revelation that ties the proofs together is that the local Markov property, in disguise, really characterizes a desirable local minimum condition: \textit{arbitrary localized disturbances to a metastable state, such as erasing a small region of qubits, can be restored to the \textit{same} metastable state by local Lindbladian dynamics.} Since the free energy monotonically decreases under the Lindbladian, the metastable state must be a local minimum against any local erasure errors. Heuristically, one should envision an energy landscape with multiple energy wells, where metastable states lie at the bottoms of the wells, behaving locally like a Gibbs state, and serving as local minima of free energy. These local perturbations take the metastable state out of equilibrium, triggering a local mixing process within the well. However, so long as the lifetime of the metastable state (the time to escape the well) exceeds the time to re-equilibrate within the well, then this process actually converges back to the original metastable state (Fig~\ref{fig:twotimescales}).

Underlying our main results is a framework that encompasses sharp equivalences between manifestations of metastability: the approximate stationarity of the Lindbladian, the local minima of free energy, the decay of a certain non-commutative Fisher information, and ultimately, static characterizations of metastable states as those satisfying approximate detailed balance conditions. At the heart of these equivalences is a theory of quantum optimal transport~\cite{carlen2012analog2wassersteinmetricnoncommutative, carlen2017gradientflowentropyinequalities, Datta2017RelatingRE, Carlen_2019, gao2021riccicurvaturequantumchannels} that we build for the new family of KMS-detailed balance Lindbladians~\cite{chen2023efficient}.

In turn, our framework revisits the long-standing problem of simulating quantum physics, a challenge that has co-evolved with quantum computation since Feynman’s proposal~\cite{feynman1982SimQPhysWithComputers}. Whereas general dynamical evolution appears overly expressive~\cite{kitaev2002classical,Feynman1986} and sensitive to initial conditions, and genuine equilibrium states are computationally intractable, metastable states bridge these extremes: they are both structured and algorithmically accessible. Furthermore, the bottom-of-well picture suggests that metastable phenomena may exhibit extraordinary stability despite long-range correlations, even in the presence of noise or repeated measurements (see Section \ref{sec:aspect_Qsim} for a mechanism by which this could emerge). 
We hope our study extends the celebrated mathematical precision of statistical mechanics outside of equilibrium, formulating a well-defined target for quantum simulation of low-temperature physics, and ultimately building towards a theory of quantum thermal metastability. 

\begin{figure}[t]
\centering

\begin{subfigure}{0.45\textwidth}
  \centering
  \includegraphics[width=0.85\linewidth]{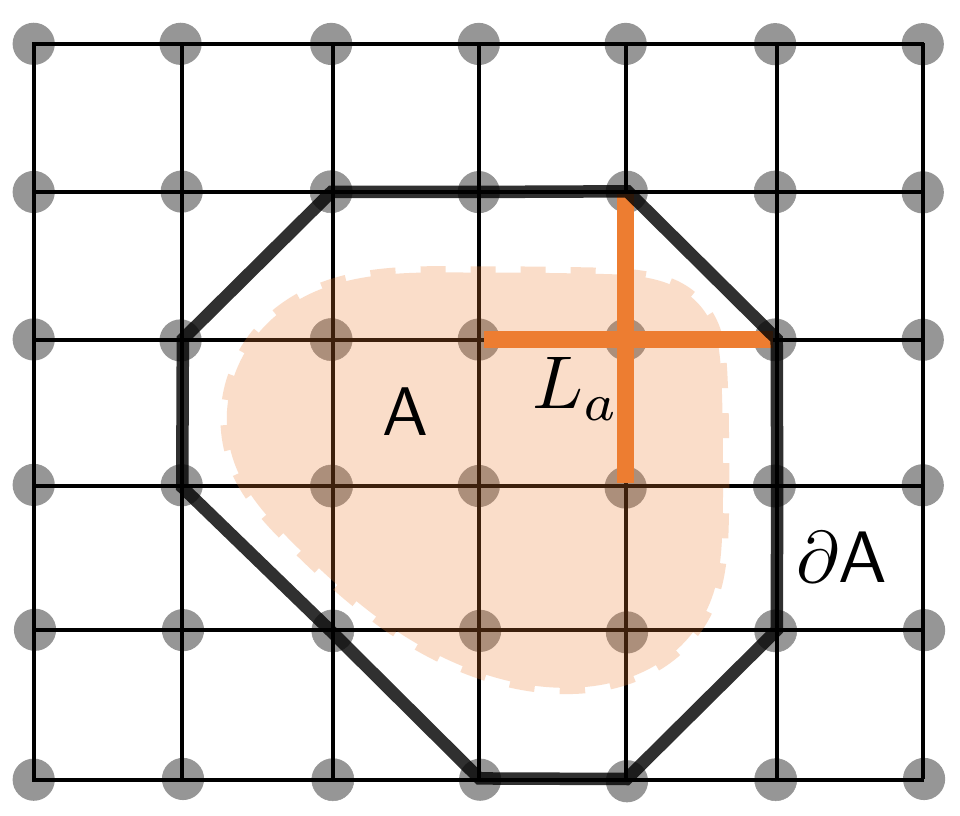}
  \caption{Classical Glauber dynamics on the region $\mathsf{A}$}
  \label{fig:glauber_subfigure}
\end{subfigure}
\begin{subfigure}{0.45\textwidth}
  \centering
  \includegraphics[width=0.85\linewidth]{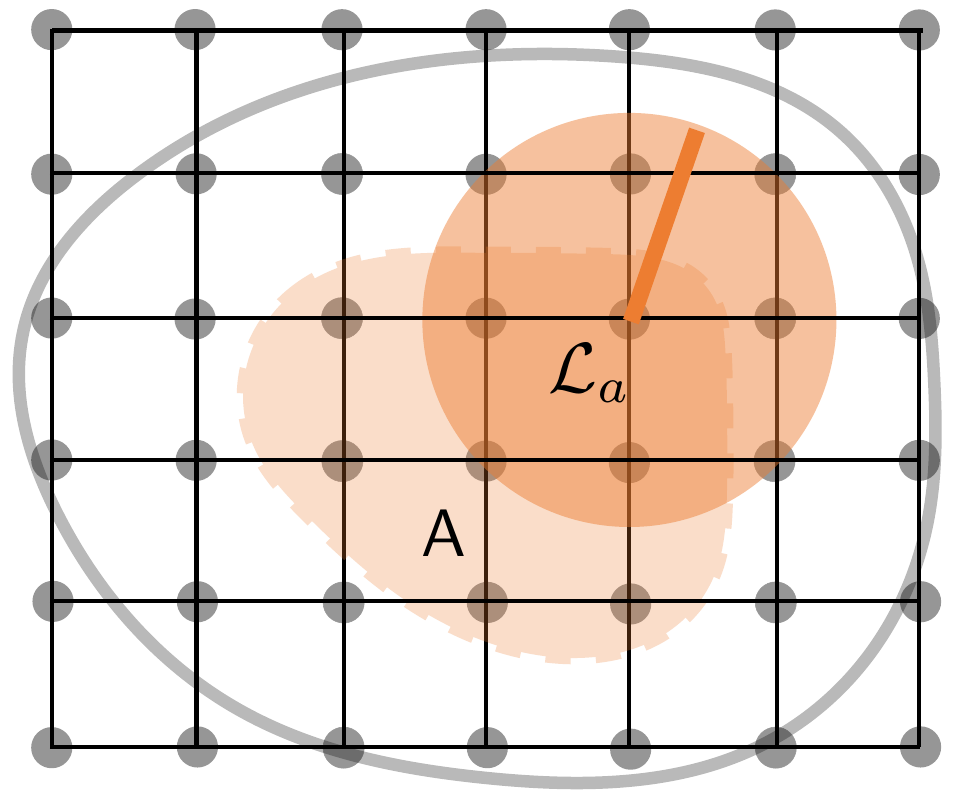}
  \caption{Lindbladian dynamics with jumps on the region $\mathsf{A}$}
  \label{fig:lindblad_subfigure}
\end{subfigure}

\captionsetup{justification=raggedright,singlelinecheck=false}
\caption{(a) Running Glauber dynamics on a restricted region $\mathsf{A}$, resamples the configuration on $\mathsf{A}$, from the Gibbs measure conditioned on the current boundary configuration. (b) In a quantum system, one can likewise define a resampling map for a region of qubits $\mathsf{A}$ by activating only a subset of the Lindbladian terms $\sum_{a\in \mathsf{A}} \CL_a.$ However, the resulting map is only quasi-local, which may perturb far-away qubits.}
\label{fig:glauber_lindblad}
\end{figure}

\section{Main results}
\label{section:results}

\paragraph{Setup.} Consider an $n$-qubit quantum system described by a few-body Hamiltonian $\vH$ with bounded \textit{interaction degree}, including, but not restricted to, geometrically local systems (see Section~\ref{sec:Ham}). We study metastable states that arise from coupling the system to a thermal bath at some inverse temperature $\beta$, as modeled by a time-independent master equation, or a \textit{Lindbladian}
\begin{align}
    \frac{\rd \vsigma}{\rd t}  = \CL[\vsigma].
\end{align}

In the spirit of Glauber dynamics, we require our minimal workable model of thermalization to satisfy quantum (Kubo-Martin-Schwinger) detailed balance and thus fix the Gibbs state, and yet respect the locality of the Hamiltonian. The Lindbladians constructed in~\cite{chen2023efficient} suit exactly this purpose:
\begin{align}
    \CL[\cdot] := -i[\vH,\cdot ] + \eta\sum_{a\in \CA}\CL_a[\cdot]\quad \text{such that}\quad \CL_a[\vrho] = 0 \quad \text{where}\quad \vrho := \frac{e^{-\beta \vH}}{\tr[e^{-\beta \vH}]}. \label{eq:lindblad_def}
\end{align}
The coupling strength is set to $\eta=1$ throughout the paper for simplicity; changing $\eta$ amounts to effectively rescaling $\epsilon.$ The precise form of $\CL_a$ depends on the Hamiltonian, the temperature, and a set of jump operators that couples the system to the bath--which we take to be single-qubit Pauli operators $\CA = \{\vX_i,\vY_i,\vZ_i\}_{i\in [n]}$  (Section~\ref{section:the_lindbladian}). The above effective model may not capture every detail of realistic system-bath dynamics; nevertheless, it provides a mathematically well-defined starting point by studying their \textit{approximate stationary states}, as our main model of metastability.
\begin{defn}
    [Metastability as approximate stationarity] \label{def:metastab} We say a quantum state $\vsigma$ is $\epsilon$-\emph{metastable} if it is approximately stationary under the Lindblad evolution~\eqref{eq:lindblad_def} 
    \begin{equation}
   \big\|\CL[\vsigma]\big\|_1 \leq \epsilon. \label{eq:def_meta}
\end{equation}
\end{defn}
The exact stationary state is uniquely the Gibbs state due to detailed balance and ergodicity of $\CL$. Yet, without a priori bounds on the mixing time, an approximately stationary state may be statistically far from the Gibbs state. Even if the above characterization appears very general, disregards initial conditions, and only captures an instant of the dynamics, our framework reveals that such minimal assumptions nevertheless have far-reaching structural consequences that draw significant parallels to the true Gibbs state.\\ 

\paragraph{An Area Law for Metastable States.} For a sufficiently metastable state, small enough regions must satisfy an area law of mutual information, just like the idealized Gibbs state. Recall the entropic definition of mutual information:
\begin{equation}
    \mathsf{I}(\mathsf{A}:\mathsf{B})_{\vsigma} := S(\vsigma_\mathsf{A})+S(\vsigma_\mathsf{B})-S(\vsigma_\mathsf{AB}), \quad \text{ where }\quad S(\vsigma) = -\tr[\vsigma\log \vsigma]
\end{equation}
where $\vsigma_\mathsf{A}$ denotes the marginal of $\vsigma$ on a region $\mathsf{A}.$ 

\begin{thm}[Metastability Implies an Area Law]\label{thm:main_meta_implies_area}
Consider a quantum system in a thermal bath, governed by the thermal Lindbladian $\CL$ of~\eqref{eq:lindblad_def}. Then, for any $\epsilon$-metastable state $\vsigma$ and a region $\mathsf{A}\subset [n]$, the bipartite mutual information of $\mathsf{A}$ with its complement $\bar{\mathsf{A}}$ satisfies
    \begin{align}
        \mathsf{I}(\mathsf{A}:\bar{\mathsf{A}})_{\vsigma} \le 2 \beta \cdot \|\partial\vH\| + n^2\cdot \epsilon^{\lambda} \cdot  e^{\mu |\mathsf{A}|} ,
    \end{align}
    where $\partial\vH$ are the Hamiltonian terms crossing the cut $(\mathsf{A}, \bar{\mathsf{A}})$, and the constants $\mu, \lambda> 0$ depend polynomially on $\beta$ and the interaction degree of the Hamiltonian. 
\end{thm}

In the $\epsilon\rightarrow 0$ limit, our statement above exactly recovers the Gibbs area law \cite{wolf2008area}, substantially constraining the viable type of quantum correlation that exists at low temperatures. The correction to the area law decays algebraically with metastability, whereas it grows exponentially with the region size.\footnote{The factor of system size $n^2$ came from standard entropic continuity bounds.} Consequently, a given $\epsilon$-metastable state has \textit{all} $\sim \log(1/n\epsilon)$-sized regions satisfying an area law.

A pertinent question is then, how stationary can a metastable state get, while remaining statistically far from equilibrium? Certainly, if the \textit{lifetime} $\sim 1/\epsilon$ of the metastable state -- after which it may have macroscopically changed -- exceeds the \textit{mixing time} of the global dynamics, then it must have already reached the Gibbs state. Physically, the lifetime of a metastable state is expected to scale exponentially with surrounding (free) energy barriers, which can be macroscopic in, for instance, classical (e.g., 2D ferromagnetic Ising) and quantum memories (e.g., 4D toric code). This, in turn, suggests that metastability area laws could be non-vacuous at macroscopic scales. 

As we discuss shortly, the exponential factor $e^{\mu |\mathsf{A}|}$ is a generic, worst-case bound for all (bounded-degree) Hamiltonians and may be quantitatively improved in specific systems (see Section~\ref{section:discussion}). Nevertheless, the qualitative scaling of our bound is essentially tight, even in the case of classical Ising models.

\begin{figure}[t]
    \centering
    \includegraphics[width=0.65\linewidth]{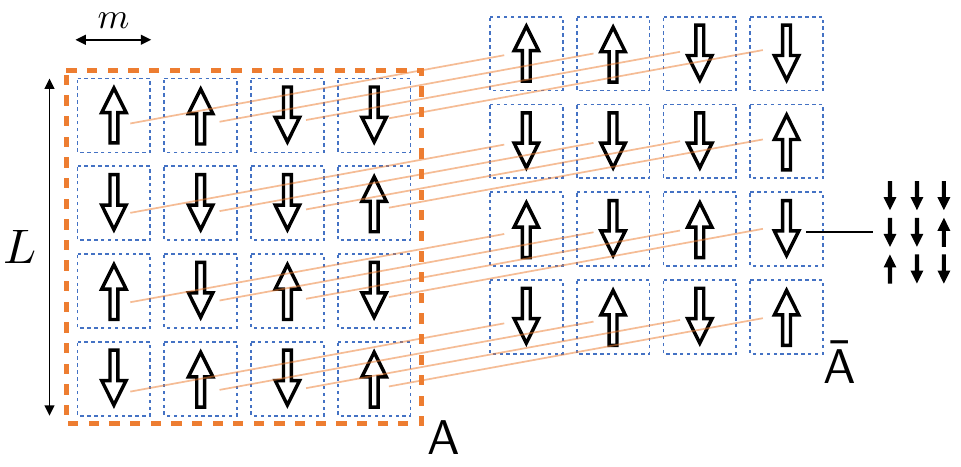}
    \captionsetup{justification=raggedright,singlelinecheck=false}
    \caption{Copies of Ising ``hard disks''. Consider a $\mathsf{L}\times 2\mathsf{L}$ square lattice of bits, partitioned into $m\times m$ blocks, where ferromagnetic Ising interactions are placed only \textit{within} each block. The left $\mathsf{L}\times \mathsf{L}$ square is initialized randomly to encode either 0 (up) or 1 (down) with suitable metastable states (Gibbs measure conditioned on majority 0 or 1). The right $\mathsf{L}\times \mathsf{L}$ square is simply a duplicate of the left, which maximally correlates the encoded $(\mathsf{L}/m)^2$ bits of information.}
    \label{fig:ising}
\end{figure}

\begin{example}
\label{example:reflected_ising} Below a critical constant temperature, the Ising ``hard disks'' defined in Fig~\ref{fig:ising} accommodate an $\sim \mathsf{L}^2 \cdot e^{-\Omega(m)}$ metastable state whose mutual information across the midway cut is at least $(\mathsf{L}/m)^2$. 
\end{example}
\noindent \autoref{example:reflected_ising} implies that the complexity in the initial state -- in the form of a (near-)volume law of mutual information -- can persist at length-scales larger than the scale $\sim \log(1/n\epsilon)$ guaranteed by \autoref{thm:main_meta_implies_area}. For instance, when $m=L^\delta$ for any constant $\delta$, the mutual information across the central cut is $L^{2(1-\delta)}$. This suggests that the dependence on $|\mathsf{A}|$ in \autoref{thm:main_meta_implies_area} cannot generically be improved past sub-exponential.\\

\paragraph{Metastable States are Locally Markov.} If the area law fundamentally constrains thermally stable quantum correlation, the Markov property, then, characterizes the extent to which metastable states are noise-resilient. In classical Gibbs measures, the Markov property is often understood as an exact conditional independence structure~\cite{clifford1990markov} across partitions of the system. For the case of non-commuting Hamiltonians, however, an exact quantum Markov property is simply false, and only recently have suitable approximate Markov properties been shown. In particular, \cite{chen2025GibbsMarkov} provided a template that inspires our characterization of metastability, known as a \textit{local Markov property}: local noise (erasures) on the Gibbs state can be corrected (quasi-)locally. While more static forms of Markov properties are available -- e.g., decay of conditional mutual information -- this noise-recovery formulation appears most natural and workable.

\begin{thm}[Metastable States are Locally Markov]\label{thm:main_meta_implies_markov}
In the setting of~\autoref{thm:main_meta_implies_area}, for any $\epsilon$-metastable state $\vsigma$ and region $\mathsf{A}$, there exists a detailed-balanced quasi-local recovery map $\CR$ such that for any noise channel $\CN_{\mathsf{A}}$ on $\mathsf{A}:$ 
\begin{align}
\|\vsigma - \CR[\CN_{\mathsf{A}}[\vsigma]]\|_1 \leq n\cdot e^{\mu \mathsf{|A|}}\cdot \epsilon^{\lambda}.
\end{align}
\noindent The parameter $\mu, \lambda>0$ depends polynomially on $\beta$ and the interaction degree of the Hamiltonian.
\end{thm}

The map $\CR$ that recovers the metastable state is precisely the (quasi-)local Lindbladian dynamics~\eqref{eq:lindblad_def} with jumps selected from the region $\mathsf{A},$ which is detailed-balanced and thus fixes the Gibbs state $\CR[\vrho] = \vrho$. Since $\CR$ is quasi-local, \autoref{thm:main_meta_implies_markov} implies a generic decay of the conditional mutual information in the metastable state under suitable tripartitions, fully extending \cite{chen2025GibbsMarkov}'s results for the genuine Gibbs state. 

Therefore, metastable states must enjoy a natural form of robustness to noise, at least locally. Since the perturbed state $\CN_{\mathsf{A}}[\vsigma]$ is fundamentally \textit{not} metastable, this process could re-equilibrate to some possibly distinct metastable state. The main surprise is that if the time-scale required for ``local mixing'', is sufficiently faster than the lifetime of the metastable state, then the local Lindbladian dynamics actually approximately recovers the \textit{original} metastable state. Although classically this statement has a straightforward proof via conditioning (see Section~\ref{section:warmup}), the challenge in the non-commutative setting is to understand when and how back-reaction to the boundary does not cause significant drift towards another metastable state.

As alluded, the above two facets of metastability (\autoref{thm:main_meta_implies_area}, \autoref{thm:main_meta_implies_markov}) are ultimately tied together through an intermediate notion: metastable states are \textit{local minima} of free energy. Since any detailed balanced recovery map must monotonically decrease the free energy (i.e., the data processing inequality for the relative entropy), the noise channel must approximately increase the free energy of the metastable state:
\begin{align}
        F(\vsigma)\approx F(\CR[\CN_{\mathsf{A}}[\vsigma]]) \le F(\CN_{\mathsf{A}}[\vsigma])\quad \text{where}\quad F(\vsigma) := \tr[\vH\vsigma] -\frac{1}{\beta}S(\vsigma). \label{eq:Fsigma_Fmarginal}
\end{align}

The area law, then, immediately follows from instantiating the erasure channel $\CN_{\mathsf{A}}[\vsigma]=\vsigma_{\mathsf{A}}\otimes \vsigma_{\bar{\mathsf{A}}}$ and the original variational arguments~\cite{wolf2008area}, see Section~\ref{section:methods}. 
Therefore, the error in our area law is dominated by the error in the Markov property, and most of our technical effort is then devoted to proving the latter.

\subsection{New aspects of quantum simulation}\label{sec:aspect_Qsim}

In addition to the fundamental characterization of metastable states, we highlight the algorithmic consequences of our framework, further expositions into even stronger forms of metastability, and their potential for efficient quantum simulation.

\textit{Algorithmic feasibility.} While equilibrium states are generally computationally intractable, metastable states—though structurally similar—are always algorithmically accessible by design: \textit{any} initial state becomes metastable after evolving for a randomly chosen time. Therefore, the expectation value of any observable, such as the energy~\cite{chen2024local}, becomes stationary.  See~\autoref{sec:metastable_time_averaging} for an exposition.
\begin{lem}[Time-averaging forces metastability]\label{lem:stab_typical_times} Fix a positive time $t>0$, and any initial state $\vsigma_0$. Then, the time-averaged state $ = \frac{1}{t}\int^t_0  e^{s\CL}[\vsigma_0]\cdot \rd s$ is metastable with error $\leq\frac{2}{t}$.
\end{lem} 

Although \autoref{lem:stab_typical_times} is by no means the only mechanism producing metastable states, it already points out that metastability holds at physically meaningful time-scales and can always be algorithmically accessed to any inverse-polynomial precision. In some systems, randomization over time may not be necessary, and the convergence to metastability is expected to be much faster (see discussion~\autoref{section:discussion}). Therefore, the statistical properties of the metastable state $\overline{\vsigma_t}$ can be reliably probed after sufficiently many independent runs of re-initialization, evolution, and measurement. 

In the bottom-of-wells picture, in general, a metastable state can be a mixture of "pure" metastable states sitting in a single well. In the initialize-evolve-measure-repeat process described above, it will be nontrivial in general to identify an initial condition corresponding to a specific well, and so the resulting data will be a mixture of different wells. In the context of quantum simulation, a challenge is whether it is possible to achieve robustness and repeatability, in the following sense: is it possible to prepare a metastable state in a single well and make a large number of measurements on it? A possible approach to this question is explored below. 

\textit{Robustness and repeatability of metastable states.} 
If the Markov property rapidly restores local disturbances to the metastable wells, it is natural to ask whether the above tomography process can be accelerated by interleaving measurements with local resampling dynamics. We thus examine a deeper interplay between metastability and quantum measurements by external observers, and suggest a mechanism through which individual wells may spontaneously emerge. In particular, we explore an algorithmic interpretation of the Markov property when experimental measurements are the noise channels.

Consider a measurement channel with Kraus operators
\begin{align}
   \CM[\vsigma] =  \sum_b \vM_b\vsigma \vM^{\dagger}_b\quad \text{where}\quad  \sum_b \vM_b^{\dagger}\vM_b = \vI,
\end{align}
and the classical outcome $b$ is recorded by the measurement apparatus, e.g., the single-qubit computational basis measurement $\vM_0 = \ket{0}\bra{0}$ and $\vM_1 = \ket{1}\bra{1}$ with $b\in \{0,1\}$. After the measurement, the state collapses, according to the measurement outcomes
\begin{align}
     \vsigma_b := \frac{\vM_b\vsigma \vM^{\dagger}_b}{\tr[\vM_b\vsigma \vM^{\dagger}_b]} \quad \text{with probability}\quad p_b := \tr[\vM_b\vsigma \vM^{\dagger}_b].
\end{align}

In this context, the Markov property says that, the post-measurement state, on \textit{average}, $\CM[\vsigma] = \sum_{b}p_b \vsigma_b,$ can be recovered by running a local resampling dynamics. However, \textit{conditioned on observing outcome $b$}, what should we get from the resampling dynamics $\CR[\vsigma_b]$? A priori, the Markov property only constrains the combined statistics, and says little about the individual post-selected events $\CR[\vsigma_b].$ 

To understand post-selection in the context of Markov properties, consider the classical 2D Ising model in the ferromagnetic phase as an analogy. Partitioning the state space into the ``mostly up'' and ``mostly down'' sectors, the Gibbs state contains their uniform mixture. However, conditioned on our observation of which sector it lies in, Glauber dynamics gets stuck in the same sector for exponentially long times before converging back to the Gibbs state. Put differently, if a metastable state happens to behave like one of these ``wells,'' it would make an even stronger form of metastability, that is stable not only under thermalization dynamics, but also under active readout (e.g., using a microscope to probe a material in a fridge). To explore this mechanism, we can define stronger versions of Markov properties as follows.
\begin{defn}[Strongly Markov states]\label{defn:strong_Markov}
Given a measurement channel $\{\vM_b\}$ we say a state is $\epsilon$-strongly Markov, if there is a recovery map $\CR$ such that
\begin{align}
    \sum_b \lnorm{\CR[\vM_b\vsigma \vM_b^{\dagger}] - p_b \vsigma}_1 \le \epsilon. 
\end{align}    

\end{defn}
In words, a state is strongly Markov when most of the measurement outcomes can be \textit{individually} recovered to the same state. Consequently, by repeatedly measuring the state and recovering, one can obtain (nearly) statistically independent measurements on the same metastable state. Here, the way we put $\epsilon$ intends to tolerate low-probability events that may kick the state ``out of the well.'' By convexity, an $\epsilon$-strongly Markov state is automatically $\epsilon$-Markov, but the converse need not be true in general.

Experimentally, if we find a strongly Markov state (one of the ``wells'') for various sets $\mathsf{A}$ under the local resampling dynamics $\CR_{\mathsf{A},t},$ we expect a fast accumulation of consistent statistics by repeatedly measuring and recovering. After extracting all observables of interest, one may wish to explore other ``wells,'' by either re-initializing or making global updates. A priori, there might be plenty (like a spin glass) or very few (like a 2D Ising model) strongly Markov states, and there should not be a universal guarantee on efficiently enumerating all such states. Even though we are not able to upgrade our Markov property to this stronger form, we believe this formally captures a precise mechanism for the emergence of macroscopic stationary thermal phenomena, which, despite its quantum mechanical nature, remains stable even under active monitoring of external observers.

\section{Discussion}
\label{section:discussion}

In this work, we seek universal structures in stationary low-temperature phenomena that are nonetheless far from true thermal equilibrium. In the spirit of Glauber dynamics, our mathematical starting point is to model the system-bath dynamics by a (KMS)-detailed-balanced, (quasi-)local Lindbladian, and study its approximate stationary states as metastable states. Our main results have proved that the celebrated area law of mutual information for quantum Gibbs states and the Markov property of classical Gibbs distributions both approximately generalize to metastable states, for region sizes that scale with the metastability. As an attempt to initiate the study for a precise theory of metastability, we have introduced notions such as local minima of free energy, a non-commutative Fisher information, and approximate detailed balance conditions.\\ 

\paragraph{Other notions of metastability.} Our rigorous results are tailored to metastability defined by the Lindbladians of \cite{chen2023efficient}, and are likely generalizable to other KMS-detailed balanced families, constructed algorithmically~\cite{DLL25,gilyen2024quantum} or derived physically~\cite{SA24}. At moderate regimes of $\beta$ (e.g., $\sim \log(n)$), we believe our free-energy centric discussion captures some aspects of local minima of energy, defined for open systems~\cite{chen2024local} or closed systems~\cite{yin2025theory}. However, our structural results may not apply to non-detailed balanced master equations or at low temperature limits~\cite{chen2024local}, as we fundamentally exploit the KMS-detailed balance of the Lindbladian dynamics at non-vanishing temperatures. These allow us to exploit the stronger notion of stationarity in statistical distances, which implies the more general form of metastability for stationarity of energy~\cite{chen2024local}. It would also be interesting to compare our open system framework with closed system notions of metastability~\cite{yin2025theory}.

A common alternative \textit{static} formulation of metastability is based on the bottlenecks between well-defined subsets of the state space. Classically, this formulation is equivalent to approximate stationarity: metastable distributions are always mixtures of ``extremal'' Gibbs distributions, i.e., restrictions of the Gibbs distribution to a subset of the state space, where transitions out (through the bottleneck) are suppressed (c.f. \cite{BH16}). Recently, \cite{gamarnik2024slowmixingquantumgibbs, rakovszky2024bottlenecksquantumchannelsfinite} introduced quantum versions of bottleneck theorems --that assume projections of the Gibbs state onto a \textit{subspace} play the role of the extremal Gibbs state-- which imply approximate stationarity. However, the converse for non-commuting Hamiltonians remains unclear, as to extracting a suitable notion of projection, or any kind of decomposition into suitably defined ``wells.'' \\

\paragraph{Other area laws.} In the study of gapped ground states, the area law has also been a central theme \cite{hastingsAreaLaw, Arad_2012, arad2013arealawsubexponentialalgorithm, Brand_o_2014, anshu2022area} to analytic and numerical methods, e.g., tensor networks \cite{FNW92, stlund1995, Vidal_2004, VGC04, S2005, S2011}. Even though we may physically expect Gibbs states at low enough temperatures to capture gapped ground states, the search for an area law beyond 1D remains active, in contrast to the general connectivity to which thermal area laws apply. We also note that several refinements to the original thermal area law have been made for specific systems \cite{SDHS16, KAA21, LS23, DKS25}.

In the setting of Lindbladian dynamics, several works have studied the structure of \textit{exact} steady states, e.g., \cite{kastoryano2013quantum, BCLSP15, mahajan2016entanglementstructurenonequilibriumsteady, Firanko_2024}. Assuming a \textit{gapped}, detailed-balanced Lindbladian dynamics, they show an area law in the steady state. While stronger in some aspects, their starting assumption is that the Lindbladian is fast-mixing, which is only expected to hold at high temperatures, e.g. \cite{kastoryano2016commuting, capel2021modified, rouze2024efficient, rouze2024optimal, capel2024quasi, kochanowski2024rapid}, or in specific quantum systems at low temperatures \cite{Bardet2023Rapid, Bardet2024, BCL24, DLLZ24}. \\

\paragraph{Faster mixing within a well.} The proof of our local Markov property is based on an a priori guarantee on the runtime of a local mixing process, which scales exponentially with the region size. Intuitively, however, if the obstacle to mixing mainly lies in the bottlenecks between the energy wells, then the system should mix rapidly within the wells. This behavior, often referred to as “rapid mixing within a phase” (or, from specific initializations), has seen notable recent progress in the context of the Ising and Potts models \cite{GS23Ising, GS23RC, GSS25}, as well as in self-correcting quantum memories \cite{BGL25}; see also \cite{koehler2024efficientlylearningsamplingmultimodal, huang2024weakpoincareinequalitiessimulated}. A related question is that of Lifshitz law, which poses that the mixing time for the Ising model in an $m\times m$ box, under $+1$ boundary conditions, should scale as $m^2$ \cite{M94, MT10, LMST13}. If a specific system satisfies some form of the above mechanisms, one expects a much stronger metastability area law. \\

\textbf{Acknowledgments.} We thank Ehud Altman, Jo\~ao Basso, Shankar Balasubramanian, Fernando Brandao, Steven Flammia, Daniel Stilck França, András Gilyén, Jeongwan Haah, Vedika Khemani, Robbie King, David (Ting-Chun) Lin, Lin Lin, Tony Metger, Joel Moore, Nick O'Dea, Michael Ragone, Shengqi Sang, Nikhil Srivastava, and Alexander Zlokapa for helpful discussions. We are especially grateful to Sidhanth Mohanty, Amit Rajaraman, and David X Wu for answering numerous questions on their work \cite{liu2024locally} and to Yunchao Liu, Reza Gheissari, Cambyse Rouze, Anurag Anshu, Quynh T. Nguyen, Robert (Hsin-Yuan) Huang, John Preskill, and Leo Zhou for collaboration on related works \cite{BGL25,chen2025GibbsMarkov,chen2025learning, chen2024local}. Diagram TikZ code was generated with the assistance of ChatGPT-4o. This work was completed in part over the ``Summer Cluster on Quantum Computing'' (2025) at the Simons Institute for the Theory of Computing. CFC is supported by a Simons-CIQC postdoctoral fellowship through NSF QLCI Grant No. 2016245, and thanks MIT for hosting various visits. TB and UV were supported by NSF Grant CCF-231173, NSF QLCI Grant 2016245 and DOE grant DE-SC0024124.\\

\textbf{Contribution statement.} The questions about area laws and stationary states were initiated by C.F.C. and U.V., and the questions concerning Markov properties in metastable states arose between C.F.C. and T.B. C.F.C. introduced the initial analytic framework, conceived and proved the Fisher Information, approximate detailed balance conditions, and the local-minima of free-energy argument, and guided the collaboration with T.B. T.B. formally proved the local Markov property for metastable states and substantially polished the arguments. All authors contributed to writing and revising the manuscript.

\section{Methods}
\label{section:methods}
First, we relate our two main results, area laws and Markov properties, through a reduction to the latter (Section~\ref{sec:arealaw_from_Markov}). The main technical content of this work involves developing a framework to bridge the starting metastability condition to the local Markov property, as summarized in~\autoref{fig:metastab_equiv_diag}. To clarify the challenges and outline our method, we explain how metastable states in \textit{classical spin systems} admit a local Markov property (Section~\ref{section:warmup}) and introduce some of the classical analogs to the concepts in~\autoref{fig:metastab_equiv_diag}. Since our framework will heavily exploit KMS-detailed balance, in Section~\ref{section:the_lindbladian} we recap the Lindbladian family of \cite{chen2023efficient} and highlight the key properties.

The rest of the Methods section explains how we derive structural consequences of metastability in each arrow in~\autoref{fig:metastab_equiv_diag}. A key conceptual step is to introduce the entropy production rate as a quantitative measure of metastability, which unveils a natural characterization of metastable states in terms of a non-commutative Fisher information (Section~\ref{section:ep_fi}). This characterization gives us handles on the static structure of metastable states, particularly the approximate detailed balance condition (Section~\ref{section:adb_overview}), as a stepping stone towards the advertised local Markov property (Section~\ref{section:adb_markov_overview}).

\begin{figure}[ht]
\centering
\begin{tikzpicture}[
  box/.style={rounded corners, minimum width=1.5cm, minimum height=1.7cm, text width=3.5cm, align=center, font=\small, draw=black, fill=gray!2},
  arrow/.style={-{Latex}, thick},
  node distance=1cm and 1.8cm,
  box2/.style={rounded corners, minimum width=1cm, minimum height=1.4cm, text width=3.0cm, align=center, font=\small, draw=black, fill=blue!10},
  box3/.style={rounded corners, minimum width=1.5cm, minimum height=1.4cm, text width=2.5cm, align=center, font=\small, draw=black, fill=gray!2},
  box4/.style={rounded corners, minimum width=1.5cm, minimum height=1.7cm, text width=3.0cm, align=center, font=\small, draw=black, fill=gray!2},
    boundingbox/.style={draw=black, thick, rounded corners=5pt, inner sep=10pt, fit=#1}
]

\node[box3] (T1) {
  \textbf{Metastability}\\[3pt]
  $\mathcal{L}[\vsigma] \approx 0$
};

\node[box4, below=of T1] (T2) {
  \textbf{Local Entropy Production}\\[3pt]
  $\mathsf{EP}_{a}[\vsigma] \approx 0$
};

\node[box, right=of T2] (T3) {
  \textbf{Approximate \\ Detailed Balance}\\[3pt]
  $\sqrt{\vsigma} \vA^a \approx \vrho^{\frac{1}{2}} \vA^a \vrho^{-\frac{1}{2}}\sqrt{\vsigma}$
};

\node[box4, left=of T2] (T4) {
  \textbf{Local\\Metastability}\\[3pt]
  $\mathcal{L}_a[\vsigma] \approx 0$
};

\node[box2, below=of T3] (T5) {
  \textbf{Locally Markov}\\[3pt]
  $\mathcal{R}[\vsigma_{\bar{\mathsf{A}}}\otimes \vsigma_{\mathsf{A}}]\approx \vsigma$
};

\node[box2, below=of T2] (T6) {
  \textbf{Area Law}\\[3pt]
  $\mathsf{I}(\mathsf{A}:\bar{\mathsf{A}})\propto \beta |\partial \mathsf{A}|$
};

\node[boundingbox=(T1)(T2)(T3)(T4)(T5)(T6)] {};

\draw[arrow] (T1) -- (T2) node[midway, right] {\footnotesize }; 
\draw[arrow] (T3) -- (T5) node[midway, right] {\footnotesize Thm.~\ref{thm:local_recovery}};
\draw[arrow] (T5) -- (T6) node[midway, above] {\footnotesize Thm.~\ref{thm:arealaw_metastable}};

\draw[arrow] (T2) -- (T3) node[midway, above] {\footnotesize Thm.~\ref{thm:meta_implie_ADB}};

\draw[arrow] (T4) -- (T2) node[midway, above] {\footnotesize };



\path (T3.west) ++(0,-0.55) coordinate (T3lowwest);
\path (T4.east) ++(0,-0.55) coordinate (T4loweast);

\draw[arrow] (T3lowwest) to [out=-160, in=-20, looseness=1] (T4loweast);

\draw[arrow, dashed] 
    (T3lowwest) to [out=-160, in=-20, looseness=1]
    node[below=8pt, pos=.9] {\footnotesize Thm.~\ref{thm:adb_to_meta}} 
  (T4loweast);

\end{tikzpicture}

\caption{ \label{fig:metastab_equiv_diag}
The equivalent notions of metastability studied in this paper and their consequences. Our starting assumption captures stationarity in statistical distance, for arbitrary observables; the entropy production rate captures the stationarity of the free energy as a particular observable, for which we give an explicit Fisher information. This, in turn, reveals a more workable form of local static equilibrium, approximate detail balance. 
}
\end{figure}
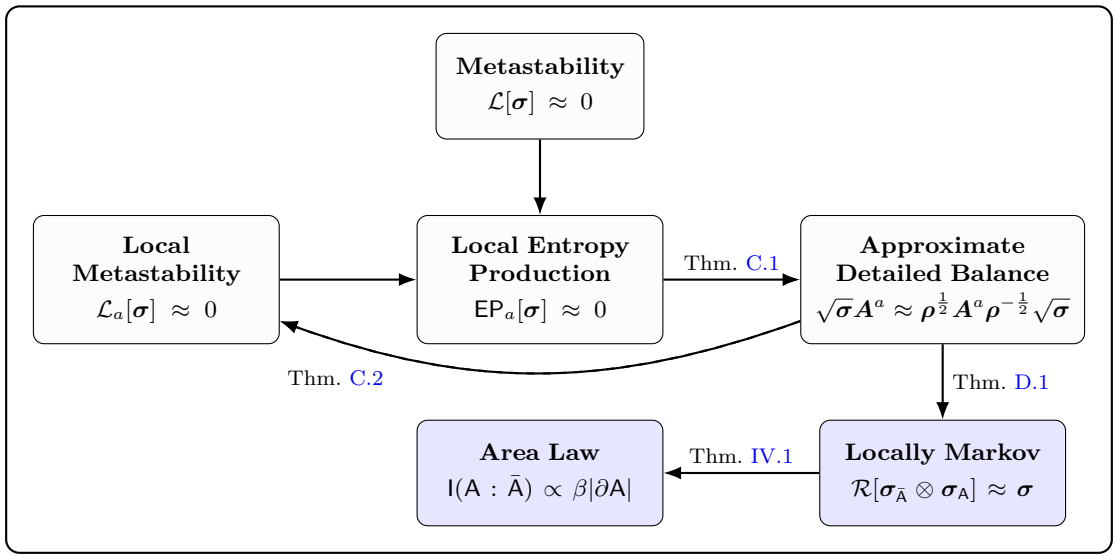

\subsection{Area Law from Local Recovery} \label{sec:arealaw_from_Markov}
The key idea behind the original thermal area law of \cite{wolf2008area} is a variational characterization of quantum Gibbs states in terms of the free energy. 
The first step towards our metastability area law is the observation that \cite{wolf2008area} remains valid for local minima of the free energy, which can be precisely captured by the local Markov property. 

Recall, the original thermal area law of \cite{wolf2008area} states that for any local Hamiltonian $\vH$ on $n$ qubits, the mutual information of its Gibbs state $\vrho$ across any spatial bipartition scales with the size of the boundary. Underlying its proof is essentially a trade-off between entropy and energy, captured by the free energy~\eqref{eq:Fsigma_Fmarginal}. Indeed, by carefully collecting terms, the mutual information across a region $\mathsf{A}\subset [n]$ of qubits in any state $\vsigma$, can be expressed in terms of the following free energy difference:
\begin{align}
    \beta F(\vsigma) - \beta F(\vsigma_{\bar{\mathsf{A}}} \otimes \vsigma_{\mathsf{A}})  = \undersetbrace{\text{mutual information}}{\mathsf{I}(\mathsf{A}:\bar{\mathsf{A}})_{\vsigma}} + \beta \cdot  \undersetbrace{\text{boundary energy}}{\tr[(\partial\vH) (\vsigma-\vsigma_{\bar{\mathsf{A}}} \otimes \vsigma_{\mathsf{A}})]}\label{eq:energy_entropy_free_energy}.
\end{align}

\noindent Where $\partial\vH$ are all the Hamiltonian terms simultaneously supported on qubits of $\mathsf{A}$ and $\bar{\mathsf{A}}$. The free energy is also intimately related to the \textit{quantum relative entropy} with respect to the Gibbs state. 
\begin{equation}
         \beta F(\vsigma)  = D(\vsigma||\vrho)  +\beta F(\vrho), \quad \text{where}\quad D(\vsigma_1||\vsigma_2) := \tr[\vsigma_1(\log \vsigma_1-\log\vsigma_2)].\label{eq:F_to_D}
   \end{equation}
\noindent This information-theoretical view of free energy plays an essential role throughout the paper. Since the relative entropy is strictly positive $D(\vsigma||\vrho) \ge 0$ if and only if the states differ, the Gibbs state is the \textit{unique global minimum} of the free energy. This already reproduces the thermal area law~\cite{wolf2008area}:
\begin{equation}
    \mathsf{I}(\mathsf{A}:\bar{\mathsf{A}})_{\vrho} \le \beta \cdot \big|\tr[(\partial\vH) (\vrho-\vrho_{A}\otimes \vrho_{\bar{\mathsf{A}}})]\big| \leq 2 \beta \cdot \|\partial \vH\|,
\end{equation}

\noindent where we simply minimized the free energy under the Gibbs state $\vsigma = \vrho$ in $\eqref{eq:energy_entropy_free_energy} \le 0$, and applied Holder's inequality. The main claim of this section is a generalization of the above argument for local minima of free energy, captured in terms of a detailed-balanced recovery map.

\begin{thm}[Local Recovery implies an Area Law]\label{thm:arealaw_metastable} Suppose there exists a recovery map $\CR$ which recovers a state $\vsigma$ from its bipartite marginals $\vsigma_{\mathsf{\bar{\mathsf{A}}}}\otimes \vsigma_\mathsf{A}$ and preserves the Gibbs state\footnote{It in fact suffices to recover $\vsigma$ from the entanglement breaking map $\vsigma_{\mathsf{\bar{\mathsf{A}}}}\otimes \vsigma_\mathsf{A}$, to prove the area law. However, the noise-recovery lends better to the perspective of $\vsigma$ as a local minima, and is ultimately what we prove in \autoref{thm:main_meta_implies_markov}. }
    \begin{align}
    \norm{\CR[\vsigma_{\mathsf{\bar{\mathsf{A}}}}\otimes \vsigma_\mathsf{A} ] - \vsigma}_1 \le \epsilon_{\mathsf{Markov}} \quad \text{and}\quad \CR[\vrho] = \vrho.\label{eq:recovery_error_overview_area_law}
    \end{align}
    Then, $\vsigma$ satisfies an approximate area law of mutual information across the region $\mathsf{A}$:
    \begin{align}
       \mathsf{I}(\mathsf{A}:\bar{\mathsf{A}})_{\vsigma} \le 2\beta \norm{\partial\vH} + 4 \epsilon_{\mathsf{Markov}} \cdot \max \bigg( \log\frac{1}{\epsilon_{\mathsf{Markov}}} , \beta \|\vH\| , n\bigg).
    \end{align}
\end{thm}

Here, we only require detailed balance of the approximate recovery map $\CR$, without requiring any locality -- despite later on obtaining such a guarantee when we explicitly construct $\CR$. The interpretation of recoverability in terms of local minima of free energy is rooted in the second law of thermodynamics. For any recovery map $\CR$ that preserves the Gibbs state, the free energy monotonically decreases:
\begin{align}
    \beta F(\CR[\vsigma]) - \beta F(\vsigma) = D(\CR[\vsigma]||\vrho) -D(\vsigma||\vrho) \le 0 \quad \text{for \textit{every}}\quad \vsigma\quad \text{assuming}\quad \CR[\vrho]=\vrho, 
\end{align}
which is a consequence of the Data Processing Inequality for the relative entropy (e.g.,~\cite{wilde2011classical}). Thus, if a generic state $\vsigma$ can be (approximately) recovered from erasures, it must have lower free energy: 
\begin{align}
      F(\vsigma) \approx F(\CR[\vsigma_{\bar{\mathsf{A}}}\otimes \vsigma_\mathsf{A} ]) \le F(\vsigma_{\bar{\mathsf{A}}}\otimes \vsigma_\mathsf{A}). \label{eq:free_energy_local_minima}
\end{align}

\noindent The proof of \autoref{thm:arealaw_metastable} then immediately follows by converting the recovery error \eqref{eq:recovery_error_overview_area_law}, to an error in relative entropy, by continuity of the relative entropy. Subsequently, we simply continue the proof of the original thermal area law, resulting in $\eqref{eq:energy_entropy_free_energy} \le 0 + (\text{error})$.

\subsection{Warm-up: Local Recovery in Metastable Spin Systems}
\label{section:warmup}

To grasp the challenges in proving a local recovery property for metastable states, we begin with a classical spin system $x\in \{0, 1\}^n$. There, the counterpart of Lindblad dynamics is Glauber dynamics, a local, detailed-balanced continuous-time Markov chain with generator $\mathsf{L} = \sum_{a\in [n]}\mathsf{L}_a$ with stationary distribution the Gibbs measure $\pi$. Let $\nu$ be a distribution on $\{0, 1\}^n$, and assume $\nu$ is metastable under each local generator $\mathsf{L}_a$:
\begin{equation}
    \forall a\in [n], \quad \mathsf{L}_a[\nu]\approx 0.
\end{equation}
\noindent Given a sample from the metastable distribution $x\sim \nu$, suppose we replace a subset $\mathsf{A}\subset [n]$ of the spins of the sample with an arbitrary configuration $\tau\in \{0, 1\}^{|\mathsf{A}|}$. Subsequently, we run the induced Markov chain $\mathsf{L}_{\mathsf{A}} = \sum_{a\in \mathsf{A}}\mathsf{L}_{a}$ restricted to the patch of spins we deleted. Does this process recover $\nu$?
\begin{equation}
    e^{t\mathsf{L}_{\mathsf{A}}}[\nu_{\bar{\mathsf{A}}}\otimes \tau_{\mathsf{A}}] \stackrel{?}{\approx} \nu.
\end{equation}

\noindent This is plausible since the generator $\mathsf{L}_{\mathsf{A}}$ induces a local mixing process, which eventually resamples the interior of $\mathsf{A}$, according to the conditional Gibbs measure $\pi_\mathsf{A}(\cdot|x_{\partial\mathsf{A}})$ of the spins currently at the boundary $\partial\mathsf{A}$ -- where the boundary $x_{\partial\mathsf{A}}$ is fixed throughout the local resampling dynamics. Since $\nu$ is approximately stationary under all such local resampling steps, we expect it to be approximately reconstructed:
\begin{align}\label{eq:two_timescales}
   \big\|\nu - e^{t\mathsf{L}_{\mathsf{A}}}[\nu_{\bar{\mathsf{A}}}\otimes \tau_{\mathsf{A}}]\big\|_1 &\leq  \under{\big\|\nu - e^{t\mathsf{L}_{\mathsf{A}}}[\nu]\big\|_1}{\text{Leakage}}  \quad +\quad   \under{\big\|e^{t\mathsf{L}_{\mathsf{A}}}[\nu] - e^{t\mathsf{L}_{\mathsf{A}}}[\nu_{\bar{\mathsf{A}}}\otimes \tau_{\mathsf{A}}]\big\|_1}{\text{Local Mixing}}. 
\end{align}

\noindent In general, provided the target region $\mathsf{A}$ is not too large (for local mixing to occur) and the state is sufficiently metastable (for minimal leakage), there should be a good time $t$ such that the metastable state is recovered.

Although this broadly outlines our argument in the quantum setting as well, generalizing the ``local-mixing'' part poses significant challenges due to uniquely quantum phenomena. The classical argument hinges on the fact that the restricted dynamics, conditioned on the boundary, eventually resample from the conditional Gibbs distribution without disturbing the boundary. In contrast, the quantum Lindbladian dynamics is \textit{non-local} and can disturb faraway qubits, making it difficult to meaningfully define a ``fixed boundary condition''. A priori, it is not even obvious why $\e^{t\CL }[\vsigma]$ should ever converge to the same state as $\e^{t\CL }[  \vsigma_{\mathsf{\bar{\mathsf{A}}}}\otimes \vsigma_\mathsf{A}],$ if the resampling dynamics were to back-react on the boundary.

To proceed, we need to further generalize other static structures of metastable states to the quantum realm, most notably a local, approximate, detailed balance condition. Indeed -- although it is implicit in \eqref{eq:two_timescales} -- a consequence of classical metastability is that the likelihood ratios in $\nu$ are approximately the same as in $\pi$, over adjacent states in the state space:
\begin{equation}
    \under{\mathsf{L}[\nu]\approx 0}{\text{Metastability}} \iff
 \under{\frac{\nu(x)}{\nu(y)} \approx \frac{\pi(x)}{\pi(y)},}{\text{Approx. Detailed Balance}} \text{ for random} \quad  x\sim \nu, \quad \text{and local update} \quad x\rightarrow y.\label{eq:classical_ADB}
\end{equation}

 This picture gives an alternative route to prove the classical Markov property of metastable states: approximate detailed balance implies that the conditional metastable measure on the small region $\mathsf{A}$, is close to that of the conditional Gibbs measure. Thus, the local dynamics that conditionally resample the Gibbs measure, also reconstruct the metastable state. As we discuss, our proof in the quantum setting requires first explicitly identifying a suitable non-commutative generalization of approximate detailed balance.

To conclude this section, we briefly comment that at first glance, the likelihood \textit{ratios} may seem rather abstract. It turns out that a natural intermediate object to consider is the \textit{log likelihood ratio}, before and after the local update. Indeed, these objects naturally appear in the time derivative of the classical relative entropy, a quantity known as the \textit{Fisher information.} This subject of classical optimal transport \cite{CNR24} theory is our main inspiration for a quantum theory of metastability. Next, we lay out an arsenal of structural characterizations of quantum metastability, building towards the proof of the local Markov property (\autoref{thm:main_meta_implies_markov}).

\subsection{The KMS Detailed-Balanced Lindbladian family of \cite{chen2023efficient}}
\label{section:the_lindbladian}

Our study of metastability depends on the model of system-bath interaction -- the KMS-detailed balanced Lindbladian family of \cite{chen2023efficient} -- and its explicit form will be exploited throughout the paper. Fix a Hamiltonian $\vH$ on $n$ qubits, an inverse temperature $\beta>0$, and a single self-adjoint jump $\vA^a = \vA^{a\dagger}$. We consider the (quasi-local) Lindbladian defined by: 
\begin{align}
		\CL_a[\cdot] = \underset{\text{``coherent''}}{\underbrace{ -\ri [\vC^a, \cdot]}} + 
		\int_{-\infty}^{\infty} \gamma(\omega) \bigg(\underset{\text{``transition''}}{\underbrace{\hat{\vA}^a(\omega)(\cdot)\hat{\vA}^{a}(\omega)^\dagg}} - \underset{\text{``decay''}}{\underbrace{\frac{1}{2}\{\hat{\vA}^{a}(\omega)^\dagg\hat{\vA}^a(\omega),\cdot\}}}\bigg)\rd\omega\label{eq:exact_DB_L}
\end{align}
\noindent with the shifted-Metropolis weight
\begin{align}
    \gamma(\omega) = \exp\L(-\beta\max\left(\omega +\frac{\beta \sigma^2}{2},0\right)\R).\label{eq:Metropolis}
\end{align} 
\noindent The central ingredient in the Lindbladian above is the operator Fourier transform \cite{chen2023quantum,chen2023efficient}. The operator FT of an operator $\vA$, associated to the Hamiltonian $\vH$, can be written as: 
\begin{align}\label{eq:OFT}
{\hat{\vA}}_\sigma(\omega)=  \frac{1}{\sqrt{2\pi}}\int_{-\infty}^{\infty} \e^{\ri \vH t} \vA \e^{-\ri \vH t} \e^{-\ri \omega t} f_{\sigma}(t)\rd t\quad \text{with}\quad f_{\sigma}(t) := e^{-\sigma^2t^2}\sqrt{\sigma\sqrt{2/\pi}}
    \end{align}

\noindent where the function $f_\sigma(t)$ above is a Gaussian filter of energy \textit{width} $ \sigma: 1/\beta \ge \sigma>0$.\footnote{Whenever implicit, we omit the subscripts ${\hat{\vA}_{\sigma}}(\omega)\equiv {\hat{\vA}}(\omega)$, $f(t)=f_\sigma(t)$.} This particular $f_\sigma(t)$ will routinely appear throughout calculations. See Section~\ref{sec:OFT_appendix} for a recap of useful properties. The ``coherent part'' $\vC^a$ is a Hermitian operator (see \autoref{thm:ckg_db}):
\begin{align}
    &\vC^a := \iint_{-\infty}^{\infty} \gamma(\omega) c(t) \cdot \hat{\vA}^a(\omega,t)^{\dagger}\hat{\vA}^a(\omega,t)  \rd t\rd \omega, \quad  \\\text{with}\quad  &c(t) := \frac{1}{\beta\sinh(2\pi t/\beta)} \quad \text{and}\quad  \hat{\vA}^a(\omega,t):=\e^{i\vH t}\hat{\vA}^a(\omega)\e^{-i\vH t}.
\end{align}
The defining feature of the above Lindbladian is achieving detailed balance while preserving locality. Indeed, the time-evolved operators respect the spatial locality of the Hamiltonian (up to some exponentially decaying tail  with distance) 
due to Lieb-Robinson bounds~\cite{Lieb1972,hastings2006spectral,chen2023speed}. Meanwhile, the Lindbladian satisfies exact KMS-detailed balance, which plays a subtle yet crucial role in enabling exact mathematical identities.

\begin{thm}[KMS Detailed Balance of $\CL$ {\cite{chen2023efficient}}]\label{thm:ckg_db}
The Lindbladian $\CL_a$ defined in \eqref{eq:exact_DB_L} satisfies KMS-$\vrho$-detailed-balance, in that it is self-adjoint with respect to the KMS-inner-product. 
\begin{equation}
    \braket{\vX, \CL_a^\dagger[\vY]}_{\vrho} =  \braket{\CL_a^\dagger[\vX], \vY}_{\vrho} \quad \text{for}\quad \vrho \propto \e^{-\beta \vH}.
\end{equation}
\noindent and hence fixes the Gibbs state exactly: $\CL_a[\vrho] =0$
\end{thm}
In the above, the KMS-inner-product is a non-commutative expectation defined by:
\begin{align}
    \langle \vX,\vY\rangle_{\vrho}:=\tr[\vX^\dagger \vrho^{\frac{1}{2}}\vY\vrho^{\frac{1}{2}}]\,\quad \text{ and } \quad \norm{\vX}_{\vrho} := \sqrt{\langle \vX,\vX\rangle_{\vrho}}.
\end{align}
In this paper, we focus on this particular class of Lindbladians due to the notably clean properties of operator Fourier transforms that define their Lindbladian operators. However, we believe our arguments can be formally extended to any family of KMS-detailed balance Lindbladians with reasonable choices of weights and filters, such as \cite{DLL25, SA24}. 

\subsection{Entropy Production and the Fisher Information}
\label{section:ep_fi}

Keeping track of the general Lindbladian dynamics $\vsigma_t = e^{t\CL}[\vsigma]$ for an arbitrary initial state may appear challenging. As hinted, a natural objective function to consider is the free energy. Indeed, by the relation with relative entropy~\eqref{eq:F_to_D} and the data-processing inequality, the Free Energy of $\vsigma_t$ monotonically decreases during the evolution
\begin{equation}
    \frac{d}{dt}F(\vsigma_t) = \frac{d}{dt}D(\vsigma_t||\vrho) \leq 0.
\end{equation}
A metastable state should, then, also be a state with vanishingly small \textit{entropy production rate}\footnote{Up to suitable regularization to ensure $\vsigma$ is full rank, see Section~\ref{section:smoothing}.}. Classically, the study of such an entropic quantity (known as the Kullback–Leibler divergence) has played a central role in the study of classical Markov chains. In the quantum case, this is a more unwieldy quantity -- even the commuting case is still under active study (e.g.,~\cite{capel2021modified})--but will be the key unlocking our structural results.
\begin{defn}
    [The Entropy Production Rate]\label{def:EP} Let $\mathcal{L}$ be a Lindbladian with fixed point $\vrho$ such that $\CL[\vrho]=0$. The {entropy production rate} of a state $\vsigma$ w.r.t $\CL$ is given by
    \begin{equation}
      \mathsf{EP}  [\vsigma]:= - \frac{d}{dt}D\big(e^{t\mathcal{L}}[\vsigma]||\vrho\big)\bigg|_{t=0} = - \tr[\mathcal{L}[\vsigma](\log \vsigma-\log \vrho)] \geq 0.
    \end{equation}
\end{defn}

\noindent When $\CL$ is a sum over individual, (quasi)-local Lindbladians $\CL = \sum_{a\in \CA}\CL_a$ each of which fixes $\CL_a[\vrho] = 0$, then linearity tells us the \textit{global} entropy production is a linear combination of the \textit{local} entropy production rates, all of which are non-negative
\begin{equation}
    \mathsf{EP}  [\vsigma] = \sum_{a\in \CA} \mathsf{EP}_{a} [\vsigma]\quad \text{where} \quad \mathsf{EP}_{a} [\vsigma]\geq 0\quad \text{for each}\quad a\in \CA.
\end{equation}
To make use of $\mathsf{EP}_a$, one of our technical contributions is to equate the entropy production to non-commutative \textit{spatial gradients} of the log-likelihood, which will later play a central role in our understanding of thermal metastability. 
\begin{thm}[Entropy Production and the Fisher Information, \autoref{thm:integral_square_logs} (Informal)] For any state $\vsigma$ and jump operator $\vA^a$, its local entropy production rate w.r.t the $\vrho$-KMS-detailed balanced Lindbladian $\mathcal{L}_a$ \eqref{eq:lindblad_def} can be written as
    \begin{equation}\label{eq:fisher_nc_intro}
 \mathsf{EP}_a[\vsigma] = \mathsf{FI}_a[\vsigma||\vrho]:= \big\| \nabla_a[\log \vsigma - \log \vrho]\big\|^2_{\mathsf{FI}, \vsigma},
\end{equation}
 for a suitable definition of gradient and weighted norm.
\end{thm}
Analogous ``entropy dissipation'' or de Bruijn identities have been established for both discrete-state classical Markov chains and quantum dynamical semigroups satisfying GNS-detailed-balance, most notably the Davies generator (c.f. \cite{carlen2024dynamics}), where the entropy production is also expressed in terms of a relative Fisher information functional. Since the Davies generator is only considered physical for commuting or few-body Hamiltonians, our extension of this theory -- borrowing tools from~\cite{carlen2024dynamics, chen2025learning} -- to the KMS-detailed balanced Lindbladians of \cite{chen2023efficient} is a first step towards non-commuting many-body Hamiltonians. 

As we explain in Section~\ref{section:EP_GF}, the notation $\|\cdot\|_{\mathsf{FI}, \vsigma}$ indicates an appropriately weighted norm capturing a non-commutative analog of the Fisher Information (\autoref{defn:fisher}). In turn, the $\nabla_a[\cdot]$ operator denotes a non-commutative discrete gradient constructed from commutators with the jump operators:
\begin{equation}\label{eq:nc_grad}
 \nabla_{a, \omega, t}[\cdot ]  =  \big[\vA^a(\omega, t), \ \cdot \ \big] \quad \text{and}\quad   \nabla_a[\cdot ]  = \iint_{-\infty}^{\infty} \nabla_{a, \omega, t}[\cdot ] \otimes \ket{\omega, t} \rd \omega \rd t.
\end{equation}
The locality of $\hat{\vA}^a(\omega, t),$ inherited from the Lindbladian and the Hamiltonian, is an essential feature to be exploited throughout. 

In the sequence, the identity~\eqref{eq:fisher_nc_intro} (see \autoref{thm:integral_square_logs}) will provide a powerful tool to understand states that are metastable. Indeed, \autoref{thm:integral_square_logs} already hints that metastable states behave ``locally'' like Gibbs states:
\begin{equation}\label{eq:intro_locally_gibbs}
  \mathsf{EP}_a[\vsigma]  \approx 0 \iff \nabla_a[\log \vsigma]\approx \nabla_a[\log \vrho] = -\beta\cdot \nabla_a[\vH],
\end{equation}

\noindent in that commutators with the ``effective Hamiltonian'' $\log \vsigma$ are the same as that with $\vH$, when weighted by the metastable state. While not at all obvious from this notation, the attentive reader may note that \eqref{eq:intro_locally_gibbs}, when $\vsigma, \vH,$ and $\vA^a$ are all classical, recovers the classical approximate detailed balance condition captured in \eqref{eq:classical_ADB}. Unfortunately, in the non-commuting setting, the ``effective Hamiltonian'' $\log \vsigma$ is more intricate, and we will need to further relate \eqref{eq:fisher_nc_intro} to a more tangible operational characterization of metastability. 

\subsection{Metastable States are Approximately Detailed Balanced}
\label{section:adb_overview}
An intermediate step towards proving the local Markov property is understanding when and how the local resampling dynamics do not drift towards another metastable state. This requires us to characterize the extent to which metastable states are in local equilibrium, by further massaging the Fisher information into a static, workable notion of \textit{approximate detailed balance}.

In particular, we prove that a state $\vsigma$ is metastable under the Lindbladian, if and only if it also obeys a certain approximate detailed balance of transition amplitudes for jump $\vA$, giving a \textit{static} characterization of metastability without explicitly referring to Lindbladians: 
\begin{align}
    \sqrt{\vsigma} \vA \approx \vrho^{\frac{1}{2}} \vA \vrho^{-\frac{1}{2}}\sqrt{\vsigma}.
\end{align}
Indeed, the Gibbs state satisfies this exactly $\sqrt{\vrho} \vA = \vrho^{\frac{1}{2}} \vA \vrho^{-\frac{1}{2}}\cdot \sqrt{\vrho}$, and the claim is that, we can likewise pass certain local operators $\vA$ around metastable state $\vsigma$, as long as we compensate with suitable powers of the Gibbs state $\vrho.$ Intriguingly, we had to study approximate detailed balance for square roots of the metastable states -- hinting at the underlying KMS-inner product -- which will eventually be exploited in the proof of the Markov property. 

Naively, the imaginary time conjugation $\vrho^{\frac{1}{2}} \vA \vrho^{-\frac{1}{2}}$ can be highly divergent, and the central technical challenge in implementing this approach is identifying the appropriate, regularized notion of approximate detailed balance that can be derived directly from approximate stationarity. 

\begin{defn}
    [The Approximate Detailed Balance Condition]\label{defn:intro_adb}
    We say a state $\vsigma$ is approximately detailed-balanced under $\vA^a$ with error
    \begin{equation}
        \mathsf{ADB}_a[\vsigma] := \iint_{-\infty}^{\infty} \lnormp{\vA^a(\omega,t)\sqrt{\vsigma}-\sqrt{\vsigma}\vrho^{-\frac{1}{2}}\vA^a(\omega,t)\vrho^{\frac{1}{2}}}{2}^2  \gamma(\omega) g(t) \rd \omega \rd t,
    \end{equation}
    with $\gamma(\omega)$ the shifted Metropolis weight~\eqref{eq:Metropolis}, and $g(t)$ is the filter function in the Dirichlet form (\autoref{lem:Dicirchlet}).
\end{defn}

The static detailed-balance of the metastable states can also be expressed in terms of non-commutative log-likelihood ratios, as a generalization of the classical case (see, e.g., \cite[Lemma 3.4]{liu2024locally}): there is a suitable weighted norm $\norm{\cdot}_{\mathsf{ADB}, \vsigma}$ (\autoref{lem:ADB_int_log}) such that 
    \begin{align}
     \mathsf{ADB}_a[\vsigma] = \big\| \nabla_a[\log \vsigma - \log \vrho]\big\|^2_{\mathsf{ADB}, \vsigma}.
\end{align}

Now that the entropy gradients are exposed, we can compare the above weighted norms to that of the entropy production. Put informally, in \autoref{thm:meta_implie_ADB} we prove that states $\vsigma$ with bounded Fisher information are also approximately detailed-balanced in the sense of \autoref{defn:intro_adb}:

\begin{thm}
    [$\mathsf{EP}$ implies $\mathsf{ADB}$, \autoref{thm:meta_implie_ADB} (Informal)] For any state $\vsigma$ and jump operator $\vA^a$,
    \begin{align}
\mathsf{ADB}_a[\vsigma] &= \big\| \nabla_a[\log \vsigma - \log \vrho]\big\|^2_{\mathsf{ADB}, \vsigma} \leq  \widetilde{\CO}\bigg(\big\| \nabla_a[\log \vsigma - \log \vrho]\big\|^2_{\mathsf{FI}, \vsigma}\bigg) = \widetilde{\CO}\big(\mathsf{FI}_a[\vsigma||\vrho]\big) .
\end{align}
\end{thm}
The $\widetilde{\CO}(\cdot)$ notation absorbs logarithmic factors in its arguments, see \autoref{thm:meta_implie_ADB}. Independent of the proof of Markov property, in the remainder of Section \ref{section:adb} we further expand the theory of metastability as in~\autoref{fig:metastab_equiv_diag}. We complete the arrows of equivalence for $\mathsf{ADB}$: if $\vsigma$ is approximately detailed-balanced (as in~\autoref{defn:intro_adb}) under a local transition operator $\vA^a$, then it is also locally metastable $\CL_a[\vsigma]\approx 0$. 

\subsection{Approximate Detailed Balance implies a Local Markov Property}
\label{section:adb_markov_overview}

To conclude the proof of the area law, we now generalize the local Markov property for Gibbs states~\cite{chen2025GibbsMarkov} to metastable states by the aforementioned approximate detailed balance condition. Recall, \cite{chen2025GibbsMarkov} showed that local erasure to the Gibbs state can be recovered by a quasi-local Lindbladian dynamics. The main subject of~\cite{chen2025GibbsMarkov} is the \textit{local time-averaging map}: Given a region $\mathsf{A}\subset [n]$ and a tunable time parameter $t\ge0$, consider the completely positive and trace-preserving (CPTP) map
\begin{align}
    \CR_{\mathsf{A},t}[\cdot] &:= \frac{1}{t}\int_{0}^t \exp\L(s\,\mathcal{L}_\mathsf{A}\R)[\cdot] \,\rd s \\
    \mathcal{L}_\mathsf{A} &:=\sum_{a\in P^1_\mathsf{A}}\CL_a,\quad\text{where}\quad P_\mathsf{A}^1 :=\{\vX_i,\vY_i,\vZ_i\}_{i\in \mathsf{A}},  \label{eq:RAt_methods}
\end{align}

\noindent which evolves the system under the KMS detailed-balanced family of Lindbladians \eqref{eq:exact_DB_L} with single-qubit Pauli jumps restricted to the local region $\mathsf{A}$, for a uniformly random time $s\in [0,t]$. 

\begin{thm}[Quantum Gibbs States can be Locally Recovered~{\cite{chen2025GibbsMarkov}}]\label{thm:cr_local_recovery} For any region $\mathsf{A}$, the time-averaged local Lindblad dynamics $\CR_{\mathsf{A},t}$ \eqref{eq:RAt_methods} recovers the Gibbs state $\vrho$ from any noise channel\footnote{In~\cite{chen2025GibbsMarkov}, the result was proven for the particular noise being the replacement with maximally mixed state. But the result seamlessly extends to general noise channels, through~\autoref{lem:decompose_replacement}.} $\CN_{\mathsf{A}}$, 
\begin{align}
      \norm{\vrho-\CR_{\mathsf{A},t} [\CN_{\mathsf{A}}[\vrho]]}_1 \le  \e^{\mu|\mathsf{A}|}\,\cdot t^{-\lambda}\label{eq:mu_lambda},
\end{align}
for some $\mu >0$ and $0<\lambda<1$ depending only on the inverse temperature $\beta$ and the interaction degree $d$ of the Hamiltonian.
\end{thm}
 Since the required timescale for recovery scales only with the region size --and is independent of the full system-size-- the time-averaging map $\CR_{\mathsf{A},t}$ remains quasi-local. This implies Gibbs states must have a particular tripartite correlation structure, known as the decay of conditional mutual information. 

The generalization to metastable states is made transparent by separating the two sources of error, as we did in the classical example~\eqref{eq:two_timescales}
\begin{align}\label{eq:quantum_two_time_scales}
    \lnorm{\vsigma-\CR_{\mathsf{A},t}[\CN_{\mathsf{A}}[\vsigma]]}_1 \le \undersetbrace{\text{Leakage}}{\norm{\vsigma- \CR_{\mathsf{A},t}[\vsigma]}_1} \quad +\quad  \undersetbrace{\text{Local Mixing}}{\big\|\CR_{\mathsf{A},t}[\vsigma-\CN_{\mathsf{A}}[\vsigma]]\big\|_1 }.
\end{align}
Since the genuine Gibbs state is exactly stationary, the leakage term vanishes. For metastable states, the leakage term is simply bounded by $t \cdot \norm{\CL_{\mathsf{A}}[\vsigma]}_1,$ and the main task reduces to extending local mixing to metastable states. The main conceptual challenge is that the Lindbladian dynamics, unlike the classical case, can actually disturb the boundary. Thus, on top of the original arguments for Gibbs states, in the non-commutative setting $\mathsf{ADB}$ will play an additional role of limiting disturbances to the boundary.

\begin{thm}[Approximate Detailed Balance implies Local Mixing]\label{thm:local_recovery_intro} In the context of \autoref{thm:cr_local_recovery}, suppose a state $\vsigma$ satisfies approximate detailed balance for all single-site Pauli operators on $\mathsf{A}$, with uniform error $\epsilon_{\mathsf{ADB}} = \max_{\vP\in P^1_\mathsf{A}} \mathsf{ADB}_{\vP}[\vsigma]$. Then, the time-averaged Lindblad dynamics $\CR_{\mathsf{A},t}$ \eqref{eq:RAt_methods} recovers $\vsigma$ from erasures on $\mathsf{A}$:
\begin{align}
    \norm{\CR_{\mathsf{A},t}[\vsigma-\CN_{\mathsf{A}}[\vsigma] ]}_1 \leq \e^{\mu|\mathsf{A}|}\,\cdot \big(t^{-1}+\sqrt{\epsilon_{\mathsf{ADB}}}\big)^{\lambda} 
\end{align}
\noindent for some $\mu >0$ and $0<\lambda<1$ depending only on $\beta,d.$
\end{thm}

When $\epsilon_{\mathsf{ADB}} = 0$, we recover the existing result for Gibbs states. In general, however, the quality of approximate detailed balance limits the achievable recovery guarantee. Unlike in the classical case~\eqref{eq:two_timescales}, this is a manifestation of a back-reaction to the boundary, which may not ever be cured at infinite time. We defer to Section~\ref{sec:proving_Markov} an extensive overview of \cite{chen2025GibbsMarkov}'s argument, and a discussion of the role of $\mathsf{ADB}$ in generalizing their result. In this section, we focus on highlighting one of the central roles of $\mathsf{ADB}$ in the proof. 

The approach of~\cite{chen2025GibbsMarkov} starts by studying local mixing in the \textit{Heisenberg picture}: ultimately to prove \autoref{thm:cr_local_recovery}, it suffices to show that for every $n$-qubit observable $\vO$, 
\begin{align}
\CR_{\mathsf{A},t}^{\dagger}[\vO] \quad \text{is nearly identity on}\quad \mathsf{A}\quad \text{when weighted by}\quad \vrho.\label{eq:locally_trivial}
\end{align}
The challenge is that the test operator $\vO$ can be a generic global operator, and unlike in the classical case we cannot hope to utilize an apriori estimate on the local mixing time.\footnote{Such as a local gap or log-Sobolev inequality for the local generator $\mathcal{L}_\mathsf{A}$; this is since $\mathcal{L}_\mathsf{A}$ is not strictly local.} Instead, \cite{chen2025GibbsMarkov} makes use of a much weaker notion of convergence: an algebraic decay of the Dirichlet form known as \textit{local stationarity}. 
\begin{equation}
  \big|\braket{\CR_{\mathsf{A},t}^{\dagger}[\vO], \CL_\mathsf{A}^{\dagger}[\CR_{\mathsf{A},t}^{\dagger}[\vO]]}_{\vrho}\big| \le \frac{2}{t} \|\vO\|^2\quad \text{for each}\quad \vO
\end{equation}
\noindent which, in a concrete sense says that running the (local) time-averaging map $\CR_{\mathsf{A}, t}$ should enforce (local) metastability $\CL_\mathsf{A}^{\dagger}[\CR_{\mathsf{A},t}^{\dagger}[\vO]]\approx 0$. The key idea behind~\eqref{eq:locally_trivial} (for Gibbs states) is an explicit link between the dynamics and the statics, put forth precisely by the Dirichlet form:
\begin{align}
    \undersetbrace{\text{Time-derivative}}{ -\frac{\rd }{\rd t} \norm{\vX}_{\vrho}^2} = -\braket{\vX,\CL_a^{\dagger}[\vX]}_{\vrho} = \undersetbrace{\text{spatial-derivative}}{\| \nabla_a[\vX]\|_{\mathsf{D}, \vrho}^2}\label{eq:static_dynamic}
\end{align}
where $\nabla_a$ is exactly the same as~\eqref{eq:nc_grad}, and $\|\cdot\|_{\mathsf{D}, \vrho}$ is another weighted inner product  (\autoref{lem:Dicirchlet}); later, we will pick $\vX = \CR_{\mathsf{A},t}^{\dagger}[\vO]$. As a consequence, the local metastability resulting from local time-averaging must also imply suppression of spatial derivatives --which are commutators with local jump operators-- ultimately ensuring~\eqref{eq:locally_trivial}.

The starting point of~\autoref{thm:local_recovery_intro} is to reproduce this connection between statics and dynamics as in \eqref{eq:static_dynamic}, but instead weighted under the metastable state $\vsigma$. Although the equalities in \eqref{eq:static_dynamic} are generically false under an arbitrary $\vsigma$, we show that if $\vsigma$ is metastable, then approximate detailed balance is precisely enough to relate the two:
\begin{align}
    {-\braket{\vX,\CL_a^{\dagger}[\vX]}_{\vsigma}} \approx {\| \nabla_a[\vX]\|_{\mathsf{D}, \vsigma}^2} \quad \text{if}\quad \mathsf{ADB}_a[\vsigma] \approx 0. \label{eq:aux_dirichlet_intro}
\end{align}
In fact, the choice of weights in the $\mathsf{ADB}$ condition~\autoref{defn:intro_adb} is mainly defined for this purpose. 

Ultimately, the remaining challenge lies in relating the decay of the $\vsigma$-weighted KMS norm of commutators which arise in \eqref{eq:aux_dirichlet_intro}, back to the guarantees of the recovery map. We refer to Section~\ref{section:outline_adb_markov} for a comprehensive outline of the role of approximate detailed balance in relating the two.

\bibliographystyle{alphaUrlePrint.bst}
\bibliography{ref}

\newcommand{\etalchar}[1]{$^{#1}$}
\begin{thebibliography}{HMRW24}

\bibitem[AAG22]{anshu2022area}
Anurag Anshu, Itai Arad, and David Gosset.
\newblock An area law for 2d frustration-free spin systems.
\newblock In {\em Proceedings of the 54th Annual ACM SIGACT Symposium on Theory of Computing}, pages 12--18, 2022.

\bibitem[AHHH08]{Alicki2008OnTS}
Robert Alicki, Michał Horodecki, Paweł Horodecki, and Ryszard Horodecki.
\newblock \href{https://api.semanticscholar.org/CorpusID:26719502}{On thermal stability of topological qubit in kitaev's 4d model}.
\newblock {\em Open Syst. Inf. Dyn.}, 17:1--20, 2008.

\bibitem[AKLV13]{arad2013arealawsubexponentialalgorithm}
Itai Arad, Alexei Kitaev, Zeph Landau, and Umesh Vazirani.
\newblock An area law and sub-exponential algorithm for 1d systems, 2013, arXiv: {{\ttfamily 1301.1162}}.

\bibitem[ALV12]{Arad_2012}
Itai Arad, Zeph Landau, and Umesh Vazirani.
\newblock \href{http://dx.doi.org/10.1103/physrevb.85.195145}{Improved one-dimensional area law for frustration-free systems}.
\newblock {\em Physical Review B}, 85(19), May 2012.

\bibitem[BCE{\etalchar{+}}22]{Balasubramanian2022TowardsAT}
Krishnakumar Balasubramanian, Sinho Chewi, Murat~A. Erdogdu, Adil Salim, and Matthew~Shunshi Zhang.
\newblock \href{https://api.semanticscholar.org/CorpusID:246706055}{Towards a theory of non-log-concave sampling: First-order stationarity guarantees for langevin monte carlo}.
\newblock In {\em Annual Conference Computational Learning Theory}, 2022.

\bibitem[BCG{\etalchar{+}}23]{Bardet2023Rapid}
Ivan Bardet, \'Angela Capel, Li~Gao, Angelo Lucia, David P\'erez-Garc\'{\i}a, and Cambyse Rouz\'e.
\newblock \href{http://dx.doi.org/10.1103/PhysRevLett.130.060401}{Rapid thermalization of spin chain commuting hamiltonians}.
\newblock {\em Phys. Rev. Lett.}, 130:060401, 2 2023.

\bibitem[BCG{\etalchar{+}}24]{Bardet2024}
Ivan Bardet, Angela Capel, Li~Gao, Angelo Lucia, David Pérez-García, and Cambyse Rouzé.
\newblock \href{http://dx.doi.org/10.1007/s00220-023-04869-5}{Entropy decay for {D}avies semigroups of a one dimensional quantum lattice}.
\newblock {\em Communications in Mathematical Physics}, 405(2), February 2024.

\bibitem[BCL{\etalchar{+}}15]{BCLSP15}
Fernando G. S.~L. Brandão, Toby~S. Cubitt, Angelo Lucia, Spyridon Michalakis, and David Perez-Garcia.
\newblock \href{http://dx.doi.org/10.1063/1.4932612}{Area law for fixed points of rapidly mixing dissipative quantum systems}.
\newblock {\em Journal of Mathematical Physics}, 56(10):102202, 10 2015.

\bibitem[BCL24]{BCL24}
Thiago Bergamaschi, Chi-Fang Chen, and Yunchao Liu.
\newblock \href{http://dx.doi.org/10.1109/FOCS61266.2024.00071}{Quantum computational advantage with constant-temperature gibbs sampling}.
\newblock In {\em 2024 IEEE 65th Annual Symposium on Foundations of Computer Science (FOCS)}, pages 1063--1085, 2024.

\bibitem[BCS57]{BCS57}
J.~Bardeen, L.~N. Cooper, and J.~R. Schrieffer.
\newblock \href{http://dx.doi.org/10.1103/PhysRev.108.1175}{Theory of superconductivity}.
\newblock {\em Phys. Rev.}, 108:1175--1204, Dec 1957.

\bibitem[BEGK00]{Bovier2000MetastabilityAL}
Anton Bovier, Michael~A. Eckhoff, V{\'e}ronique Gayrard, and Markus Klein.
\newblock \href{https://api.semanticscholar.org/CorpusID:5769372}{Metastability and low lying spectra¶in reversible markov chains}.
\newblock {\em Communications in Mathematical Physics}, 228:219--255, 2000.

\bibitem[BEGK04]{Bovier2004MetastabilityIR}
Anton Bovier, Michael~A. Eckhoff, V{\'e}ronique Gayrard, and Markus Klein.
\newblock \href{https://api.semanticscholar.org/CorpusID:6552552}{Metastability in reversible diffusion processes i: Sharp asymptotics for capacities and exit times}.
\newblock {\em Journal of the European Mathematical Society}, 6:399--424, 2004.

\bibitem[BGK05]{Bovier2005MetastabilityIR}
Anton Bovier, V{\'e}ronique Gayrard, and Markus Klein.
\newblock \href{https://api.semanticscholar.org/CorpusID:14729515}{Metastability in reversible diffusion processes ii. precise asymptotics for small eigenvalues}.
\newblock {\em Journal of the European Mathematical Society}, 7:69--99, 2005.

\bibitem[BGL25]{BGL25}
Thiago Bergamaschi, Reza Gheissari, and Yunchao Liu.
\newblock Rapid mixing for gibbs states within a logical sector: a dynamical view of self-correcting quantum memories, 2025, arXiv: {{\ttfamily 2507.10976}}.

\bibitem[BH14]{Brand_o_2014}
Fernando G. S.~L. Brandão and Michał Horodecki.
\newblock \href{http://dx.doi.org/10.1007/s00220-014-2213-8}{Exponential decay of correlations implies area law}.
\newblock {\em Communications in Mathematical Physics}, 333(2):761–798, November 2014.

\bibitem[BH16]{BH16}
Anton Bovier and Frank~Den Hollander.
\newblock \href{http://dx.doi.org/10.4171/022-3/26}{{\em Metastability: A Potential-Theoretic Approach}}, volume 351.
\newblock Springer, 2016.

\bibitem[BM80]{AJBray_1980}
A~J Bray and M~A Moore.
\newblock \href{http://dx.doi.org/10.1088/0022-3719/13/24/005}{Replica theory of quantum spin glasses}.
\newblock {\em Journal of Physics C: Solid State Physics}, 13(24):L655, aug 1980.

\bibitem[BSS24]{bosboom2024unique}
Vincent Bosboom, Matthias Schlottbom, and Felix~L Schwenninger.
\newblock On the unique solvability of radiative transfer equations with polarization.
\newblock {\em Journal of Differential Equations}, 393:174--203, 2024.

\bibitem[CAN25]{chen2025learning}
Chi-Fang Chen, Anurag Anshu, and Quynh~T Nguyen.
\newblock Learning quantum gibbs states locally and efficiently.
\newblock {\em arXiv preprint arXiv:2504.02706}, 2025.

\bibitem[Car24]{carlen2024dynamics}
Eric Carlen.
\newblock Dynamics and quantum optimal transport: Three lectures on quantum entropy and quantum markov semigroups.
\newblock In {\em Optimal Transport on Quantum Structures}, pages 29--89. Springer, 2024.

\bibitem[CGKR24]{capel2024quasi}
{\'A}ngela Capel, Paul Gondolf, Jan Kochanowski, and Cambyse Rouz{\'e}.
\newblock Quasi-optimal sampling from gibbs states via non-commutative optimal transport metrics.
\newblock {\em arXiv preprint arXiv:2412.01732}, 2024.

\bibitem[CHPZ24]{chen2024local}
Chi-Fang Chen, Hsin-Yuan Huang, John Preskill, and Leo Zhou.
\newblock Local minima in quantum systems.
\newblock In {\em Proceedings of the 56th Annual ACM Symposium on Theory of Computing}, pages 1323--1330, 2024.

\bibitem[CKBG23]{chen2023quantum}
Chi-Fang Chen, Michael~J Kastoryano, Fernando~GSL Brand{\~a}o, and Andr{\'a}s Gily{\'e}n.
\newblock Quantum thermal state preparation.
\newblock {\em arXiv preprint arXiv:2303.18224}, 2023.

\bibitem[CKG23]{chen2023efficient}
Chi-Fang Chen, Michael~J Kastoryano, and Andr{\'a}s Gily{\'e}n.
\newblock An efficient and exact noncommutative quantum gibbs sampler.
\newblock {\em arXiv preprint arXiv:2311.09207}, 2023.

\bibitem[Cli90]{clifford1990markov}
Peter Clifford.
\newblock Markov random fields in statistics.
\newblock {\em Disorder in physical systems: A volume in honour of John M. Hammersley}, pages 19--32, 1990.

\bibitem[CLY23]{chen2023speed}
Chi-Fang~Anthony Chen, Andrew Lucas, and Chao Yin.
\newblock Speed limits and locality in many-body quantum dynamics.
\newblock {\em Reports on Progress in Physics}, 86(11):116001, 2023.

\bibitem[CM12]{carlen2012analog2wassersteinmetricnoncommutative}
Eric~A. Carlen and Jan Maas.
\newblock An analog of the 2-wasserstein metric in non-commutative probability under which the fermionic fokker-planck equation is gradient flow for the entropy, 2012, arXiv: {{\ttfamily 1203.5377}}.

\bibitem[CM17]{carlen2017gradientflowentropyinequalities}
Eric~A. Carlen and Jan Maas.
\newblock Gradient flow and entropy inequalities for quantum markov semigroups with detailed balance, 2017, arXiv: {{\ttfamily 1609.01254}}.

\bibitem[CM19]{Carlen_2019}
Eric~A. Carlen and Jan Maas.
\newblock \href{http://dx.doi.org/10.1007/s10955-019-02434-w}{Non-commutative calculus, optimal transport and functional inequalities in dissipative quantum systems}.
\newblock {\em Journal of Statistical Physics}, 178(2):319–378, November 2019.

\bibitem[CNWR24]{CNR24}
Sinho Chewi, Jonathan Niles-Weed, and Philippe Rigollet.
\newblock Statistical optimal transport, 2024, arXiv: {{\ttfamily 2407.18163}}.

\bibitem[CR25]{chen2025GibbsMarkov}
Chi-Fang Chen and Cambyse Rouz{\'e}.
\newblock Quantum gibbs states are locally markovian.
\newblock {\em arXiv preprint arXiv:2504.02208}, 2025.

\bibitem[CRSF21]{capel2021modified}
\'{A}ngela Capel, Cambyse Rouz\'{e}, and Daniel Stilck~França.
\newblock The modified logarithmic {S}obolev inequality for quantum spin systems: classical and commuting nearest neighbour interactions, 2021.
\newblock arxiv{2009.11817}.

\bibitem[CST{\etalchar{+}}21]{childs2021theory}
Andrew~M Childs, Yuan Su, Minh~C Tran, Nathan Wiebe, and Shuchen Zhu.
\newblock Theory of trotter error with commutator scaling.
\newblock {\em Physical Review X}, 11(1):011020, 2021.

\bibitem[DLL25]{DLL25}
Zhiyan Ding, Bowen Li, and Lin Lin.
\newblock \href{http://dx.doi.org/10.1007/s00220-025-05235-3}{Efficient quantum gibbs samplers with kubo-martin-schwinger detailed balance condition}.
\newblock {\em Communications in Mathematical Physics}, 406(3), 2025.

\bibitem[DLLZ24]{DLLZ24}
Zhiyan Ding, Bowen Li, Lin Lin, and Ruizhe Zhang.
\newblock Polynomial-time preparation of low-temperature gibbs states for 2d toric code, 2024, arXiv: {{\ttfamily 2410.01206}}.

\bibitem[DR17]{Datta2017RelatingRE}
Nilanjana Datta and Cambyse Rouz'e.
\newblock \href{https://api.semanticscholar.org/CorpusID:119669102}{Relating relative entropy, optimal transport and fisher information: A quantum hwi inequality}.
\newblock {\em Annales Henri Poincar{\'e}}, 21:2115 -- 2150, 2017.

\bibitem[Fey82]{feynman1982SimQPhysWithComputers}
Richard~P. Feynman.
\newblock \href{http://dx.doi.org/10.1007/BF02650179}{Simulating physics with computers}.
\newblock {\em International Journal of Theoretical Physics}, 21(6-7):467--488, 1982.

\bibitem[Fey86]{Feynman1986}
R.~P. Feynman.
\newblock \href{http://dx.doi.org/10.1007/BF01886518}{Quantum mechanical computers}.
\newblock {\em Foundations of Physics}, 16(6):507--531, 1986.

\bibitem[FGA24]{Firanko_2024}
Raz Firanko, Moshe Goldstein, and Itai Arad.
\newblock \href{http://dx.doi.org/10.1063/5.0167353}{Area law for steady states of detailed-balance local lindbladians}.
\newblock {\em Journal of Mathematical Physics}, 65(5), May 2024.

\bibitem[FNW92]{FNW92}
M.~Fannes, Bruno Nachtergaele, and Reinhard Werner.
\newblock \href{http://dx.doi.org/10.1007/BF02099178}{Finitely correlated states on quantum spin chains}.
\newblock {\em Communications in Mathematical Physics}, 144, 03 1992.

\bibitem[GCDK24]{gilyen2024quantum}
Andr{\'a}s Gily{\'e}n, Chi-Fang Chen, Joao~F Doriguello, and Michael~J Kastoryano.
\newblock Quantum generalizations of glauber and metropolis dynamics.
\newblock {\em arXiv preprint arXiv:2405.20322}, 2024.

\bibitem[GKZ24]{gamarnik2024slowmixingquantumgibbs}
David Gamarnik, Bobak~T. Kiani, and Alexander Zlokapa.
\newblock Slow mixing of quantum gibbs samplers, 2024, arXiv: {{\ttfamily 2411.04300}}.

\bibitem[Gla63]{glauber}
Roy~J. Glauber.
\newblock \href{http://dx.doi.org/10.1063/1.1703954}{Time‐dependent statistics of the ising model}.
\newblock {\em Journal of Mathematical Physics}, 4(2):294--307, 02 1963.

\bibitem[GR14]{gradshteyn2014table}
I.~S. Gradshteyn and I.~M. Ryzhik.
\newblock {\em Table of Integrals, Series, and Products}.
\newblock Academic Press, 8th edition, 2014.
\newblock Edited by Daniel Zwillinger and Victor Moll.

\bibitem[GR21]{gao2021riccicurvaturequantumchannels}
Li~Gao and Cambyse Rouzé.
\newblock Ricci curvature of quantum channels on non-commutative transportation metric spaces, 2021, arXiv: {{\ttfamily 2108.10609}}.

\bibitem[GS23a]{GS23Ising}
Reza Gheissari and Alistair Sinclair.
\newblock \href{http://dx.doi.org/10.1214/22-AAP1911}{{Low-temperature Ising dynamics with random initializations}}.
\newblock {\em The Annals of Applied Probability}, 33(5):3916 -- 3957, 2023.

\bibitem[GS23b]{GS23RC}
Reza Gheissari and Alistair Sinclair.
\newblock Spatial mixing and the random-cluster dynamics on lattices, 2023, arXiv: {{\ttfamily 2207.11195}}.

\bibitem[GSLW19]{gilyen2018QSingValTransf}
Andr\'{a}s Gily\'{e}n, Yuan Su, Guang~Hao Low, and Nathan Wiebe.
\newblock \href{http://dx.doi.org/10.1145/3313276.3316366}{Quantum singular value transformation and beyond: exponential improvements for quantum matrix arithmetics}.
\newblock In {\em Proceedings of the 51st Annual ACM SIGACT Symposium on Theory of Computing}, STOC 2019, page 193–204, New York, NY, USA, 2019. Association for Computing Machinery.

\bibitem[GSS25]{GSS25}
Reza Gheissari, Allan Sly, and Youngtak Sohn.
\newblock Rapid phase ordering for ising and potts dynamics on random regular graphs, 2025, arXiv: {{\ttfamily 2505.15783}}.

\bibitem[Has07]{hastingsAreaLaw}
M~B Hastings.
\newblock \href{http://dx.doi.org/10.1088/1742-5468/2007/08/P08024}{An area law for one-dimensional quantum systems}.
\newblock {\em Journal of Statistical Mechanics: Theory and Experiment}, 2007(08):P08024, aug 2007.

\bibitem[HK06]{hastings2006spectral}
Matthew~B Hastings and Tohru Koma.
\newblock Spectral gap and exponential decay of correlations.
\newblock {\em Communications in mathematical physics}, 265:781--804, 2006.

\bibitem[HMRW24]{huang2024weakpoincareinequalitiessimulated}
Brice Huang, Sidhanth Mohanty, Amit Rajaraman, and David~X. Wu.
\newblock Weak poincar\'e inequalities, simulated annealing, and sampling from spherical spin glasses, 2024, arXiv: {{\ttfamily 2411.09075}}.

\bibitem[Isi25]{Ising:1925em}
Ernst Ising.
\newblock \href{http://dx.doi.org/10.1007/BF02980577}{{Contribution to the Theory of Ferromagnetism}}.
\newblock {\em Z. Phys.}, 31:253--258, 1925.

\bibitem[KAA21]{KAA21}
Tomotaka Kuwahara, Álvaro~M. Alhambra, and Anurag Anshu.
\newblock \href{http://dx.doi.org/10.1103/physrevx.11.011047}{Improved thermal area law and quasilinear time algorithm for quantum gibbs states}.
\newblock {\em Physical Review X}, 11(1), March 2021.

\bibitem[KACR24]{kochanowski2024rapid}
Jan Kochanowski, Alvaro~M Alhambra, Angela Capel, and Cambyse Rouz{\'e}.
\newblock Rapid thermalization of dissipative many-body dynamics of commuting hamiltonians.
\newblock {\em arXiv preprint arXiv:2404.16780}, 2024.

\bibitem[KB16]{kastoryano2016commuting}
Michael~J Kastoryano and Fernando~GSL Brandao.
\newblock Quantum gibbs samplers: The commuting case.
\newblock {\em Communications in Mathematical Physics}, 344:915--957, 2016.

\bibitem[KKS25]{DKS25}
Donghoon Kim, Tomotaka Kuwahara, and Keiji Saito.
\newblock \href{http://dx.doi.org/10.1103/PhysRevLett.134.020402}{Thermal area law in long-range interacting systems}.
\newblock {\em Phys. Rev. Lett.}, 134:020402, Jan 2025.

\bibitem[KLV24]{koehler2024efficientlylearningsamplingmultimodal}
Frederic Koehler, Holden Lee, and Thuy-Duong Vuong.
\newblock Efficiently learning and sampling multimodal distributions with data-based initialization, 2024, arXiv: {{\ttfamily 2411.09117}}.

\bibitem[KSVV02]{kitaev2002classical}
Alexei~Yu Kitaev, Alexander Shen, Mikhail~N Vyalyi, and Mikhail~N Vyalyi.
\newblock {\em Classical and quantum computation}.
\newblock Graduate studies in mathematics. American Mathematical Society, 2002.

\bibitem[KT73]{KT73}
J.~M. Kosterlitz and D.~J. Thouless.
\newblock \href{http://dx.doi.org/10.1088/0022-3719/6/7/010}{Ordering, metastability and phase transitions in two-dimensional systems}.
\newblock {\em Journal of Physics C: Solid State Physics}, 6(7):1181, apr 1973.

\bibitem[KT13]{kastoryano2013quantum}
Michael~J. Kastoryano and Kristan Temme.
\newblock Quantum logarithmic {S}obolev inequalities and rapid mixing.
\newblock {\em Journal of Mathematical Physics}, 54(5):052202, 2013.

\bibitem[Kuw24]{kuwahara2024clustering}
Tomotaka Kuwahara.
\newblock Clustering of conditional mutual information and quantum markov structure at arbitrary temperatures.
\newblock {\em arXiv preprint arXiv:2407.05835}, 2024.

\bibitem[LC17]{low2016HamSimQSignProc}
Guang~Hao Low and Isaac~L. Chuang.
\newblock \href{http://dx.doi.org/10.1103/physrevlett.118.010501}{Optimal hamiltonian simulation by quantum signal processing}.
\newblock {\em Physical Review Letters}, 118(1), January 2017.

\bibitem[LMR{\etalchar{+}}24]{liu2024locally}
Kuikui Liu, Sidhanth Mohanty, Prasad Raghavendra, Amit Rajaraman, and David~X Wu.
\newblock Locally stationary distributions: A framework for analyzing slow-mixing markov chains.
\newblock In {\em 2024 IEEE 65th Annual Symposium on Foundations of Computer Science (FOCS)}, pages 203--215. IEEE, 2024.

\bibitem[LMST13]{LMST13}
Eyal Lubetzky, Fabio Martinelli, Allan Sly, and {Fabio Lucio} Toninelli.
\newblock \href{http://dx.doi.org/10.4171/JEMS/363}{Quasi-polynomial mixing of the 2d stochastic ising model with plus boundary up to criticality}.
\newblock {\em Journal of the European Mathematical Society}, 15(2):339--386, 2013.

\bibitem[LPW{\etalchar{+}}17]{Markovchain_mixing}
David~Asher Levin, Yuval Peres, Elizabeth~L. Wilmer, James Propp, and David~B. Wilson.
\newblock {\em Markov chains and mixing times}.
\newblock American Mathematical Society, 2017.

\bibitem[LR72]{Lieb1972}
Elliott~H. Lieb and Derek~W. Robinson.
\newblock \href{http://dx.doi.org/10.1007/BF01645779}{The finite group velocity of quantum spin systems}.
\newblock {\em Commun. Math. Phys.}, 28(3):251--257, 1972.

\bibitem[LS23]{LS23}
Marius Lemm and Oliver Siebert.
\newblock \href{http://dx.doi.org/10.22331/q-2023-08-16-1083}{Thermal area law for lattice bosons}.
\newblock {\em Quantum}, 7:1083, August 2023.

\bibitem[Mar94]{M94}
Fabio Martinelli.
\newblock \href{http://dx.doi.org/10.1007/BF02187060}{On the two-dimensional dynamical ising model in the phase coexistence region.}
\newblock {\em Journal of Statistical Physics}, 76:1179–1246, 1994.

\bibitem[MF10]{MT10}
Fabio Martinelli and Toninelli {Fabio Lucio}.
\newblock \href{http://dx.doi.org/10.1007/s00220-009-0963-5}{On the mixing time of the 2d stochastic ising model with “plus” boundary conditions at low temperature.}
\newblock {\em Commun. Math. Phys.}, 296:175–213, 2010.

\bibitem[MFM{\etalchar{+}}16]{mahajan2016entanglementstructurenonequilibriumsteady}
Raghu Mahajan, C.~Daniel Freeman, Sam Mumford, Norm Tubman, and Brian Swingle.
\newblock Entanglement structure of non-equilibrium steady states, 2016, arXiv: {{\ttfamily 1608.05074}}.

\bibitem[Ons44]{Onsager44}
Lars Onsager.
\newblock \href{http://dx.doi.org/10.1103/PhysRev.65.117}{Crystal statistics. i. a two-dimensional model with an order-disorder transition}.
\newblock {\em Phys. Rev.}, 65:117--149, Feb 1944.

\bibitem[OR95]{stlund1995}
Stellan \"Ostlund and Stefan Rommer.
\newblock \href{http://dx.doi.org/10.1103/physrevlett.75.3537}{Thermodynamic limit of density matrix renormalization}.
\newblock {\em Physical Review Letters}, 75(19):3537–3540, November 1995.

\bibitem[RFA24]{rouze2024efficient}
Cambyse Rouz{\'e}, Daniel~Stilck Fran{\c{c}}a, and {\'A}lvaro~M Alhambra.
\newblock Efficient thermalization and universal quantum computing with quantum gibbs samplers.
\newblock {\em arXiv preprint arXiv:2403.12691}, 2024.

\bibitem[RPBK24]{rakovszky2024bottlenecksquantumchannelsfinite}
Tibor Rakovszky, Benedikt Placke, Nikolas~P. Breuckmann, and Vedika Khemani.
\newblock Bottlenecks in quantum channels and finite temperature phases of matter, 2024, arXiv: {{\ttfamily 2412.09598}}.

\bibitem[RSFA24]{rouze2024optimal}
Cambyse Rouz{\'e}, Daniel Stilck~Fran{\c{c}}a, and {\'A}lvaro~M Alhambra.
\newblock Optimal quantum algorithm for gibbs state preparation.
\newblock {\em arXiv e-prints}, pages arXiv--2411, 2024.

\bibitem[SA25]{SA24}
Matteo Scandi and {\'A}lvaro~M. Alhambra.
\newblock Thermalization in open many-body systems and kms detailed balance, 2025, arXiv: {{\ttfamily 2505.20064}}.

\bibitem[Sch05]{S2005}
U.~Schollwöck.
\newblock \href{http://dx.doi.org/10.1103/revmodphys.77.259}{The density-matrix renormalization group}.
\newblock {\em Reviews of Modern Physics}, 77(1):259–315, April 2005.

\bibitem[Sch11]{S2011}
Ulrich Schollwöck.
\newblock \href{http://dx.doi.org/10.1016/j.aop.2010.09.012}{The density-matrix renormalization group in the age of matrix product states}.
\newblock {\em Annals of Physics}, 326(1):96–192, January 2011.

\bibitem[SDHS16]{SDHS16}
Nicholas~E. Sherman, Trithep Devakul, Matthew~B. Hastings, and Rajiv R.~P. Singh.
\newblock \href{http://dx.doi.org/10.1103/physreve.93.022128}{Nonzero-temperature entanglement negativity of quantum spin models: Area law, linked cluster expansions, and sudden death}.
\newblock {\em Physical Review E}, 93(2), February 2016.

\bibitem[VGRC04]{VGC04}
F.~Verstraete, J.~J. Garc\'{\i}a-Ripoll, and J.~I. Cirac.
\newblock \href{http://dx.doi.org/10.1103/PhysRevLett.93.207204}{Matrix product density operators: Simulation of finite-temperature and dissipative systems}.
\newblock {\em Phys. Rev. Lett.}, 93:207204, Nov 2004.

\bibitem[Vid04]{Vidal_2004}
Guifré Vidal.
\newblock \href{http://dx.doi.org/10.1103/physrevlett.93.040502}{Efficient simulation of one-dimensional quantum many-body systems}.
\newblock {\em Physical Review Letters}, 93(4), July 2004.

\bibitem[Wil11]{wilde2011classical}
Mark~M Wilde.
\newblock From classical to quantum shannon theory.
\newblock {\em arXiv preprint arXiv:1106.1445}, 2011.

\bibitem[WVHC08]{wolf2008area}
Michael~M Wolf, Frank Verstraete, Matthew~B Hastings, and J~Ignacio Cirac.
\newblock Area laws in quantum systems: mutual information and correlations.
\newblock {\em Physical review letters}, 100(7):070502, 2008.

\bibitem[YSL25]{yin2025theory}
Chao Yin, Federica~M Surace, and Andrew Lucas.
\newblock Theory of metastable states in many-body quantum systems.
\newblock {\em Physical Review X}, 15(1):011064, 2025.

\end{thebibliography}
\newpage
\vspace{-0.2cm}
In the appendix, we fully develop our theory of metastability by introducing the basic notation, preliminaries, and setup (\autoref{sec:prelim}). We then analyze the explicit form of entropy production for the KMS-detailed balance Lindbladians (\autoref{section:EP_GF}). The obtained Fisher information then serves as a basis for deriving more workable notions of local equilibrium, called approximated detailed balance (\autoref{section:adb}). Ultimately, these tools provide sufficient conditions for extending the Markov property~\cite{chen2025GibbsMarkov} for quantum Gibbs states, to metastable states.
\vspace{-0.8cm}
\tableofcontents
\vspace{-0.3cm}
\newpage

\appendix
\section{Preliminaries}
\label{sec:prelim}
\subsection{Notation}\label{sec:recap_notation}
We use extensively
\begin{align}
    a \lesssim b
\end{align}
to mean $a \le c b$ for an absolute constant $c$. We write scalars, functions, and vectors in normal font, matrices in bold font $\vO$, and superoperators in curly font~$\CL$ with matrix arguments in square brackets $\CL[\vrho]$. We use $\CO(\cdot),\Omega (\cdot)$ to denote asymptotic upper and lower bounds.
\begin{align}
\vI&: &\text{the identity operator}\\
\beta&: &\text{ inverse temperature}\\
\vrho_{\beta}&:= \frac{\e^{-\beta \vH }}{\tr[ \e^{-\beta \vH }]} (\equiv \vrho) \quad &\text{the Gibbs state with inverse temperature $\beta$}\\
\vsigma& &\text{the metastable state}\\
\mathsf{A} \subset \Lambda&:& \text{A subset of vertices}\\
\labs{\mathsf{A}}&: & \text{cardinality of the region $A$}\\
n &= \labs{\Lambda} &\text{ system size (number of qubits) of the Hamiltonian $\vH$}\\
\{\vA^a\}_{a\in \CA}&: &\text{set of jumps (for defining the Lindbladian)}\\
P_\mathsf{A}&:& \text{set of nontrivial Pauli string on region $A$}\\
P^1_\mathsf{A}&:=\{\vX_i,\vY_i,\vZ_i\}_{i\in \mathsf{A}}  & \text{set of 1-local Pauli on region $A$}
\end{align}
Fourier transform notations:
\begin{align}
\vH &= \sum_i E_i \ketbra{\psi_i}{\psi_i}&\text{the Hamiltonian of interest and its eigendecomposition}\\
\vP_{E}&:= \sum_{i:E_i = E} \ketbra{\psi_i}{\psi_i}&\text{eigenspace projector for energy $E$}\\
\nu& \in B(\vH)  &\text{the set of Bohr frequencies, i.e., energy differences}\\
\vA_\nu&:=\sum_{E_2 - E_1 = \nu } \vP_{E_2} \vA \vP_{E_1} &\text{amplitude of $\vA$ that changes the energy by exactly $\nu$}\\
{\vA}(t) &:=\e^{i\vH t}\vA \e^{-i\vH t}& \text{Heisenberg-evolved operator $\vA$}\\
\hat{\vA}(\omega) &:= \frac{1}{\sqrt{2\pi}}\int_{-\infty}^{\infty} \e^{-\ri \omega t}f(t) \vA(t)\mathrm{d}t& \text{operator Fourier Transform for $\vA$ weighted by $f$}\\
\hat{f}(\omega)&=\lim_{K\rightarrow  \infty}\frac{1}{\sqrt{2\pi}}\int_{-K}^{K}\e^{-\ri\omega t} f(t)\mathrm{d}t & \text{the Fourier transform of function $f$}		
\end{align}
Norms: 
\begin{align}
	\norm{\vO}&:= \sup_{\ket{\psi},\ket{\phi}} \frac{\bra{\phi} \vO \ket{\psi}}{\norm{\ket{\psi}}\cdot \norm{\ket{\phi}}}= \norm{\vO}_{\infty} \quad &\text{the operator norm of a matrix $\vO$}\\
  	\norm{\vO}_p&:= (\tr \labs{\vO}^p)^{1/p} \quad&\text{the Schatten p-norm of a matrix $\vO$}\\
    \norm{\vO}_{\vrho}&:= \tr[ \vO^{\dagger}\sqrt{\vrho}\vO \sqrt{\vrho} ]^{1/2} \quad&\text{the $\vrho$-KMS-norm of a matrix $\vO$}\\
 \norm{\CL}_{p-p} &:= \sup_{\vO\ne 0} \frac{\normp{\CL[\vO]}{p}}{\normp{\vO}{p}}\quad&\text{the induced $p-p$ norm of a superoperator $\CL$}
\end{align}

\subsection{Hamiltonians with bounded interaction degree}\label{sec:Ham}
On a set $\Lambda$ of $n = \labs{\Lambda}$ qubits, we consider Hamiltonians $\vH$ with few-body terms $\vH_{\gamma}$ 
\begin{align}
    \vH = \sum_{\gamma\in \Gamma} \vH_{\gamma}\quad \text{where}\quad \norm{\vH_{\gamma}} \le 1.
\end{align}
From this decomposition, we define the interaction graph with vertices corresponding to the set $\Gamma$, and we draw an edge between $\gamma_1$ and $\gamma_2$ if and only if the terms have overlapping supports (self-loops allowed):
\begin{align}
\gamma_1\sim \gamma_2\quad \iff\quad    \text{Supp}(\vH_{\gamma_1}) \cap \text{Supp}(\vH_{\gamma_2}) \ne \emptyset.
\end{align}
Similarly, we may consider any subset of vertices $\mathsf{A}\subset \Lambda$ and write 
\begin{align}
    \mathsf{A} \sim \gamma \iff \mathsf{A}\cap \text{Supp}(\vH_{\gamma})\ne \emptyset.
\end{align}
The maximal degree of the interaction graph is denoted by $d$, and we are particularly working in the regime where $d$ is a constant independent of the system size $n$. 

\subsection{Properties of Operator Fourier transforms}
\label{sec:OFT_appendix}

Here, we collect several basic properties of the operator Fourier transforms. Recall the operator FT \eqref{eq:OFT} of an operator $\vA$, associated to the Hamiltonian $\vH$, can be written as: 
\begin{align}
\hat{\vA}_{\sigma}(\omega)=  \frac{1}{\sqrt{2\pi}}\int_{-\infty}^{\infty} \e^{\ri \vH t} \vA \e^{-\ri \vH t} \e^{-\ri \omega t} f_{\sigma}(t)\rd t.
    \end{align}

\noindent The function $f_\sigma(t)$ above is a Gaussian filter, of energy \textit{width} $\sigma>0$:
    \begin{align}
       f_{\sigma}(t) = e^{-\sigma^2t^2}\sqrt{\sigma\sqrt{2/\pi}}\quad \text{and}\quad  \hat{f}_{\sigma}(\omega)=\frac{e^{- \omega^2/4\sigma^2}}{\sqrt{\sigma\sqrt{2\pi}}}  = \frac{1}{\sqrt{2\pi}}\int_{-\infty}^{\infty}\e^{-\ri\omega t} f_\sigma(t)\mathrm{d}t,\label{eq:fwft}
    \end{align}

\noindent with $\hat{f}_\sigma(\omega)$ the FT of $f_\sigma(t)$.  Whenever implicit, we omit the subscripts ${\hat{\vA}_{\sigma}}(\omega)\equiv {\hat{\vA}}(\omega)$, $f(t)=f_\sigma(t)$.

One of the central recurring technical tools we require is the \textit{Bohr frequency decomposition}, i.e. the decomposition of an operator $\vA$ into the Bohr frequencies $\nu\in B(\vH)=\text{Spec}(\vH)-\text{Spec}(\vH)$ (the set of energy differences) associated to the spectral decomposition of $\vH=\sum_{i}E_i\vP_{E_i}$: 
\begin{align}
        \vA = \sum_{\nu\in B(\vH)} \vA_{\nu},\quad \text{where}\quad \vA_\nu&:=\sum_{E_2 - E_1 = \nu } \vP_{E_2} \vA \vP_{E_1} \quad \text{satisfies that} \quad (\vA_{\nu})^{\dagger} = (\vA^{\dagger})_{-\nu},\label{eqn:Aoperator}
\end{align}
where $\vP_{E}$ is the eigenspace projector for an energy $E\in \text{Spec}(\vH)$. In this notation, we can study such a decomposition of the operator FT of some operator $\vA$:
\begin{equation}
    \hat{\vA}(\omega)=\sum_{\nu\in B(\vH)}  \vA_\nu \cdot \hat{f}_\sigma(\omega-\nu),
\end{equation}
\noindent which essentially localizes the operator around the energy band $\omega\pm\sigma$; i.e., it is approximately band diagonal. 

\begin{lem}[Decomposing into Bohr frequencies~{\cite[Appendix A]{chen2023quantum}}]\label{lem:Bohr_decomp} For any Hamiltonian $\vH$, the Heisenberg dynamics for a (not necessarily Hermitian) operator $\vA$ can be decomposed as
\begin{align}
    \vA(t):=\e^{\ri \vH t} \vA \e^{-\ri \vH t} &= \sum_{\nu\in B(\vH)} \e^{\ri \nu t} \vA_\nu.
\end{align}
\end{lem}
The following basic properties will be exploited extensively.
\begin{lem}[Operator Parseval's identity,{ adapted from~\cite[Proposition A.2]{chen2023quantum}}]\label{lem:OParseval}
	For a set of matrices $\{\vA^a\}_{a\in \CA}$ and a Hamiltonian $\vH$, 
	then for any real-valued function $\gamma,$
	\begin{align}
		\lnorm{\sum_{a\in A}\int_{-\infty}^{\infty} \gamma(\omega)\hat{\vA}^a(\omega)^{\dagger}\hat{\vA}^a(\omega) \mathrm{d}\omega}
		&\le \norm{\gamma(\omega)}_{\infty}\nrm{\sum_{a\in A} \vA^{a\dagger}\vA^a}, \quad \text{and}\label{eq:contFourierIdentity}\\
		\lnorm{\sum_{a\in A}\int_{-\infty}^{\infty} \gamma(\omega)\hat{\vA}^a(\omega)\hat{\vA^a}(\omega)^{\dagger} \mathrm{d}\omega}
		 &\le  \norm{\gamma(\omega)}_{\infty}\nrm{\sum_{a\in A} \vA^a\vA^{a\dagger}}.	\label{eq:contFourierIdentity2}
	\end{align}
\end{lem}

\begin{lem}[Decomposing an operator by the energy change {\cite[Lemma IX.1]{chen2025GibbsMarkov}}]\label{lem:sumoverenergies}
For any (not necessarily Hermitian) operator $\vA$, we have that
\begin{align}
\vA =  \frac{1}{\sqrt{2\sigma\sqrt{2\pi}}} \int_{-\infty}^{\infty} \hat{\vA}(\omega)\rd \omega.
\end{align}
\end{lem}

In what follows, we review several facts on the norm bounds of these operator Fourier transforms. 

\begin{lem}
    [Norm Bounds at High Frequencies {\cite[Corollary IX.2]{chen2025GibbsMarkov}}]\label{lem:high_freq_bound} For any $\beta, \omega\in \mathbb{R}$ and operator $\vA$, 
    \begin{equation}
        \|\hat{\vA}(\omega)\|\leq \frac{e^{-\beta\omega+\sigma^2\beta^2}}{\sqrt{\sigma}\sqrt{2\pi}}\cdot \|e^{-\beta\vH} \vA e^{\beta\vH}\|
    \end{equation}
\end{lem}

\begin{lem}
    [Imaginary Time Conjugation of Paulis {\cite[Corollary IX.1]{chen2025GibbsMarkov}}]\label{lem:norm_bound_pauli} Let $\vH$ denote a bounded degree $d$ Hamiltonian. 
    Then, for any $\beta\in \mathbb{R}$ and $w$-qubit Pauli $\vA$, 
    \begin{equation}
        \|e^{-\beta\vH} \vA e^{\beta\vH}\|\leq \bigg(\frac{1}{1-2d|\beta|}\bigg)^w
    \end{equation}
\end{lem}

\begin{lem}[Norm Bounds on Imaginary Time Conjugation~{\cite[Lemma IX.2]{chen2025GibbsMarkov}}]\label{lem:bounds_imaginary_conjugation}
For any $\beta ,\omega\in \BR$ and operator $\vA$, the operator Fourier transform $\hat{\vA}(\omega)$ with uncertainty $\sigma$~\eqref{eq:OFT},~\eqref{eq:fwft} satisfies 
\begin{align}
        \e^{\beta \vH} \hat{\vA}(\omega) \e^{-\beta \vH}
& = \e^{\beta\omega} \cdot \hat{\vA}(\omega+2\sigma^2\beta) \e^{\sigma^2\beta^2}.
\end{align}
Thus, 
\begin{align}
    \norm{\e^{\beta \vH} \hat{\vA}(\omega) \e^{-\beta \vH}} \le \frac{\e^{\sigma^2\beta^2}}{\sqrt{{\sigma}\sqrt{2\pi}}} \e^{\beta\omega} \norm{\vA}.
\end{align}
\end{lem}

\subsection{Double Bohr frequency decomposition}
In our study of metastability, we need to consider two different sets of eigenstates, of both the metastable state and the ideal Gibbs state. Suppose we fix a pair of Hamiltonians $\vH_1$, $\vH_2$. Then, for any operator $\vA$, we consider its sequential decomposition into both Bohr frequencies $\nu_1,\nu_2$ of $\vH_1$, $\vH_2$ via:
\begin{align}
(\vA_{\nu_1})_{\nu_2} := \sum_{E'_2 - E_2 = \nu_2  } \sum_{E'_1 - E_1 = \nu_1 } \vP_{E'_2} (\vP_{E'_1} \vA \vP_{E_1})\vP_{E_2}\quad \text{for each}\quad \nu_1 \in B(\vH_1),\ \nu_2 \in B(\vH_2). 
\end{align}
In general, the order of decomposition matters $(\vA_{\nu_1})_{\nu_2} \ne (\vA_{\nu_2})_{\nu_1}$ as $\vH_1$ and $\vH_2$ may not commute. Nevertheless, for the expressions we care about, their double Bohr frequency decomposition still takes a manageable form, and we can conveniently manipulate the coefficients without writing down the full expansion.

\begin{lem}[Double Bohr frequency decomposition{, adapted from~\cite[Lemma III.2]{chen2025learning}}]\label{lem:double_Bohr}
For any operator $\vA,$ and Hermitian operators $\vH_1,\vH_2$, $z\in \BC,$
    \begin{align}
 \e^{z\vH_2}\e^{-z\vH_1}\vA\e^{z\vH_1}\e^{-z\vH_2}&= \sum_{\nu_1\in B(\vH_1),\nu_2\in B(\vH_2)} (\vA_{\nu_{1}})_{\nu_2} e^{z(\nu_2-\nu_1)},\label{eq:Osinh}\\
 \e^{z\vH_2}[\vH_2-\vH_1,\e^{-z\vH_1}\vA\e^{z\vH_1}]\e^{-z\vH_2}&= \sum_{\nu_1\in B(\vH_1),\nu_2\in B(\vH_2)} (\vA_{\nu_{1}})_{\nu_2} (\nu_2-\nu_1)e^{z(\nu_2-\nu_1)}\\
         [\vA,\vH_2] - [\vA,\vH_1] &= -\sum_{\nu_1\in B(\vH_1),\nu_2\in B(\vH_2)} (\vA_{\nu_1})_{\nu_2} (\nu_2-\nu_1)\label{eq:OHH},
    \end{align}
    where $B(\vH_1),B(\vH_2)$ are respectively the set of Bohr frequencies of $\vH_1,\vH_2$.
\end{lem}

\subsection{Weighted Inner Products}

To quantify the notion of detailed balance, we introduce a set of weighted inner products and norms. 
\begin{defn}
    [$(\vrho, s)$ Weighted Inner Product]\label{defn:s_inner} Given a full-rank state $\vrho$ and $s\in [-\frac{1}{2}, \frac{1}{2}]$, we define
    \begin{equation}
        \langle \vX,\vY\rangle_{\vrho, s}:=\tr[\vX^\dagger \vrho^{\frac{1}{2}+s}\vY\vrho^{\frac{1}{2}-s}]\,.
    \end{equation}
    We denote by $\|\vX\|_{\vrho, s} = \sqrt{\langle \vX,\vX\rangle_{\vrho, s}}$ the $(\vrho, s)$-weighted $2$-norm.
\end{defn}

\begin{rmk}
    We recognize this notation is distinct from the convention in the literature, where $s\in [0, 1]$. We adopt this choice to later simplify formulas that interpolate over different values of $s$.
\end{rmk}

\noindent The $s=0$ case is the well-known Kubo-Martin-Schwinger (KMS) inner product. The reader may then question why we introduced other weighted norms if we are going to focus on KMS. First of all, other $\vrho$-weighted norms are technically possible in the choice of detailed balance condition -- e.g., GNS~\cite{kastoryano2016commuting}, where $s=\frac{1}{2}$. Perhaps more important will be the role of interpolations between the several $s$ norms in our analysis. To conclude this discussion, we remark that the conversion between these weighted norms to the operator norm always holds, but sometimes may be suboptimal.
\begin{lem}[Operator and Weighted Norm Bounds]\label{lem:operatornorm} Unconditionally, we have 
\begin{equation}
    \norm{\vX}_{\vrho}\le \norm{\vX}\quad \text{and}\quad \braket{\vX,\vY}_{\vrho} \le \norm{\vX}\norm{\vY}.
\end{equation}
\end{lem}

\subsection{The Dirichlet Form}
\noindent Instead of thinking in the Schrodinger picture (i.e., the evolution of the state under the dynamics), the starting point of our analysis is to understand the evolution of operators. For this purpose, \cite[Lemma C.2]{rouze2024efficient} derived a clean and explicit expression for the Dirichlet form of the Lindbladians \eqref{eq:exact_DB_L} we study, which in some sense quantifies the KMS inner product of operators under $\CL$, as an inner product of gradients:
\begin{lem}[The Dirichlet form as an Inner Product of Gradients {\cite[Lemma C.2]{rouze2024efficient},\cite[Lemma X.3]{chen2025GibbsMarkov}}]
\label{lem:Dicirchlet}
    The Dirichlet form for the $\vrho$-KMS detailed-balanced Lindbladian  $\CL_{a}$~\eqref{eq:exact_DB_L}, is a $\vrho$-weighted inner product of commutators:
\begin{align}
   \CE_a(\vX,\vY) &:= -\braket{\vX,\CL_a^{\dagger}[\vY]}_{\vrho} \\&= \iint_{-\infty}^{\infty} \braket{[\hat{\vA}^a(\omega,t),\vX ], [\hat{\vA}^a(\omega,t),\vY ]}_{\vrho}\cdot g(t) h(\omega) \cdot \rd t \rd \omega,
\end{align}
where the specific functions $g, h$ are time and frequency filter functions:
\begin{align}
 h(\omega)&= e^{-\sigma^2\beta^2/8} e^{-\labs{\omega}\beta/2}\ge 0, \quad g(t) = \frac{1}{\beta\cosh(2\pi t/ \beta)}\ge 0.
\end{align}
\end{lem}

By ``inner product of gradients", we recall from \eqref{eq:nc_grad} that commutators with the jump operators are playing the role of a specific non-commutative discrete derivative (see also \cite{carlen2024dynamics}). Indeed, when $\vY=\vX$, one can re-express the above in terms of a weighted norm $\big\|\cdot\big\|_{\mathsf{D},\vrho}^2$ of gradients of $\vX$:
\begin{align}
    \CE_a(\vX, \vX) = \iint_{-\infty}^\infty \bigg\|\nabla_{a, \omega, t}[\vX]\bigg\|_{\vrho}^2\cdot  g(t)h(\omega)\cdot \rd t\rd\omega := \big\|\nabla_{a}[\vX]\big\|_{\mathsf{D},\vrho}^2
\end{align}

We note that the particular $h(\omega), g(t)$ that arise in the expression above depend on the particular choice of KMS-detailed balanced Lindbladian. To generalize our approach to other such Lindbladians, it suffices to carefully keep track of these weights. To conclude, we note that $g, h$ are properly normalized, in that they integrate to constants. 
\begin{align}
    \int_{-\infty}^{\infty} g(t) \rd t = \int_{-\infty}^{\infty} \frac{1}{2\pi \cosh(x)} \rd x = \frac{1}{2}\label{eq:int_g}
\end{align}

\section{Entropy Production and the Fisher Information}
\label{section:EP_GF}

When a classical probability density $\nu_t$ evolves towards some stationary equilibrium density $\pi$ under a diffusion process (e.g., a Fokker–Planck equation), it satisfies an identity connecting the decay of the relative entropy $D(\nu_t||\pi)$ to the local structure of the energy landscape. This identity, known as \textit{entropy dissipation identity}, relates the entropy production to the relative Fisher information of $\nu_t$ and $\pi$:
\begin{equation}\label{eq:continuous_ed}
    \frac{\rd}{\rd t} D(\nu_t||\pi) =-\mathsf{FI}[\nu_t||\pi], \quad  \mathsf{FI}[\nu||\pi]:= \int \bigg\|\nabla \log \frac{\nu}{\pi}\bigg\|^2\cdot  \nu(x) \cdot \rd x.
\end{equation}

\noindent What is remarkable about this expression is that the time derivative of the relative entropy is governed by (a weighted norm of) its spatial derivatives. In some sense, the Fisher information is meant to capture the how the density have have correct thermal ratios within each local-minimum, while tolerate ranges of weights across local minima (c.f. \cite{Bovier2000MetastabilityAL, Bovier2004MetastabilityIR, Bovier2005MetastabilityIR, Balasubramanian2022TowardsAT, BH16}).\footnote{For instance, when $\mathsf{FI}[\nu||\pi]\approx 0$, it implies $\nu(x)/\pi(x)$ is locally constant/flat around samples from $\nu$, reproducing \eqref{eq:classical_ADB}.}  Analogous entropy dissipation identities have been established for both discrete-state classical Markov chains and quantum Markov chains (dynamical semigroups) satisfying GNS detailed-balance (c.f. \cite{carlen2024dynamics}), where the rate of decay of the relative entropy is again expressed in terms of a relative Fisher information functional.

The goal of this section is to extend this theory to KMS-detailed balanced Lindbladians. In particular, the main result of this section is to derive an entropy dissipation identity akin to \eqref{eq:continuous_ed}, which will play a central role in our understanding of the several equivalent notions of metastability. We organize the rest of this section as follows. We begin in Section~\ref{subsection:results_EP_GF}, by formally presenting the main conclusions of this section. In Section~\ref{section:classical_fi}, we include a classical analog as an exercise, in the context of a discrete-state classical Markov chain. We then show how to exhibit the entropy production differential (Section~\ref{section:exhibiting_gradients}), and conclude with the characterization of the entropy production rate in Section~\ref{subsection:EP_GF}.

\subsection{The Results of this Section}
\label{subsection:results_EP_GF}

For any KMS-detailed balanced local Lindbladian $\mathcal{L}_a$ (associated with a jump operator $\vA^a$), we write its local entropy production rate as:
\begin{equation}\label{eq:ep_local}
    \mathsf{EP}_{\mathcal{L}_a}[\vsigma] = - \tr[\mathcal{L}_a[\vsigma](\log \vsigma - \log \vrho)] = - \frac{\rd}{\rd t}D(e^{t\mathcal{L}_a}[\vsigma]||\vrho)\bigg|_{t=0}.
\end{equation}

The main result of this section offers an entropy dissipation identity relating the rate-of-change of the relative entropy to its ``spatial" derivatives.
\begin{thm}
    [The Entropy Dissipation Identity]\label{thm:integral_square_logs} Let $\mathcal{L}_a$ denote the local Lindbladian (\ref{eq:exact_DB_L}) associated to a jump operator $\vA^a$, which is (KMS)-$\vrho$-detailed-balanced. Then, for any full-rank state $\vsigma$,
    \begin{equation}
        \mathsf{EP}_{\mathcal{L}_a}[\vsigma] = \mathsf{FI}_a[\vsigma||\vrho].
    \end{equation}
\end{thm}
To write down the Fisher information $\mathsf{FI}_a[\vsigma||\vrho],$ we introduce an analog of ``spatial derivatives'' in terms of commutators with the jump operators. For any observable $\vO$, we define the $a$ and $(a, \omega, t)$ gradients:
\begin{equation}
    \nabla_{a, \omega, t}[\vO] := [\vA^a(\omega, t), \vO], \quad \text{and}\quad \nabla_a[\vO] = \iint_{-\infty}^{\infty} \nabla_{a, \omega, t}[\vO] \otimes \ket{\omega, t} \rd \omega \rd t.
\end{equation}

\begin{defn} [The Fisher Information]\label{defn:fisher} In the context of~\autoref{thm:integral_square_logs}, for any state $\vsigma$, we denote the quantum \emph{Fisher Information} of $\vsigma$ relative to $\vrho$ and $\vA^a$ as 
\begin{equation}
    \mathsf{FI}_a[\vsigma||\vrho]:= \bigg\| \nabla_a[\log \vsigma - \log \vrho]\bigg\|^2_{\mathsf{FI}, \vsigma} :=  \iint_{-\infty}^{\infty}\int_{-1/2}^{1/2} \bigg\| \nabla_{a, \omega, t}[\log \vsigma - \log \vrho]\bigg\|_{\vsigma, s}^2 h_s(\omega)g_s(t)\rd s \rd \omega \rd t,
\end{equation}

\noindent where we introduce the $s$-weighted time and frequency filters $h_s(\omega), g_s(t)$, defined by
    \begin{align}
\text{for each}\quad s\in (-\frac{1}{2}, \frac{1}{2}):\quad h_s(\omega) &:=\exp(\frac{s\beta(2\omega-s\beta\sigma^2)}{2})h(\omega-s\beta\sigma^2)> 0 \\
    g_s(t) &:= \frac{2}{\beta} \frac{\cos(s\pi)\cosh(\frac{2\pi t}{\beta})}{\cosh(\frac{4\pi t}{\beta})+\cos(2s\pi)} = \frac{1}{2\pi}\int_{-\infty}^{\infty}e^{i \nu t} \frac{\cosh(s\beta \nu/2)}{2\cosh(\beta\nu/4)} \rd \nu > 0
\end{align}
with $h(\omega)$ as in the Dirichlet form in \autoref{lem:Dicirchlet}.
\end{defn}

The above is defined with respect to each local jump operator $\vA^a$ regarding the entropy production of each local Lindbladian $\mathcal{L}_a$, which extends to sums of Lindbladians by linearity. The $s$-weighted filter functions $h_s, g_s$ may seem ad-hoc, but are inherited from the Dirichlet form in \autoref{lem:Dicirchlet} and modified according to each $s.$ The following a priori bounds for the weights will be invoked frequently (see Section~\ref{sec:proof_suph}).
\begin{lem} The time-filter is integrable \begin{align}
    \int_{-\infty}^{\infty} g_s(t)\rd t = \frac{\cosh(s\beta \nu/2)}{2\cosh(\beta\nu/4)}\mid_{\nu=0} = \frac{1}{2}, \quad \text{for each}\quad s \in (-\frac{1}{2},\frac{1}{2})
\end{align}
which remains bounded in taking the $s\rightarrow \pm \frac{1}{2}$ limit. 
\end{lem}
\begin{lem}\label{lem:sup_h}
    The $s$-weighted frequency filter is bounded in the regime $s\in [-\frac{1}{2}, \frac{1}{2}]$:
\begin{equation}
    \sup_{s\in [-\frac{1}{2}, \frac{1}{2}]} \sup_{\omega} h_s(\omega) \leq 1.
\end{equation}    
In particular, $h_{-\frac{1}{2}}(\omega) = \gamma(\omega).$
\end{lem}

\begin{rmk}
    In the perfect energy resolution limit where the Gaussian width approaches $\sigma\rightarrow 0$, $\mathcal{L}_a$ recovers the Davies generator and we recover the Fisher Information functional for GNS detailed-balanced Lindbladians, see \cite{carlen2024dynamics} for a review. 
\end{rmk}

\autoref{thm:integral_square_logs} will be a stepping stone towards understanding states which are metastable, in that they are approximately stationary under the Lindblad evolution and thereby have a bounded entropy production rate. Indeed, \autoref{thm:integral_square_logs} suggests that approximately stationary states behave ``locally" like Gibbs states:
\begin{equation}\label{eq:locally_flat}
    \mathsf{EP}_{\mathcal{L}_a}[\vsigma]\approx 0 \Rightarrow \nabla_a[\log \vsigma]\approx -\beta\cdot \nabla_a[\vH].
\end{equation}
\noindent The notion of locality here is with respect to the jump operators (interpreted as spatial directions). This connection will be made quantitative in the Section \ref{section:adb}.

\subsection{The Classical Case: discrete space}
\label{section:classical_fi}
We dedicate this section to a simple exposition of the classical analog of \autoref{thm:integral_square_logs} in the context of discrete-state classical Markov chains, which marks some differences from the continuous form~\eqref{eq:continuous_ed}.

\begin{defn}
    [The (Classical) Entropy Production Rate] Fix a Markov chain with transition matrix $\mathsf{P}$, continuous time generator $\mathsf{P}-\mathbb{I}$, and stationary distribution $\pi$. For an arbitrary distribution $\nu$, we define
    \begin{equation}
         \mathsf{EP}_{\mathsf{P}}(\nu) \equiv  \mathbb{E}_{x\sim \pi}\mathbb{E}_{y\sim \mathsf{P}x} \bigg[(f(y)-f(x))\log \frac{f(y)}{f(x)}\bigg] >0
    \end{equation}
    With $f(x)\equiv \nu(x)/\pi(x)$ the \emph{likelihood ratio}. 
\end{defn}

In the above, the $\log$ likelihood ratio $\log \frac{f(y)}{f(x)} = \nabla[\log f](x)$ plays the role of the entropy gradient, in that it compares the log likelihood at two adjacent states $x\sim y$ under the Markov chain updates. The analog of the entropy dissipation identity rewrites the above in terms of a positive function of this entropy gradient, by leveraging an integral identity to exhibit the logs (\autoref{lem:int_log}):
\begin{align}
   \mathsf{EP}_{\mathsf{P}}(\nu) =  \mathbb{E}_x\mathbb{E}_{y\sim \mathsf{P}x} \bigg[(f(y)-f(x))\log \frac{f(y)}{f(x)}\bigg] &= \int_0^1 ds\cdot \mathbb{E}_{x, y}\bigg[ f(x) \bigg(\frac{f(y)}{f(x)}\bigg)^s \log^2 \frac{f(y)}{f(x)}\bigg] \tag*{(Fact~\ref{lem:int_log})} \\ &= \int_0^1 ds \cdot \mathbb{E}_{x, y}\bigg| f(y)^{\frac{s}{2}} f(x)^{\frac{1-s}{2}}  \log \frac{f(y)}{f(x)}\bigg|^2. \label{eq:ent_dis}
\end{align}

We encourage the reader to recall and contrast this with both the continuous case Eq.~\eqref{eq:continuous_ed}, and our non-commutative analog in~\autoref{defn:fisher}. The discrete-space classical case \eqref{eq:ent_dis} already exhibits a convex combination of $s$-weighted norms of discrete gradients, and the quantum setting (in both our result and the GNS setting \cite{carlen2024dynamics}) is a noncommutative version of this statement.

\begin{lem}\label{lem:int_log}
    For any $\alpha\neq 0$, we have that, $\alpha-1 = \log \alpha\cdot \int_{0}^1 \alpha^s ds.$
\end{lem}

\subsection{Exhibiting Entropy Gradients}
\label{section:exhibiting_gradients}

The starting point of our proof is to recollect the Dirichlet form of the Lindbladians of \cite{chen2023efficient} (\autoref{lem:Dicirchlet}). Recall that the Dirichlet form allows us to study the time derivative of the expectation value of any observable $\vO$, under the evolution of $\mathcal{L}_a$, at any current state $\vsigma$; and it can be written as: 
\begin{align}\label{equation:dirichlet_O_GF}
    \tr[-\CL_a^{\dagger}[\vO]\vsigma] &= - \tr\L[ \sqrt{\vrho} (\frac{1}{\sqrt{\vrho}}\vsigma\frac{1}{\sqrt{\vrho}})\sqrt{\vrho}\CL_a^{\dagger}[\vO]\R] = \CE_{a}\bigg(\frac{1}{\sqrt{\vrho}}\vsigma\frac{1}{\sqrt{\vrho}},\vO\bigg) \\ &
    =\iint_{-\infty}^{\infty} \tr\L[\sqrt{\vrho}[\vA^a(\omega,t),\vO]\sqrt{\vrho}[\vA^a(\omega,t),\frac{1}{\sqrt{\vrho}}\vsigma\frac{1}{\sqrt{\vrho}}]^{\dagger}\R] h(\omega)g(t) \rd \omega \rd t\label{eq:LOsigma_Dirichlet}.
\end{align}

The broad goal of our analysis is to understand the evolution of a special observable $\vO = \log \vsigma-\log \vrho$, which corresponds to the entropy production; the case of generic observables is discussed next. The idea is to re-write the evolution \eqref{eq:LOsigma_Dirichlet}, in terms of (yet another) entropy differential $\log \vsigma-\log \vrho$, to make a square. For this purpose, we leverage the following matrix identity that generalizes \autoref{lem:int_log}.

\begin{lem}[Exposing logarithms]\label{lem:expose_log}
For any operator $\vA,$ and full rank PSD matrices $\vrho,\vsigma:$
    \begin{align}
        \sqrt{\vrho}[\vA, \vrho^{-\frac{1}{2}}\vsigma\vrho^{-\frac{1}{2}} ] \sqrt{\vrho} = \int_{-1/2}^{1/2} \sqrt{\vsigma} \L(\vsigma^{s}[\vrho^{-s}\vA\vrho^{s},\log\vsigma-\log\vrho] \vsigma^{-s} \R)
         \sqrt{\vsigma}\rd s.
    \end{align}
\end{lem}
Related arguments can be found in~\cite[Lemma 2.14]{carlen2024dynamics},~\cite[Lemma III.3]{chen2025learning}.
\begin{proof}
It suffices to study the individual Bohr frequencies $\vA = \sum_{\nu\in B(\vH)}\vA_{\nu}$. For each $\nu\in B(\vH),$
        \begin{align}
        \sqrt{\vrho}[\vA_{\nu}, \vrho^{-\frac{1}{2}}\vsigma\vrho^{-\frac{1}{2}}] \sqrt{\vrho}
        &= \sqrt{\vsigma}\L( \vsigma^{-\frac{1}{2}}\vrho^{\frac{1}{2}}\vA_{\nu}\vrho^{-\frac{1}{2}}\vsigma^{\frac{1}{2}}-\vsigma^{\frac{1}{2}}\vrho^{-\frac{1}{2}}\vA_{\nu}\vrho^{\frac{1}{2}}\vsigma^{-\frac{1}{2}} \R)\sqrt{\vsigma} \\
        & = -\sum_{\mu\in B(\vH_{eff})} \sqrt{\vsigma} (\vA_{\nu})_{\mu}\sqrt{\vsigma}\cdot 2\sinh(\frac{\beta(\nu-\mu)}{2})\label{eq:reduce_sinh}
    \end{align}
    using that $\log\vsigma =-\beta \vH_{eff}$ and $\log\vrho =-\beta \vH + (const.) \vI$ and the double bohr frequency decomposition~\autoref{lem:double_Bohr}. Now, we introduce a derivative and integration to expose a term linear in $\nu-\mu$, which translates to commutator with $\log \vsigma - \log \vrho$
    \begin{align}
        \eqref{eq:reduce_sinh} &= - \int_{-1/2}^{1/2} \sum_{\mu\in B(\vH_{eff})} \sqrt{\vsigma} (\vA_{\nu})_{\mu} \sqrt{\vsigma} \cosh(s\beta(\nu-\mu)) \beta (\nu-\mu)  \rd s \\ 
        & = \int_{-1/2}^{1/2} \sqrt{\vsigma} \vsigma^{s}[\vrho^{-s}\vA_{\nu}\vrho^{s},\log\vsigma-\log\vrho] \vsigma^{-s}
         \sqrt{\vsigma}\rd s, \tag*{(By~\autoref{lem:double_Bohr})}
    \end{align}
    \noindent using that $2\sinh(x/2) =\int_{-1/2}^{1/2}\cosh(sx) x\rd s $, and the symmetry of integration domain $\int_{-1/2}^{1/2} \cosh(sx) \rd s = \int_{-1/2}^{1/2} e^{xs}  \rd s.$
\end{proof}

By applying the above identity to the expression in Eq.~\eqref{equation:dirichlet_O_GF} for the relevant time-derivative, we arrive at the following conclusion. For any observable $\vO$, its rate-of-change can be written as 
\begin{align}\label{eq:grad_flow}
\tr[\CL_a^{\dagger}[\vO]\vsigma]  = - &\iint_{-\infty}^{\infty} \int_{-1/2}^{1/2} \tr\L[ [\vA^a(\omega,t),\vO] \vsigma^{\frac{1}{2}-s}[\vrho^{-s}\vA^a(\omega, t)\vrho^{s},\log\vsigma-\log\vrho]^\dagger \vsigma^{\frac{1}{2}+s}\R]h(\omega)g(t)\rd \omega \rd t\rd s
\end{align}

We remark that this expression is not too dissimilar to a \textit{gradient flow} equation, which writes the time-derivative of the state $\vsigma$ itself, $\mathcal{L}[\vsigma]$, in terms of second derivatives of the entropy gradients. We leave exploring this connection to future work. For a discussion of the GNS case, see e.g. \cite[Eq. 2.16]{carlen2024dynamics}.

\subsection{The Entropy Dissipation Identity}
    \label{subsection:EP_GF}
    Now, the central challenge in making~\eqref{eq:grad_flow} a ``square'' is that our equation seems fundamentally not symmetric in its spatial derivatives -- i.e., the conjugation $\vrho^{s}\vA(\omega,t)\vrho^{-s}$ only acts on one side (why should this quantity even be positive?). This issue is absent in the GNS-detailed-balanced case, and appears in the KMS, finite energy resolution setting.

This is indeed an issue for generic observables $\vO$. But for the particular case of the entropy production rate $\vO = \log \vsigma-\log \vrho$, we can expose the positivity of the expression by a simple symmetrization. Since $\vO = \log \vsigma-\log \vrho$ appears \textit{twice}, taking the complex conjugation has precisely the effect of moving the conjugation from $\vA_{\nu_2}$ to $\vA_{\nu_1}$. For each $s\in[-1/2,1/2]$ and $\omega, t \in \BR,$
\begin{align}
&\tr\L[[\vA(\omega,t),\vO] \vsigma^{\frac{1}{2}-s} \L[\vrho^{-s}\vA(\omega,t)\vrho^{s},\vO\R]^{\dagger} \vsigma^{\frac{1}{2}+s}
         \R]^*=\tr\L[\L[\vrho^{-s}\vA(\omega,t)\vrho^{s},\vO\R] \vsigma^{\frac{1}{2}-s} [\vA(\omega,t),\vO]^{\dagger} \vsigma^{\frac{1}{2}+s}
         \R].\label{eq:Alog_conjugate}
\end{align}

This enables us to cleanly rewrite the effect of the conjugation by shifting the operator Fourier transform frequency weights. 

\begin{lem}[Symmetrization] \label{lem:symmetrize_ArhoArhos1s2}
For any operator $\vA$, $s_1,s_2\in[-1/2,1/2]$ and $\omega, t \in \BR$, we have that 
    \begin{align}
        &\vrho^{-s_1}\vA(\omega)\vrho^{s_1} \otimes (\vrho^{-s_2}\vA(\omega)\vrho^{s_2})^{\dagger} + \vrho^{-s_2}\vA(\omega)\vrho^{s_2} \otimes (\vrho^{-s_1}\vA(\omega)\vrho^{s_1})^{\dagger}\\
        =& \exp\bigg[\frac{s_+\beta(2\omega+s_+\beta\sigma^2)}{2}\bigg]\sum_{\nu_1,\nu_2\in B(\vH)} \vA_{\nu_1}(\omega+s_+\beta\sigma^2) \otimes \vA_{\nu_2}(\omega+s_+\beta\sigma^2)^{\dagger} \cosh\bigg(\frac{ (s_1-s_2)\beta (\nu_1-\nu_2)}{2}\bigg).
    \end{align}
    where $s_+:=s_1+s_2.$ 
\end{lem}
\begin{proof}
    We directly compute in the Bohr frequency decomposition and rewrite in terms of $\nu_1+\nu_2$ and $\nu_1-\nu_2$
\begin{align}
    \vrho^{-s_1}\vA(\omega)\vrho^{s_1} &\otimes (\vrho^{-s_2}\vA(\omega)\vrho^{s_2})^{\dagger}=\sum_{\nu_1,\nu_2\in B(\vH)} \vA(\omega)_{\nu_1}e^{\beta s_1\nu_1} \otimes (\vA(\omega)_{\nu_2})^{\dagger}e^{\beta s_2\nu_2} \\
    &=\sum_{\nu_1,\nu_2\in B(\vH)} \vA(\omega)_{\nu_1}e^{\beta (s_1+s_2)\nu_1/2} \otimes (\vA(\omega)_{\nu_2})^{\dagger} e^{\beta (s_1+s_2)\nu_2/2} \cdot e^{\frac{\beta}{2}(s_1-s_2)(\nu_1-\nu_2)}\\
    &=\sum_{\nu_1,\nu_2\in B(\vH)} e^{\beta s_+\vH/2}\vA(\omega)_{\nu_1}e^{-\beta s_+\vH/2}\otimes (e^{\beta s_+\vH/2}\vA(\omega)_{\nu_2}e^{-\beta s_+\vH/2})^{\dagger} \cdot e^{\frac{\beta}{2}(s_1-s_2)(\nu_1-\nu_2)}\label{eq:Anu2conjugate}.
\end{align}
Use~\autoref{lem:bounds_imaginary_conjugation} and symmetrize $s_1\leftrightarrow s_2$ to conclude the proof.
\end{proof}
\begin{proof}

[of~\autoref{thm:integral_square_logs}] Since the entropy production rate is a real number, we can symmetrize the expression~\eqref{eq:Alog_conjugate} and apply~\autoref{lem:symmetrize_ArhoArhos1s2}
\begin{align}
&\tr[-\CL[\vsigma](\log\vsigma-\log\vrho)]=\frac{1}{2}\tr[-\CL[\vsigma](\log\vsigma-\log\vrho)]+\frac{1}{2}\tr[-\CL[\vsigma](\log\vsigma-\log\vrho)]^*\\
    &= \frac{1}{2}\iint_{-\infty}^{\infty}\int_{-1/2}^{1/2} \tr\L[[\vA(\omega,t),\log\vsigma-\log\vrho] \vsigma^{\frac{1}{2}-s} \L[\vrho^{-s}\vA(\omega,t)\vrho^{s},\log\vsigma-\log\vrho\R]^{\dagger} \vsigma^{\frac{1}{2}+s}
         \R] h(\omega)g(t)\rd s \rd \omega \rd t\\
    &+\frac{1}{2}\iint_{-\infty}^{\infty}\int_{-1/2}^{1/2} \tr\L[[\vrho^{-s}\vA(\omega,t)\vrho^{s},\log\vsigma-\log\vrho] \vsigma^{\frac{1}{2}-s} \L[\vA(\omega,t),\log\vsigma-\log\vrho\R]^{\dagger} \vsigma^{\frac{1}{2}+s}
         \R] h(\omega)g(t)\rd s \rd \omega \rd t \tag*{(By~\eqref{eq:Alog_conjugate})}\\
    &=\iint_{-\infty}^{\infty}\int_{-1/2}^{1/2} \exp(\frac{s\beta(2\omega+s\beta\sigma^2)}{2})h(\omega)\lnorm{ \vsigma^{\frac{1+2s}{4}}\L[\vA(\omega+\frac{s\beta\sigma^2}{2},t),\log\vsigma-\log\vrho\R] \vsigma^{\frac{1-2s}{4}}}^2_2  g_s(t)\rd s \rd \omega \rd t \label{eq:getting_gs}
\end{align}
where we combine the $\cosh(\frac{s\beta (\nu_1-\nu_2)}{2})$ coefficients from \autoref{lem:symmetrize_ArhoArhos1s2} with the $\frac{1}{2\cosh(\beta(\nu_1-\nu_2)/4)}$ coefficient from the Dirichlet form (\autoref{lem:Dicirchlet}), to define $g_s(t):$
\begin{align}
    \sum_{\nu_1,\nu_2} \cosh(\frac{s\beta (\nu_1-\nu_2)}{2}) \int_{-\infty}^{\infty} \vA(t)_{\nu_1}\otimes (\vA(t)_{\nu_2})^{\dagger} g(t) \rd t &=\sum_{\nu_1,\nu_2} \vA_{\nu_1}\otimes (\vA_{\nu_2})^{\dagger}  \cdot \frac{ \cosh(s\beta (\nu_1-\nu_2)/2)}{2\cosh(\beta(\nu_1-\nu_2)/4)}\\
    &= \int_{-\infty}^{\infty} \vA(t)\otimes \vA(t)^{\dagger} g_s(t) \rd t. 
\end{align}
We express $g_s(t)$ as a ``unified" time-domain function:
\begin{align}
    g_s(t) :=  \frac{1}{2\pi}\int_{-\infty}^{\infty}e^{i \nu t} \frac{\cosh(s\beta \nu/2)}{2\cosh(\beta\nu/4)} \rd \nu
    = \frac{2}{\beta} \frac{\cos(s\pi)\cosh(\frac{2\pi t}{\beta})}{\cosh(\frac{4\pi t}{\beta})+\cos(2s\pi)}> 0 \tag*{(By~\autoref{lem:cos_int})}.
\end{align}
We may further simplify the $\omega$ integrals by a change of variables
\begin{align}
    \eqref{eq:getting_gs}
         &=\iint_{-\infty}^{\infty}\int_{-1/2}^{1/2} \lnorm{ \vsigma^{\frac{1+2s}{4}}\L[\vA(\omega,t),\log\vsigma-\log\vrho\R] \vsigma^{\frac{1-2s}{4}}}_2^2 h_s(\omega)g_s(t)\rd s \rd \omega \rd t 
\end{align}
where for each $s$ such that $\labs{s}<1/2$, we defined 
\begin{align}
h_s(\omega) &:=\exp\bigg(\frac{s\beta(2\omega-s\beta\sigma^2)}{2}\bigg)h(\omega-s\beta\sigma^2)\ge 0
\end{align}
which concludes the proof.
\end{proof}
The following explicitly integral formula was employed above:
\begin{lem}[{\cite[Eq 3.981.10]{gradshteyn2014table}}]\label{lem:cos_int} 
\begin{align}
    \frac{1}{2\pi} \int_{-\infty}^{\infty} \cos( t \nu) \frac{\cosh(b\nu)}{\cosh(c\nu)} \rd \nu = \frac{1}{c} \frac{\cos(\frac{b\pi}{2c})\cosh(\frac{\pi t}{2c})}{\cosh(\frac{\pi t}{c})+\cos(\frac{b\pi}{c})}\quad \text{for all}\quad \labs{\Re(b)}\le \Re(c), \quad \text{all real}\quad t,
\end{align}
\end{lem}

\subsubsection{Deferred proof of~\autoref{lem:sup_h}}\label{sec:proof_suph}

\begin{proof}

[of~\autoref{lem:sup_h}]
Recall from~\autoref{thm:integral_square_logs},
\begin{equation}
    h_s(\omega) =\exp\bigg[\frac{s\beta(2\omega-s\beta\sigma^2)}{2}\bigg]\cdot h(\omega-s\beta\sigma^2) \text{ with } h(\omega)= \exp\bigg[-\frac{\sigma^2\beta^2}{8}  -\frac{\labs{\omega}\cdot \beta}{2} \bigg] 
\end{equation}

\noindent Let us first group the $\omega$-dependent terms. In the region $s\in [-\frac{1}{2}, \frac{1}{2}]:$
\begin{equation}
    \exp\bigg(s\beta\omega - \beta\cdot  |\omega - s\beta\sigma^2|/2\bigg) \leq \exp \bigg( |s|\cdot \beta^2\sigma^2/2\bigg). \tag*{(Triangle inequality)}
\end{equation}
\noindent Consequently, 
\begin{equation}
    \sup_{\omega} h_s(\omega) \leq \exp \bigg( \frac{|s|\beta^2\sigma^2}{2}-\frac{s^2\beta^2\sigma^2}{2} -\frac{\sigma^2\beta^2}{8}\bigg) = \exp\bigg( -\frac{\beta^2\sigma^2}{2} \big(\frac{1}{2}-|s|\big)^2\bigg)\leq 1
\end{equation}
When $s=-\frac{1}{2},$
\begin{align}
h_{-\frac{1}{2}}(\omega) &= \exp\bigg(-\frac{\beta(2\omega + \beta\sigma^2/2)}{4} -\frac{\sigma^2\beta^2}{8} - \frac{\beta\labs{\omega+\beta\sigma^2/2}}{2}\bigg) \\&= \exp\bigg(-\frac{\beta(\omega + \beta\sigma^2/2)}{2} -\frac{\beta\labs{\omega+\beta\sigma^2/2}}{2}\bigg)
\end{align}
use that $-\labs{x}/2 - x/2 = - \max(x,0)$ to conclude the proof.

\end{proof}

\section{Metastable States are Approximately Detailed Balanced}
\label{section:adb}

Through the study of entropy production, we have seen that metastability can be phrased in terms of the Fisher information, interpreted as a gradient of a non-commutative ``log-likelihood ratio''. While this perspective provides an initial step toward a \textit{static} characterization of metastable states without explicitly referencing Lindbladian \textit{dynamics}, it can appear unwieldy at first. The goal of this section is to relate the entropic expression to the behavior of concrete physical observables in thermal equilibrium, and most importantly, to expose clear, interpretable connections to other meaningful notions of metastability.

The genuine quantum Gibbs state $\vrho$ satisfies a microscopic dynamical reversibility, or a detailed balance of transition amplitudes. Put formally, for any Hermitian operator written in the Bohr frequency decomposition $\vA =\vA^{\dagger}= \sum_{\nu \in B(\vH)} \vA_{\nu}$, we have
\begin{align}
    \vA_{\nu} \vrho = \vrho\vA_{\nu}   e^{\beta \nu} \quad \text{for every}\quad \nu \in B(\vH). \label{eq:Arho_rhoA}
\end{align}
This operator-valued identity implies that the probability of $\vA$ changing the energy by $\nu$, is equal to that of $\vA$ changing the energy by $-\nu$, up to a Boltzmann factor:
\begin{align}\label{eq:trans_prob}
    \tr[(\vA_{\nu})^{\dagger}\vA_{\nu} \vrho]  =  \tr[\vA_{\nu} (\vA_{\nu})^{\dagger}\vrho] e^{-\beta \nu} = \tr[(\vA_{-\nu})^{\dagger} \vA_{-\nu}\vrho] e^{-\beta \nu}.
\end{align}

The main result of this section shows that a state $\vsigma$ is metastable under the Lindbladian, if and only if it also obeys an \textit{approximate} detailed balance condition (\autoref{thm:meta_implie_ADB}), roughly
\begin{align}
    \sqrt{\vsigma} \vA_{\nu} \approx \vA_{\nu} \sqrt{\vsigma} e^{-\beta \nu/2}.
\end{align}

For the proof of local recovery, the most natural and effective formulation of the approximate detailed balance condition requires taking square roots of the metastable state $\sqrt{\vsigma}$. These \textit{static} characterization of metastability will ultimately give us a more tangible, plug-and-play form of metastability. We refer the reader to \cite{BH16, liu2024locally} for related discussions in the classical setting. 

\subsection{The Results of this Section}

We begin with identifying a quantitative notion of detailed balance.

\begin{defn}
    [The Approximate Detailed Balance Condition ($\mathsf{ADB}$)]\label{defn:adb} 
    For a state $\vsigma$ and a jump $\vA^a$, the approximately detailed-balanced condition is defined by
    \begin{equation}
        \mathsf{ADB}_a[\vsigma] := \iint_{-\infty}^{\infty} \lnormp{\vA^a(\omega,t)\sqrt{\vsigma}-\sqrt{\vsigma}\vrho^{-\frac{1}{2}}\vA^a(\omega,t)\vrho^{\frac{1}{2}}}{2}^2  \gamma(\omega) g(t) \rd \omega \rd t,
    \end{equation}
    with $\gamma(\omega)$ the shifted Metropolis weight~\eqref{eq:Metropolis}, and $g(t)$ as in the Dirichlet form (\autoref{lem:Dicirchlet}).
\end{defn}
Alternative formulations of approximate detailed balance exist and are discussed at the end of this section. Nevertheless, we argue that the particular choice above arises naturally and interplays very nicely with other important notions throughout. Indeed, if $\vsigma$ satisfies~\autoref{defn:adb}, then the integrand must be small for each $\omega,t$ subject to the weights $\gamma(\omega),g(t).$ Therefore, $\vsigma$ and $\vrho$ are somewhat interchangeable, and mechanically, we can freely pass $\vsigma^{1/2}$ through $\vA(\omega,t),$ provided we compensate with suitable powers of $\vrho$. Although the appearance of $\vsigma^{1/2}$ may appear intimidating, it arises due to the ubiquitous presence of the KMS inner product (\autoref{defn:s_inner}).

Equipped with this definition, the main result of this section (\autoref{thm:meta_implie_ADB}) states that low entropy production implies an approximate detailed balance condition for local observables. This ultimately reinforces the intuition that metastable states behave locally like Gibbs states, as in \eqref{eq:locally_flat}. 
\begin{thm}[Slow Entropy Production implies Approximate Detailed Balance]\label{thm:meta_implie_ADB}
For any state $\vsigma$ and a Hermitian jump operator $\vA^a$, we have that
\begin{equation}
    \mathsf{ADB}_a[\vsigma] \lesssim \mathsf{FI}_a[\vsigma||\vrho] \times \L(1+ \log \bigg(\frac{\norm{\log\vsigma-\log\vrho}^2 \norm{\vA^a}^2}{\mathsf{FI}_a[\vsigma||\vrho]}\bigg)\R)
\end{equation}
\noindent with $\mathsf{FI}_a[\vsigma||\vrho]$ the relative Fisher information in~\autoref{defn:fisher}.  
\end{thm}

The essential technical statement in this section is that the physical notion of approximate detailed balance condition can be written as an entropy gradient (with the same $\nabla_a$ as in Eq.~\eqref{eq:nc_grad}, see~\autoref{lem:ADB_int_log})
    \begin{equation}\label{eq:adb_as_eg}
        \mathsf{ADB}_a[\vsigma] =\big \|\nabla_a[\log \vsigma-\log \vrho]\big\|_{\mathsf{ADB}, \vsigma}^2,
    \end{equation}
    \noindent under a suitable weighted norm $\|\cdot\|_{\mathsf{ADB}, \vsigma}$ resembling that of the Fisher Information $\|\cdot\|_{\mathsf{FI}, \vsigma}$. The proof of~\autoref{thm:meta_implie_ADB} then boils down to a comparison of the integration weights.

Having introduced several forms of metastability, we draw equivalence between all of them by establishing a converse, in that approximately detailed-balanced states are also metastable. 

\begin{thm}
    [Approximate Detailed Balance implies Metastability]\label{thm:adb_to_meta} If a state $\vsigma$ is approximately detailed-balanced under a Hermitian jump $\vA^a$, it is also locally approximately stationary:
    \begin{align}
    \norm{\CL_a[\vsigma]}_1  \lesssim \|\vA^a\|\cdot   \sqrt{\mathsf{ADB}_a[\vsigma]}.
    \end{align}
\end{thm}

In part, this should come as no surprise, due to the underlying classical intuition of approximate detailed balance in terms of likelihood ratios (recall \eqref{eq:classical_ADB}). At a technical level, \autoref{thm:adb_to_meta} captures that the operator FT $\hat{\vA}(\omega,t)$ and the filter functions $\gamma(\omega),$ $g(t)$ in~\autoref{defn:adb} parallel that of the Dirichlet form (\autoref{lem:Dicirchlet}).

\begin{rmk}
    \autoref{thm:meta_implie_ADB}, \autoref{thm:adb_to_meta} together conclude our discussion on the equivalences between definitions of metastability, as summarized in~\autoref{fig:metastab_equiv_diag}. In particular, it also allows us to conclude
    \begin{equation}
        \CL[\vsigma]\approx 0 \iff \CL_a[\vsigma]\approx 0 \text{ and } [\vH, \vsigma]\approx 0
    \end{equation}
    i.e., metastable states must also be approximately band-diagonal in the energy eigenbasis. 
\end{rmk}

We further remark that when focusing on static thermal equilibrium properties, the time-dependence in the definition of approximate detailed balance (\autoref{defn:adb}) is less relevant, and can be removed to focus only on the operator Fourier transform $\hat{\vA}(\omega)$.

\begin{lem}[ADB without time]\label{lem:meta_implies_ADB_without_t}
For any state $\vsigma$ and jump operator $\vA^a$, we have that 
     \begin{align}
         \int_{-\infty}^{\infty} \lnormp{\vA^a(\omega)\sqrt{\vsigma}-\sqrt{\vsigma}\vrho^{-\frac{1}{2}}\vA^a(\omega)\vrho^{\frac{1}{2}}}{2}^2  \gamma(\omega) \rd \omega \lesssim \mathsf{ADB}_a[\vsigma].
     \end{align}
\end{lem}

We organize the remainder of this section as follows. We begin in Section~\ref{section:adb_as_eg} by presenting the main technical tool \eqref{eq:adb_as_eg}, on the relationship between $\mathsf{ADB}$ and the entropy gradient. Subsequently, we prove in Section~\ref{section:proof_of_adb} the main result of this section, \autoref{thm:meta_implie_ADB}. The remaining sections are the additional results: the converse in \autoref{thm:adb_to_meta} (in Section~\ref{section:adb_and_metastability}), the alternative form of $\mathsf{ADB}$ in \autoref{lem:meta_implies_ADB_without_t} (in Section~\ref{section:proof_of_adb}).

\subsection{Approximate Detailed Balance as an Entropy Gradient}
\label{section:adb_as_eg}
The key idea in relating approximate detailed balance to metastability is an integral-of-logarithm identity akin to the entropy production rate (\autoref{thm:integral_square_logs}), with slightly different weights. Since the integral is over positive expressions, we expect that a state with a low entropy production rate must be approximately detailed-balanced. 

\begin{lem}[Approximate Detailed Balance as an Entropy Gradient]\label{lem:ADB_int_log}
For any Hermitian $\vA,$
  \begin{align}
     \mathsf{ADB}[\vsigma]&=\iint_{-\infty}^{\infty} \lnormp{\vA(\omega,t)\sqrt{\vsigma}-\sqrt{\vsigma}\vrho^{-\frac{1}{2}}\vA(\omega,t)\vrho^{\frac{1}{2}}}{2}^2  \gamma(\omega) g(t) \rd \omega \rd t \\
     &=   \iint_{-\infty}^{\infty} \int_{-1/2}^{1/2} \lnormp{\vsigma^{\frac{1+2s}{4}}[\vA(\omega,t),\log\vsigma-\log\vrho]\vsigma^{\frac{1-2s}{4}} }{2}^2   h_{s}(\omega) g^{\mathsf{ADB}}_{s}(t)\rd s \rd \omega \rd t \label{eq:adb_as_eg_formal}
  \end{align} 
  where $h_{s}(\omega)$ is defined in~\autoref{thm:integral_square_logs} and
  \begin{align}
      g^{\mathsf{ADB}}_{s}(t) :=  \frac{1}{2}\int^{\frac{1}{2}-\labs{s}}_{-(\frac{1}{2}-\labs{s})} g_{s'}(t) \rd s' =\frac{1}{2\pi \beta} \ln \bigg(\frac{\cosh(\frac{2\pi t}{\beta})+\cos(s\pi)}{\cosh(\frac{2\pi t}{\beta})-\cos(s\pi)}\bigg) >0.
  \end{align}
\end{lem}

The expression in the RHS of \eqref{eq:adb_as_eg_formal} is a convex combination of $\vsigma$-weighted $s$-norms of the gradients $\nabla_{a,\omega, t}$ of the entropy differential, as suggested in \eqref{eq:adb_as_eg}.

\begin{rmk}
    The time-filter $g^{\mathsf{ADB}}_{s}(t)$ is normalized in the following sense:
    \begin{align}
    \int_{-\infty}^{\infty} g^{\mathsf{ADB}}_{s}(t) \rd t = \frac{1}{2}\big(\frac{1}{2}-\labs{s}\big) \quad \text{for each}\quad \labs{s}\le 1/2\label{eq:int_gADB}.
    \end{align}
\end{rmk}
    \begin{proof}
    
        [of \autoref{lem:ADB_int_log}] We begin by exhibiting the entropy gradient, similarly to \autoref{lem:expose_log}. By~\autoref{lem:double_Bohr}, for any $\vA,$
\begin{align}
    \vA\sqrt{\vsigma}-\sqrt{\vsigma}\vrho^{-\frac{1}{2}}\vA\vrho^{\frac{1}{2}} &= - \sum_{\nu\in B(\vH)}\sum_{\mu\in B(\vH_{eff})} (\vA_{\nu})_{\mu}(e^{\frac{\beta(\nu-\mu)}{2}}-1)\sqrt{\vsigma}\\
    &= -\int_0^{1/2} \sum_{\nu\in B(\vH)}\sum_{\mu\in B(\vH_{eff})} (\vA_{\nu})_{\mu}\beta(\nu-\mu) e^{s\beta(\nu-\mu)}\sqrt{\vsigma} \rd s\\
    &= -\int_0^{1/2} \vsigma^{s_1}[\vrho^{-s_1}\vA\vrho^{s_1}, \log\vsigma-\log\vrho]\vsigma^{\frac{1}{2}-s_1}
\end{align}
using that $e^{\frac{\beta x}{2}}-1= \int_0^{1/2} e^{s\beta x}\beta x\rd s.$  Next, we plug the above into the integral expression for $ \mathsf{ADB}[\vsigma]$. With $s_+=s_1+s_2$, and $\gamma(\omega)=h_{-1/2}(\omega)$, we have
\begin{align}
     \mathsf{ADB}[\vsigma] &= \iint_{-\infty}^{\infty} \lnormp{\vA(\omega,t)\sqrt{\vsigma}-\sqrt{\vsigma}\vrho^{-\frac{1}{2}}\vA(\omega,t)\vrho^{\frac{1}{2}}}{2}^2  \gamma(\omega) g(t) \rd \omega \rd t \\
    &=\iint_0^{1/2} \tr\L[\vsigma^{s_+}\L[\vrho^{-s_1}\vA(\omega,t)\vrho^{s_1}, \log\vsigma-\log\vrho\R]\vsigma^{1-s_+} \L[\vrho^{-s_2}\vA(\omega,t)\vrho^{s_2}, \log\vsigma-\log\vrho\R]^{\dagger} \R]  \\  & \quad \quad\quad \quad\quad \quad\quad \quad\quad \quad\quad \times \gamma(\omega) g(t) \rd s_1\rd s_2  \rd \omega \rd t = \\
    &= \iint_{-\infty}^{\infty} \iint_0^{1/2}  \lnormp{\vsigma^{\frac{s_+}{2}}\L[\vA(\omega+s_+\beta \sigma^2,t),\log\vsigma-\log\vrho\R]\vsigma^{\frac{1-s_+}{2}}}{2}^2   
 \\ &\quad \quad\quad \quad\quad \quad\quad \quad\quad \quad\quad\times \exp\bigg(\frac{s_+\beta(2\omega+s_+\beta\sigma^2)}{2}\bigg)
    \gamma(\omega) g_{s_2-s_1}(t) \rd s_1 \rd s_2  \rd \omega \rd t \label{eq:s1s2}
\end{align}
where we applied the symmetrization statement in \autoref{lem:symmetrize_ArhoArhos1s2} and a similar argument as~\eqref{eq:getting_gs}. We remark that the above expression depends only on the sum and difference of $s_1, s_2$, which suggests that one should rewrite the above via the change-of-variables $s = s_1+s_2-1/2,$ and $s_1-s_2 = s'$. To do so, we first note that for any two functions $f_+, f_-$, the identity
\begin{align}
    \iint_0^{1/2} f_+(s_1+s_2) f_-(s_1-s_2) \rd s_1 \rd s_2 &= \frac{1}{2}\iint f_+(s_1+s_2) f_-(s_1-s_2)  \indicator(\labs{s}+\labs{s'}\le \frac{1}{2})\rd s\rd s'\\
    &= \frac{1}{2}\int_{-1/2}^{1/2} f_+(s + \frac{1}{2}) \int_{-(\frac{1}{2}-\labs{s})}^{\frac{1}{2}-\labs{s}} f_-(s')  \rd s' \rd s 
\end{align}
enables us to express the double integral over $s_1, s_2$ in terms of $s, s'$. Applied to \eqref{eq:s1s2},
\begin{align}
     \eqref{eq:s1s2} =\iint_{-\infty}^{\infty} \int_{-1/2}^{1/2} & 
 \lnormp{\vsigma^{\frac{1+2s}{4}}\L[\vA(\omega+s_+\beta \sigma^2,t),\log\vsigma-\log\vrho\R]\vsigma^{\frac{1-2s}{4}} }{2}^2 \times  \\ \times &\exp\bigg(\frac{s_+\beta(2\omega+s_+\beta\sigma^2)}{2}\bigg) \gamma(\omega) 
 \times \int_{-(\frac{1}{2}-\labs{s})}^{\frac{1}{2}-\labs{s}}  g_{s'}(t) \rd s'   \rd s \rd \omega \rd t.\label{eq:int_gs'}
\end{align}
Under the change of the variables $\omega \rightarrow \omega - s_+\beta \sigma^2 $, we finally are able to exhibit the entropy gradients: 
\begin{align}
    \eqref{eq:int_gs'} =\iint_{-\infty}^{\infty} \int_{-1/2}^{1/2} \lnorm{\nabla_{a, \omega, t}[\log\vsigma-\log\vrho]}_{\vsigma, s}^2 & \cdot \exp\bigg(\frac{s_+\beta(2\omega-s_+\beta\sigma^2)}{2}\bigg)\gamma(\omega-s_+\beta\sigma^2) \\ & \times \int_{-(\frac{1}{2}-\labs{s})}^{\frac{1}{2}-\labs{s}}  \frac{g_{s'}(t)}{2} \rd s'  \rd s   \rd \omega \rd t.
\end{align}

\noindent Finally, to obtain the advertised form, we evaluate the integral over $g_{s'}(t)$
\begin{align}
    \int^{r}_{-r} \frac{g_{s'}(t)}{2}\rd s' &= \int^{r}_{-r} \frac{1}{\beta} \frac{\cos(s'\pi)\cosh(\frac{2\pi t}{\beta})}{\cosh(\frac{4\pi t}{\beta})+\cos(2s'\pi)} \rd s' \\
    &= \frac{1}{4\pi}\int^{r}_{-r}\int^{\infty}_{-\infty}\cos(\nu t) \frac{\cosh(s'\beta \nu/2)}{2\cosh(\beta\nu/4)} \rd \nu \rd s'\\
    &=\frac{1}{4\pi}\int^{\infty}_{-\infty}\frac{2}{ \beta }\cos(\nu t) \frac{\sinh(r\beta \nu/2)}{\nu \cosh(\beta\nu/4)} \rd \nu  = \frac{1}{2\pi \beta} \ln \frac{\cosh(\frac{2\pi t}{\beta})+\sin(r \pi)}{\cosh(\frac{2\pi t}{\beta})-\sin(r \pi)} =:g^{\mathsf{ADB}}_{\frac{1}{2}-r}(t),
\end{align}
using the integral formula in \autoref{lem:cossinh_vcosh}. Expose $h_s$ by a direct calculation (\autoref{lem:shift_hs}) to conclude the proof.

\end{proof}

We collect the deferred lemmas used in the proof above here. 

\begin{lem}[{\cite[Eq. 4.114.2]{gradshteyn2014table}}]\label{lem:cossinh_vcosh}
\begin{align}
    \frac{1}{2\pi}\int_{-\infty}^{\infty} \frac{\cos(\nu t)}{\nu} \frac{\sinh(b\nu)}{\cosh(c\nu)} \rd \nu = \frac{1}{2\pi} \ln \frac{\cosh(\frac{\pi t}{2c})+\sin(\frac{b\pi}{2c})}{\cosh(\frac{\pi t}{2c})-\sin(\frac{b\pi}{2c})} \quad \text{for each}\quad \labs{\Re(b)} < \Re(c),\ t\in \BR.
\end{align}
\end{lem}

\begin{lem}\label{lem:shift_hs}
Let $h_s(\omega)$ as in~\autoref{thm:integral_square_logs} and $s_+ = s+1/2$. Then,
\begin{align}
   h_s(\omega)= \exp\bigg(\frac{s_+\beta(2\omega-s_+\beta\sigma^2)}{2}\bigg) \gamma(\omega-s_+\beta\sigma^2).
\end{align}
\end{lem}

\begin{proof} Recall $\gamma(\omega) = h_{-\frac{1}{2}}(\omega)=\exp\big(\frac{-\beta(2\omega+\beta\sigma^2/2)}{4}\big)h(\omega+\beta\sigma^2/2)$ (\autoref{lem:sup_h} and \autoref{thm:integral_square_logs}). Then,
    \begin{align}
     &\exp\bigg[\frac{s_+\beta(2\omega-s_+\beta\sigma^2)}{2} \bigg]  \gamma(\omega-s_+\beta\sigma^2) = \\
     =&\exp\bigg[\frac{s_+\beta(2\omega-s_+\beta\sigma^2)}{2} - \frac{\beta(2\omega+(1/2-2s_+)\beta\sigma^2)}{4}\bigg] h(\omega+ (1/2-s_+)\beta\sigma^2)\\
     =& \exp\bigg[\frac{(s+1/2)\beta(2\omega-(s+1/2)\beta\sigma^2)}{2} - \frac{\beta(2\omega-(2s-\frac{1}{2})\beta\sigma^2)}{4}\bigg] \cdot h(\omega -s\beta\sigma^2)\\
     =& \exp\bigg[\frac{s\beta (2\omega-s\beta\sigma^2)}{2}\bigg] h(\omega -s\beta\sigma^2) = h_s(\omega).
\end{align}
\end{proof}

\subsubsection{Metastability implies ADB -- Proof of \autoref{thm:meta_implie_ADB}}
\label{section:proof_of_adb}

To show \autoref{thm:meta_implie_ADB}, we compare the scalar weights in the integral-of-log expression between approximate detailed balance (\autoref{lem:ADB_int_log}) and the entropy production rate (\autoref{thm:integral_square_logs}). The comparison boils down to the ratio between their filter functions in time, ${g^{\mathsf{ADB}}_{s}(t)}/{g_{s}(t)}$, which unfortunately, diverges near $s\rightarrow 0.$ Therefore, we introduce a cut-off $s_0$ and split the integral in $\mathsf{ADB}$ into two regimes:
    \begin{align}
        &\iint_{-\infty}^{\infty} \lnormp{\vA^a(\omega,t)\sqrt{\vsigma}-\sqrt{\vsigma}\vrho^{-\frac{1}{2}}\vA^a(\omega,t)\vrho^{\frac{1}{2}}}{2}^2  \gamma(\omega) g(t) \rd \omega \rd t \\
     &= \iint_{-\infty}^{\infty} \bigg(\undersetbrace{\text{compare}}{\int_{s_0\le\labs{s}\le 1/2}} + \undersetbrace{\text{truncate}}{\int_{\labs{s}\le s_0}}\bigg) \lnorm{\nabla_{a, \omega, t}[\log\vsigma-\log\vrho]}_{\vsigma, s}^2  h_{s}(\omega)  g^{\mathsf{ADB}}_{s}(t)\rd s \rd \omega \rd t. \label{eq:split_s0}
    \end{align}

Near the singularity $(s\rightarrow 0)$, we will need to leverage the following a priori bound.

\begin{lem}[A priori norm bounds]\label{lem:apriori_bounds} For any $\vA, \vO,\vO',$ quantum state $\vsigma$ and $s\in[-1/2,1/2],$ and any function $h_s(\omega),$ 
    \begin{align}
        \labs{\int_{-\infty}^{\infty} \tr\L[[\vA(\omega),\vO] \vsigma^{\frac{1}{2}+s} \L[\vA(\omega),\vO'\R]^{\dagger} \vsigma^{\frac{1}{2}-s}
         \R] h_s(\omega)\rd \omega} \le 8\norm{\vO} \norm{\vO'} \norm{\vA}^2 \sup_{\omega}\labs{h_s(\omega)}.
    \end{align}
\end{lem}

Whose proof for clarity is deferred to the next subsection, Section \ref{subsection:deferred_proofs_apriori}.

\begin{proof}

[of~\autoref{thm:meta_implie_ADB}] 
Consider a tunable cut-off $s_0\in[0,1/4].$ We rescale by $\norm{\vA^a}=1.$\\

\noindent \textbf{Case I: The large $s$ regime.} In the part of the integration where the ratio is bounded -- i.e. large $s$ -- we can simply leverage that upper bound on the ratio to compare the entropy gradient expression for $\mathsf{ADB}$, with the entropy gradient expression for the Fisher Information $\mathsf{FI}$:
    \begin{align}
    \iint_{-\infty}^{\infty} \int_{s_0\le\labs{s}\le 1/2} \|\nabla_{a, \omega, t}[\log\vsigma-\log\vrho]\|_{\vsigma, s}^2 h_{s}(\omega) g^{\mathsf{ADB}}_{s}(t)\rd s \rd \omega \rd t \le \sup_{t,s_0\le\labs{s}\le 1/2} \bigg(\frac{g^{\mathsf{ADB}}_{s}(t)}{g_{s}(t)}\bigg) \times \mathsf{FI}_a[\vsigma||\vrho].
    \end{align}
Now, we explicitly compute the supremum of said ratio, as follows. For $s\in(-1/2,1/2), t \in \BR$,
\begin{align}
    \frac{g^{\mathsf{ADB}}_{s}(t)}{g_{s}(t)} &= \frac{\beta}{2} \frac{\cosh(\frac{4\pi t}{\beta})+\cos(2s\pi)}{\cos(s\pi)\cosh(\frac{2\pi t}{\beta})} \cdot \frac{1}{2\pi \beta} \ln \frac{\cosh(\frac{2\pi t}{\beta})+\cos(s\pi)}{\cosh(\frac{2\pi t}{\beta})-\cos(s\pi)}\\
    &=\frac{1}{4\pi} \frac{2+\frac{\cos(2s\pi)-1}{\cosh(\frac{2\pi t}{\beta})^2}}{x} \cdot \ln \bigg(\frac{1+x}{1-x}\bigg)  \quad \text{ using }\quad \cosh(2z) = 2\cosh(z)^2-1.
    \\ &\le \frac{3}{4\pi} \frac{1}{x} \cdot \ln \frac{1+x}{1-x} \quad \quad\quad \quad \quad\quad \quad \quad \text{where}\quad x :=\frac{\cos(s\pi)}{\cosh(2\pi t/\beta)} \in [0,1].
\end{align}
Since $\frac{1}{x} \cdot \ln \frac{1+x}{1-x}$ is increasing in $x\in [0, 1],$
\begin{align}
    \sup_{s_0\le \labs{s}\le 1/2, t\in \BR} \frac{g^{\mathsf{ADB}}_{s}(t)}{g_{s}(t)} &\le \sup_{x\in [0,\cos(s_0\pi)]} \frac{1}{4\pi} \frac{3}{x} \cdot \ln \frac{1+x}{1-x}, \quad \text{ using}\quad \frac{\cos(s\pi)}{\cosh(2\pi t/\beta)} \le \cos(s\pi) \text{ for }|s|\leq \frac{1}{2} \\&  = \frac{3}{4\pi} \frac{1}{\cos s_0\pi}\cdot \log\bigg(\frac{1+\cos(s_0\pi)}{1-\cos(s_0\pi)}\bigg) \\
     &\le \frac{3\sqrt{2}}{4\pi} \cdot 2\cdot \log\big(\frac{2}{s_0\pi}\big) \lesssim \ln \bigg(\frac{1}{s_0}\bigg).
\end{align}
The last line uses that in the regime $|s_0|\leq \frac{1}{4}$, we have $\frac{1}{\sqrt{2}}\leq \cos s_0\pi \leq 1 - \frac{s_0^2\pi^2}{2}$.

\noindent \textbf{Case II: The small $s\rightarrow 0$ regime.} Near the divergence, we instead appeal to the a priori bound 
    \begin{align}
        &\iint_{-\infty}^{\infty} \int_{\labs{s}\le s_0} \|\nabla_{a, \omega, t}[\log\vsigma-\log\vrho]\|_{\vsigma, s}^2 h_{s}(\omega) g^{\mathsf{ADB}}_{s}(t)\rd s \rd \omega \rd t\\
        & \lesssim \| \log\vsigma-\log\vrho\|^2 \cdot \int_{-\infty}^{\infty} \rd t \int_{\labs{s}\le s_0} \rd s \cdot g^{\mathsf{ADB}}_{s}(t) \cdot \sup_\omega h_s(\omega) \tag*{(\autoref{lem:apriori_bounds})} \\
        & \leq \| \log\vsigma-\log\vrho\|^2 \cdot  \int_{-\infty}^{\infty} \rd t \int_{\labs{s}\le s_0} \rd s \cdot g^{\mathsf{ADB}}_{s}(t) \tag*{(\autoref{lem:sup_h})} \\
        & \leq \| \log\vsigma-\log\vrho\|^2 \cdot   \int_{\labs{s}\le s_0} \rd s\cdot  \bigg(\frac{1}{2}-\labs{s}\bigg) \tag*{By~\eqref{eq:int_gADB}} \\
        &\lesssim \| \log\vsigma-\log\vrho\|^2 \cdot s_0.
    \end{align}

    \noindent \textbf{Combining the Regimes. } Combining the two bounds, we arrive at:
    \begin{align}
        \eqref{eq:split_s0} &\lesssim \mathsf{FI}_a[\vsigma||\vrho]\times \log\bigg(\frac{1}{s_0}\bigg) +  s_0 \times \norm{\log\vsigma-\log\vrho}^2.
    \end{align}
    Choose
    \begin{align}
        s_0 = \min\L(\frac{\mathsf{FI}_a[\vsigma||\vrho]}{\norm{\log\vsigma-\log\vrho}^2 },1/4\R)
    \end{align}
    such that 
    \begin{align}
        \log\bigg(\frac{1}{s_0}\bigg) =\max\L(\log(\frac{\norm{\log\vsigma-\log\vrho}^2 }{\mathsf{FI}_a[\vsigma||\vrho]}),\log(4)\R) \lesssim \ln \bigg(\frac{\norm{\log\vsigma-\log\vrho}^2 }{\mathsf{FI}_a[\vsigma||\vrho]}\bigg),\quad\text{and}\quad s_0 \lesssim \frac{\mathsf{FI}_a[\vsigma||\vrho]}{\norm{\log\vsigma-\log\vrho}^2},
    \end{align}
    using the apriori bound that $\mathsf{FI}_a[\vsigma||\vrho] \le \frac{1}{2}\norm{\log\vsigma-\log\vrho}^2$. 
    Restore the factors of $\norm{\vA^a}$ to conclude the proof. 
\end{proof}

\subsubsection{Deferred Proofs of A Priori Bounds}
\label{subsection:deferred_proofs_apriori}

We dedicate this section to the deferred proofs of the a priori bounds. 
\begin{proof} 

[of \autoref{lem:apriori_bounds}]
Let $\vX_{\omega} = [\vA(\omega),\vO]$ and $\vY_{\omega} = [\vA(\omega),\vO']$, then for any $s\in[-1/2,1/2],$ we apply a Cauchy-Schwarz
\begin{align}
        \labs{\int_{-\infty}^{\infty}\tr\L[\vX_{\omega} \vsigma^{\frac{1}{2}+s} \vY_{\omega}^{\dagger} \vsigma^{\frac{1}{2}-s}\R] h_s(\omega) \rd \omega}^2 \le \int_{-\infty}^{\infty} &\tr\L[\vX_{\omega} \vsigma^{\frac{1}{2}+s} \vX_{\omega}^{\dagger} \vsigma^{\frac{1}{2}-s}\R] h_s(\omega)  \rd \omega \times \\ &\times  \int_{-\infty}^{\infty}\tr\L[\vY_{\omega} \vsigma^{\frac{1}{2}+s} \vY_{\omega}^{\dagger} \vsigma^{\frac{1}{2}-s}\R]h_s(\omega) \rd \omega 
\end{align}
By point-wise applying~\autoref{lem:generalized_ALT}, and then the AM-GM inequality, to each of the terms on the RHS above:
\begin{align}
    \int_{-\infty}^{\infty}\tr \L[\vX_{\omega} \vsigma^{\frac{1}{2}+s} \vX_{\omega}^{\dagger} \vsigma^{\frac{1}{2}-s}\R] h_s(\omega)  \rd \omega &\leq \int_{-\infty}^{\infty}\tr[\vX_{\omega}^{\dagger}\vX_{\omega}\vsigma]^{\frac{1}{2}+s}  \tr[\vX_{\omega}\vX_{\omega}^{\dagger}\vsigma]^{\frac{1}{2}-s} h_s(\omega) \rd \omega  \\
    &\leq \int_{-\infty}^{\infty}\big(\tr[\vX_{\omega}^{\dagger}\vX_{\omega}\vsigma] + \tr[\vX_{\omega}\vX_{\omega}^{\dagger}\vsigma]\big) h_s(\omega)\rd \omega
\end{align}
    Although it would be straightforward to bound the integral above in terms of norms of $\vA_\omega$, our actual goal is to exhibit norms of $\vA$. To do so, we introduce purifications: 
\begin{align}
    \vV &:= \int_{-\infty}^{\infty} \hat{\vA}(\omega) \otimes \bra{\omega}\  \rd \omega \quad \text{where}\quad \braket{\omega|\omega'} = \delta(\omega-\omega')\\
    \vW &:= \int_{-\infty}^{\infty}\ket{\omega}\bra{\omega} h_s(\omega) \rd \omega.
\end{align}
We can now expand the integrals above in terms of the purifications, 
\begin{align}
    &\int_{-\infty}^{\infty}( \tr[\vX_{\omega}\vX_{\omega}^{\dagger}\vsigma]) h_s(\omega)\rd \omega =  \int_{-\infty}^{\infty}\tr\L[[\vA(\omega),\vO]\L[\vA(\omega),\vO\R]^{\dagger} \vsigma
         \R] h_s(\omega) \rd \omega =  \\
        &=\tr\L[ \vV(\vO \otimes \vW)\vV^{\dagger}\vO^{\dagger}\vsigma -\vV (\vO \vO^{\dagger}\otimes \vW) \vV^{\dagger}\vsigma +\vO \vV ( \vO^{\dagger}\otimes \vW) \vV^{\dagger}\vsigma - \vO\vV( \vI\otimes \vW)\vV^{\dagger}\vO^{\dagger}\vsigma
         \R] \\
         &\le 4\norm{\vO}^2\norm{\vV}^2 \norm{\vW} \le 4\norm{\vO}^2 \norm{\vA}^2\sup_{\omega} \undersetbrace{\le 1}{\labs{h_s(\omega)}} \tag*{(By~\autoref{lem:OParseval})},
\end{align}
which gives us the claimed result. 
\end{proof}

\begin{lem}[{\cite[Lemma 4.3]{bosboom2024unique}}]\label{lem:generalized_ALT}
For any Hermitian $\vH_1, \vH_2$ and operator $\vX,$ 
\begin{align}
    \labs{\tr[\vX \vH_1 \vX^{\dagger}\vH_2]} \le \tr[\vX^{\dagger}\vX\labs{\vH_1}^p]^{1/p}\cdot \tr[\vX\vX^{\dagger}\labs{\vH_2}^q]^{1/q}
\end{align}
whenever $1/p+ 1/q =1$ for $1 \le q ,p\le \infty.$
\end{lem}

\subsection{The Converse: Approximate Detailed Balance implies Metastability}
\label{section:adb_and_metastability}

We dedicate this section to the proof of the converse implication, that if $\vsigma$ satisfies approximate detailed balance then it also is metastable.

\begin{proof} 

[of~\autoref{thm:adb_to_meta}] For any jump $\vA^a\equiv \vA$ and any test operator $\vO$ such that $\norm{\vO}\le1,$ one can express the time-derivative of its expectation value via the Dirichlet form:
    \begin{align}
    \tr[-\CL^{\dagger}[\vO]\vsigma] &=\iint_{-\infty}^{\infty} \tr\L[\sqrt{\vrho}[\vA(\omega,t),\vO]\sqrt{\vrho}[\vA(\omega,t),\frac{1}{\sqrt{\vrho}}\vsigma\frac{1}{\sqrt{\vrho}}]^{\dagger}\R] h(\omega)g(t) \rd \omega \rd t \tag*{(Recall~\eqref{eq:LOsigma_Dirichlet})}\\
    &= \iint_{-\infty}^{\infty} \tr\L[[\vA(\omega,t),\vO]\L(\vrho^{\frac{1}{2}}\vA(\omega,t)\vrho^{-\frac{1}{2}}\vsigma-\vsigma\vrho^{-\frac{1}{2}}\vA(\omega,t)\vrho^{\frac{1}{2}}\R)^{\dagger}\R] h(\omega)g(t) \rd \omega \rd t \label{eq:LsigmaO_KMS}.
\end{align}
The above argument resembles the approximate detailed balance condition; telescope the differences to get
\begin{align}\label{eq:adb-to-meta-telescope}
    \vrho^{\frac{1}{2}}\vA(\omega,t)\vrho^{-\frac{1}{2}}\vsigma - \vsigma\vrho^{-\frac{1}{2}}\vA(\omega,t)\vrho^{\frac{1}{2}}
    =& \L(\vrho^{\frac{1}{2}}\vA(\omega,t)\vrho^{-\frac{1}{2}}\sqrt{\vsigma} - \sqrt{\vsigma} \vA(\omega,t)\R)\sqrt{\vsigma} \\&+ \sqrt{\vsigma}\L(\vA(\omega,t) \sqrt{\vsigma}-\sqrt{\vsigma}\vrho^{-\frac{1}{2}}\vA(\omega,t)\vrho^{\frac{1}{2}}\R)
\end{align}
and bound individual terms. For the first term,
\begin{align*}
    &\labs{\iint_{-\infty}^{\infty} \tr\L[[\vA(\omega,t),\vO]\L((\vrho^{\frac{1}{2}}\vA(\omega,t)\vrho^{-\frac{1}{2}}\sqrt{\vsigma} - \sqrt{\vsigma} \vA(\omega,t))\sqrt{\vsigma}\R)^{\dagger}\R] h(\omega)g(t) \rd \omega \rd t}\\
    =&\labs{\iint_{-\infty}^{\infty} \tr\L[[\vA(\omega,t),\vO]\sqrt{\vsigma} \L(\sqrt{\vsigma}\vrho^{-\frac{1}{2}}\vA(\omega,t)^{\dagger}\vrho^{\frac{1}{2}} -  \vA(\omega,t)^{\dagger}\sqrt{\vsigma}\R)\R] h(\omega)g(t) \rd \omega \rd t}\\
    \le& |\tr[\vU\vV\vO\sqrt{\vsigma}]| + |\tr[\vU(\vO\otimes \vI)\vV\sqrt{\vsigma}]| \leq  2\norm{\vU}_2\cdot \norm{\vO} \cdot  \norm{\vV}\cdot  \norm{\sqrt{\vsigma}}_2 \leq 2\norm{\vU}_2\cdot\norm{\vV}
\end{align*}
where we introduced the purifications of the weighted jump operators and ``ADB-like" operators:
\begin{align}
    \vV&:= \iint_{-\infty}^{\infty} \sqrt{\frac{h(\omega)^2}{\gamma(-\omega)}g(t)}\vA(\omega,t) \otimes \ket{\omega,t}\rd \omega \rd t,\\
    \vU&:= \iint_{-\infty}^{\infty} \sqrt{\gamma(-\omega)g(t)} (\sqrt{\vsigma}\vrho^{-\frac{1}{2}}\vA(\omega,t)^{\dagger}\vrho^{\frac{1}{2}} -  \vA(\omega,t)^{\dagger}\sqrt{\vsigma} )\otimes \bra{\omega,t} \rd \omega \rd t.
\end{align}
We proceed by first evaluating the norm of the weighted jump operators:
\begin{align}
    \norm{\vV}^2 &= \lnorm{  \iint_{-\infty}^{\infty} \vA^{a}(\omega,t)^{\dagger}\vA^{a}(\omega,t) \frac{h(\omega)^2}{\gamma(-\omega)}g(t) \rd \omega \rd t}  \\
    &\leq \sup_\omega\frac{h(\omega)^2}{\gamma(-\omega)} \cdot \int_{-\infty}^{\infty}\rd t \cdot  g(t)\lnorm{  \int_{-\infty}^{\infty} \vA^{a}(\omega,t)^{\dagger}\vA^{a}(\omega,t)  \rd \omega } \tag*{(Triangle Ineq.)} \\
    & \le \norm{\vA}^2 \sup_\omega\frac{h(\omega)^2}{\gamma(-\omega)} \int_{-\infty}^{\infty} g(t)\rd t \tag*{(operator Parseval's, \autoref{lem:OParseval})}
    \\ &\le \frac{\norm{\vA}^2}{2} e^{\beta^2\sigma^2/4} \tag*{(Eq.~\eqref{eq:int_g})}
\end{align}
\noindent where in the last line we also used the following bound:
\begin{align}
    \frac{h(\omega)^2}{\gamma(-\omega)} &=  \frac{\exp\big(-\beta\labs{\omega} -\frac{\sigma^2\beta^2}{4} \big)}{\exp\big(-\beta\max(-\omega+\beta\sigma^2/2,0)\big)} \\&= e^{-\sigma^2\beta^2/4}\cdot \exp\bigg[\beta\bigg(-|\omega| + \max(-\omega+\beta\sigma^2/2, 0)\bigg)\bigg] \\
    & \leq  e^{-\sigma^2\beta^2/4}\cdot e^{\sigma^2\beta^2/2} \leq e^{\sigma^2\beta^2/4}. \quad \text{(Case enumeration)}
\end{align}

\noindent Further, the 2-norm of $\vU$ recovers exactly the form of approximate detailed balance:
\begin{align}
    \norm{\vU}^2_2 = \tr[\vU^\dagger \vU] &= \iint_{-\infty}^{\infty} \gamma(-\omega)g(t) \lnormp{\sqrt{\vsigma}\vrho^{-\frac{1}{2}}\vA(\omega,t)^{\dagger}\vrho^{\frac{1}{2}} -  \vA(\omega,t)^{\dagger}\sqrt{\vsigma}}{2}^2 \rd \omega \rd t\\
    &=\iint_{-\infty}^{\infty} \gamma(\omega)g(t) \lnormp{\sqrt{\vsigma}\vrho^{-\frac{1}{2}}\vA(\omega,t)\vrho^{\frac{1}{2}} -  \vA(\omega,t)\sqrt{\vsigma}}{2}^2 \rd \omega \rd t = \mathsf{ADB}[\vsigma]
\end{align}
where at last we used $\vA(\omega,t)^{\dagger}=\vA(-\omega,t)$ and $\omega\rightarrow-\omega$. The second term which arises from the telescoping equation \eqref{eq:adb-to-meta-telescope} is analogous. Simplify the prefactors by the assumption $\sigma \le 1/\beta$ to conclude the proof.
\end{proof}

\subsection{Alternative to the ADB condition -- Proof of~\autoref{lem:meta_implies_ADB_without_t}}
\label{section:proof-of-adb-without-t}

In this section, we prove \autoref{lem:meta_implies_ADB_without_t}, which makes sense of the approximate detailed balance condition without the time evolution. We remark that in the case of GNS detailed-balanced Lindbladians, this is the expression that would naturally arise as the approximate detailed balance condition, since the time integral drops completely. 
\begin{align}
         & \int_{-\infty}^{\infty} \lnormp{\vA^a(\omega)\sqrt{\vsigma}-\sqrt{\vsigma}\vrho^{-\frac{1}{2}}\vA^a(\omega)\vrho^{\frac{1}{2}}}{2}^2  \gamma(\omega) \rd \omega \lesssim  \\
         &  \iint_{-\infty}^{\infty} \lnormp{\vA^a(\omega, t)\sqrt{\vsigma}-\sqrt{\vsigma}\vrho^{-\frac{1}{2}}\vA^a(\omega, t)\vrho^{\frac{1}{2}}}{2}^2  \gamma(\omega) g(t) \rd \omega \rd t = \mathsf{ADB}_a[\vsigma].
     \end{align}

We refer the reader back to Section~\ref{sec:OFT_appendix} for basic properties of the operator Fourier transform. Recall that the operator FT of an operator $\vA\equiv \vA_\sigma$ is defined by a Gaussian time-filter function $f_\sigma(t)$ with a specific width $\sigma$. In our analysis in this subsection, it will become relevant to relate operator FTs of different widths. 

The strategy in our proof of \autoref{lem:meta_implies_ADB_without_t}, will be to compare and replace the time-weight $g(t)$ which arises in the $\mathsf{ADB}$ condition, to a Gaussian filter function $f_\sigma(t)$. Subsequently, we reason that the integrals over $\vA_\sigma(\omega, t)$ with Gaussian filter functions in time, ultimately just have the effect of pinching the Gaussian width $\sigma\rightarrow \sigma'$ to a smaller width, resulting in an expression akin to \autoref{lem:meta_implies_ADB_without_t}. To make this concrete, we first briefly state (and defer the proof of) two convolution properties of operator FTs. 

\begin{lem}[Convolving an Operator FT]\label{lem:A_convolve_indentity} For any $\sigma_3^2 = \sigma_1^2+\sigma_2^2,$ $\omega\in \mathbb{R}$, and operator $\vA$
    \begin{align}
        \vA_{\sigma_3}(\omega)
        &= \frac{\sqrt{\sigma_3\sqrt{\pi/2}}}{\sqrt{\sigma_1\sigma_2}}\int_{-\infty}^{\infty} \vA_{\sigma_1}(\omega') \hat{f}_{\sigma_2}(\omega-\omega')\rd \omega'.
    \end{align}
\end{lem}

\begin{lem}[Twirling operator FTs]\label{lem:OFT_Guassian_time_avg}
For any $\sigma_1,\sigma_2>0,$ $\omega\in \BR,$ and operator $\vA,$
\begin{align}
    \int_{-\infty}^{\infty} \hat{\vA}_{\sigma_1}(\omega,t) \otimes \hat{\vA}_{\sigma_1}(\omega,t)^{\dagger} \cdot f_{\sigma_2}(t)^2 \rd t = \int_{-\infty}^{\infty} \hat{f}_{\sigma_2}(\omega-\omega')^2 \hat{\vA}_{\sigma_3}(\omega') \otimes  \hat{\vA}_{\sigma_3}(\omega')^{\dagger} \rd \omega'
\end{align}
where $1/\sigma_3^2=1/\sigma_1^2+ 1/\sigma_2^2.$
\end{lem}

We now prove \autoref{lem:meta_implies_ADB_without_t}.

\begin{proof}

[of~\autoref{lem:meta_implies_ADB_without_t}] To begin, let us compare the time-filter $g(t)$ from the Dirichlet form to a Gaussian weight, with a suitable uncertainty $\sigma_0 = \frac{1}{\beta}$. 
\begin{align}
        \frac{f_{\sigma_0}(t)^2}{g(t)} = {\beta \cosh(2\pi t/\beta)\cdot \sigma \sqrt{2/\pi}e^{-2\sigma_0^2t^2}} \le\beta \sigma_0 \sqrt{2/\pi} e^{\pi^2/2\beta^2\sigma_0^2}= \sqrt{2/\pi} e^{\pi^2/2} = (const.) >0
    \end{align}
\noindent where we used $e^{bx-ax^2} = e^{-a(x-b/2a)^2 +b^2/4a}\le e^{b^2/4a}$. This suggests that we can effectively replace $g(t)$ with $f_{\sigma}(t)^2$ in the $\mathsf{ADB}$ definition, to give a closed-form expression for the time integral 
\begin{align}
\mathsf{ADB}[\vsigma] &= \iint_{-\infty}^{\infty} \lnormp{\vA_\sigma(\omega,t)\sqrt{\vsigma}-\sqrt{\vsigma}\vrho^{-\frac{1}{2}}\vA_\sigma(\omega,t)\vrho^{\frac{1}{2}}}{2}^2  \gamma(\omega) g(t) \rd \omega \rd t   \\
    &\gtrsim \iint_{-\infty}^{\infty} \lnormp{\vA_\sigma(\omega,t)\sqrt{\vsigma}-\sqrt{\vsigma}\vrho^{-\frac{1}{2}}\vA_\sigma(\omega,t)\vrho^{\frac{1}{2}}}{2}^2  \gamma(\omega) f_{\sigma_0}(t)^2 \rd \omega \rd t \\
    &= \iint_{-\infty}^{\infty} \lnormp{\vA_{\sigma'}(\omega')\sqrt{\vsigma}-\sqrt{\vsigma}\vrho^{-\frac{1}{2}}\vA_{\sigma'}(\omega')\vrho^{\frac{1}{2}}}{2}^2  \gamma(\omega) \hat{f}_{\sigma_0}(\omega-\omega')^2 \rd \omega'\rd \omega. \label{eq:gamma_convolve_gaussian}
\end{align}

\noindent where, in the last equality, we applied the twirling property \autoref{lem:OFT_Guassian_time_avg}. This is almost what we want, except that the operator Fourier transforms have been ``pinched'' to a smaller uncertainty:
\begin{align}
\bigg(\frac{1}{\sigma'}\bigg)^2 = \bigg(\frac{1}{\sigma}\bigg)^2 + \bigg(\frac{1}{\sigma_0}\bigg)^2 = \frac{1}{\sigma^2} + \beta^2    
\end{align}
To recover the desired result, we next try to point-wise express $\vA_{\sigma}(\omega)$ in terms of $\vA_{\sigma'}(\omega)$ using the convolution property (\autoref{lem:A_convolve_indentity}). Recall we assume $\sigma \leq \beta^{-1}$. For each $\omega\in \BR,$  
\begin{align}
    &\big\|\vA_{\sigma}(\omega)\sqrt{\vsigma}-\sqrt{\vsigma}\vrho^{-\frac{1}{2}}\vA_{\sigma}(\omega)\vrho^{\frac{1}{2}}\big\|^2 \\
    =&\frac{\sigma\sqrt{\pi/2}}{\sigma'\sigma_0}\cdot \lnormp{\int_{-\infty}^{\infty} \bigg(\vA_{\sigma'}(\omega')\sqrt{\vsigma}-\sqrt{\vsigma}\vrho^{-\frac{1}{2}}\vA_{\sigma'}(\omega')\vrho^{\frac{1}{2}}\bigg)\hat{f}_{\sigma_0}(\omega-\omega') \rd \omega'}{2} \tag*{(\autoref{lem:A_convolve_indentity})}^2 \\
    \lesssim & \beta \cdot  \bigg(\int_{-\infty}^{\infty} \lnormp{ \vA_{\sigma'}(\omega')\sqrt{\vsigma}-\sqrt{\vsigma}\vrho^{-\frac{1}{2}}\vA_{\sigma'}(\omega')\vrho^{\frac{1}{2}}}{2} \cdot \hat{f}_{\sigma_0}(\omega-\omega') \rd \omega'\bigg)^2 \tag*{(Triangle Ineq.)} \\
    \le &\beta \cdot \int_{-\infty}^{\infty} \hat{f}_{\sigma_0}(\omega-\omega')  \rd \omega'\cdot \int_{-\infty}^{\infty} \lnormp{ \vA_{\sigma'}(\omega')\sqrt{\vsigma}-\sqrt{\vsigma}\vrho^{-\frac{1}{2}}\vA_{\sigma'}(\omega')\vrho^{\frac{1}{2}}}{2}^2 \cdot  \hat{f}_{\sigma_0}(\omega-\omega') \rd \omega'. \tag*{(Cauchy-Schwarz)}
\end{align}

\noindent Next, we note that the shifted Gaussian frequency functions satisfy
\begin{equation}
    \hat{f}_{\sigma_0}(\omega-\omega') =  e^{-(\omega-\omega')^2/4\sigma_0^2}/\sqrt{\sigma_0\sqrt{2\pi}}\quad \text{such that} \quad \sup_{\omega}\int_{-\infty}^{\infty} \hat{f}_{\sigma_0}(\omega-\omega')  \rd \omega' \lesssim \sqrt{\sigma_0} = 1/\sqrt{\beta}
\end{equation}

\noindent which enables us to conveniently re-express our lower bound for $\mathsf{ADB}$ in terms of the original operator FTs of the jump operators:
\begin{align}
    &\int_{-\infty}^{\infty} \lnormp{\vA_{\sigma}(\omega)\sqrt{\vsigma}-\sqrt{\vsigma}\vrho^{-\frac{1}{2}}\vA_{\sigma}(\omega)\vrho^{\frac{1}{2}}}{2}^2  \gamma(\omega) \rd \omega \lesssim \\  \sqrt{\beta}& \iint_{-\infty}^{\infty} \lnormp{\vA_{\sigma'}(\omega')\sqrt{\vsigma}-\sqrt{\vsigma}\vrho^{-\frac{1}{2}}\vA_{\sigma'}(\omega')\vrho^{\frac{1}{2}}}{2}^2  \gamma(\omega) \hat{f}_{\sigma_0}(\omega-\omega') \rd \omega \rd \omega'.
\end{align}
Finally, we compare the convolution of $\gamma(\omega)$ with that of~\eqref{eq:gamma_convolve_gaussian}. Since the shifted Metropolis weight $\gamma(\omega)$ is positive and changes exponentially (which is dominated by Gaussian decay)
\begin{align}
    \gamma(\omega') e^{-\beta\labs{\omega-\omega'}} \le \gamma(\omega) \le \gamma(\omega') e^{\beta\labs{\omega-\omega'}},
\end{align}
we obtain a point-wise bound. For each $\omega'\in \BR,$
\begin{align}
     \int_{-\infty}^{\infty}\gamma(\omega) \hat{f}_{\sigma_0}(\omega-\omega') \rd \omega &\le \int_{-\infty}^{\infty}\gamma(\omega') e^{\beta\labs{\omega-\omega'}} \hat{f}_{\sigma_0}(\omega-\omega') \rd \omega\\ &=  \gamma(\omega') \undersetbrace{\lesssim\sqrt{\beta}}{\int_{-\infty}^{\infty} e^{\beta\labs{x}} \hat{f}_{\sigma_0}(x) \rd x}\lesssim \frac{\gamma(\omega')}{\sqrt{\beta}}  \undersetbrace{\gtrsim\beta} {\int_{-\infty}^{\infty} e^{-\beta\labs{x}} \hat{f}_{\sigma_0}(x)^2 \rd x} \\
     &=  \int_{-\infty}^{\infty}\frac{\gamma(\omega')}{\sqrt{\beta}} e^{-\beta\labs{\omega-\omega'}} \hat{f}_{\sigma_0}(\omega-\omega')^2 \rd \omega \le  \int_{-\infty}^{\infty} \frac{\gamma(\omega)}{\sqrt{\beta}} \hat{f}_{\sigma_0}(\omega-\omega')^2 \rd \omega. \label{eq:gammafsigma0}
\end{align}
Compare~\eqref{eq:gammafsigma0} with~\eqref{eq:gamma_convolve_gaussian} and use the implicit assumption that $\sigma \le 1/\beta$ to conclude the proof. 
\end{proof}
To conclude this subsection, we present the proofs of the two convolution properties we used. 

\begin{proof}

[of the convolution property, \autoref{lem:A_convolve_indentity}]
\begin{align}
    \vA_{\sigma_3}(\omega) = \sum_{\nu} \vA_{\nu} \hat{f}_{\sigma_3}(\omega-\nu)
        &= \frac{\sqrt{\sigma_3\sqrt{\pi/2}}}{\sqrt{\sigma_1\sigma_2}} \sum_{\nu}\vA_{\nu}  \int_{-\infty}^\infty\hat{f}_{\sigma_1}(\omega'-\nu)\hat{f}_{\sigma_2}(\omega-\omega') \rd \omega\\
        &= \frac{\sqrt{\sigma_3\sqrt{\pi/2}}}{\sqrt{\sigma_1\sigma_2}}\int_{-\infty}^{\infty} \vA_{\sigma_1}(\omega') \hat{f}_{\sigma_2}(\omega-\omega')\rd \omega'.
\end{align}
\end{proof}

\begin{proof}

[of the twirling property, \autoref{lem:OFT_Guassian_time_avg}] We evaluate the time $t$ integral for the quadratic expression
    \begin{align}
        \int_{-\infty}^{\infty} \hat{\vA}_{\sigma_1}(\omega) (t) \otimes \hat{\vA}_{\sigma_1}(\omega)^{\dagger} (t)& f_{\sigma_2}(t)^2 \rd t  =\sum_{\nu_1,\nu_2} \vA_{\nu_1}\otimes  (\vA_{\nu_2})^{\dagger} \hat{f}_{\sigma_1}(\omega-\nu_1) \hat{f}_{\sigma_1}(\omega-\nu_2) \exp(-\frac{(\nu_1-\nu_2)^2}{8\sigma_2^2})\\
        & = \frac{1}{\sigma_1\sqrt{2\pi}} \sum_{\nu_1,\nu_2} \vA_{\nu_1}\otimes  (\vA_{\nu_2})^{\dagger} \exp(-\frac{(\nu_1+\nu_2-2\omega)^2}{8\sigma_1^2})\exp(-\frac{(\nu_1-\nu_2)^2}{8\sigma_3^2})\\
&= \int_{-\infty}^{\infty} \undersetbrace{=\hat{f}_{\sigma_2}(\omega-\omega')^2}{\frac{\exp\L(-\frac{(\omega-\omega')^2}{2(\sigma_1^2-\sigma_3^2)}\R)}{\sqrt{2\pi(\sigma_1^2-\sigma_3^2)}}} \hat{\vA}_{\sigma_3}(\omega') \otimes  \hat{\vA}_{\sigma_3}(\omega')^{\dagger} \rd \omega',
\end{align}
where the second line rewrites the bivariate Gaussian in terms of the variables $\nu_1+\nu_2, \nu_1-\nu_2$.
\end{proof}

\section{Approximate Detailed Balance Implies a Local Markov Property}
\label{sec:proving_Markov}
We have seen in~\eqref{eq:two_timescales} that running Glauber dynamics with updates restricted to a local region defines a resampling map for the conditional Gibbs measure, \textit{conditioned} on the neighborhood of the region. Recent work \cite{chen2025GibbsMarkov} for quantum Gibbs states provided a template that inspires our characterization of metastability, known as a \textit{local Markov property}: local disturbances to the Gibbs state can be recovered by a quasi-local Lindbladian dynamics, the time-averaging map on a region $\mathsf{A}\subseteq [n]$
\begin{align}
    \CR_{\mathsf{A},t}[\cdot] &:= \frac{1}{t}\int_{0}^t e^{s\mathcal{L}_\mathsf{A}}[\cdot] \,\rd s, \quad\text{where}\quad \mathcal{L}_\mathsf{A} :=\sum_{a\in P^1_\mathsf{A}}\CL_a. \label{eq:RAt}
\end{align}
In this section, we prove\footnote{The reader may note \autoref{thm:local_recovery} is slightly distinct from the advertised bound in \autoref{thm:local_recovery_intro}. As we discuss shortly, this is since the latter captures only the ``Local Mixing" contribution to the error in the Markov property.} that metastable states also admit such a local-recoverability property, provided that the time-scale behind the local mixing process is sufficiently faster than the lifetime of the metastable state. 

\begin{thm}[Approximate Detailed Balance implies a Local Markov Property]\label{thm:local_recovery} Consider a Hamiltonian $\vH$ with interaction degree at most $d$ and a region $\mathsf{A}\subset [n]$. Suppose a state $\vsigma$ satisfies approximate detailed balance (\autoref{defn:adb}) for all single-site Pauli operators on $\mathsf{A}$, with error $\epsilon_{\mathsf{ADB}} = \max_{\vP\in P^1_\mathsf{A}} \mathsf{ADB}_{\vP}[\vsigma].$ Then for any $t>0$, the time-averaged dynamics $\CR_{\mathsf{A},t}$ \eqref{eq:RAt} defines an approximate recovery map for $\vsigma$, with error
 \begin{align}
     \norm{\vsigma-\CR_{\mathsf{A},t}[\CN_{\mathsf{A}}[\vsigma] ]}_1 \lesssim \e^{\mu|\mathsf{A}|}\,\cdot t^{-\lambda} + |\mathsf{A}|\cdot t\cdot \epsilon_{\mathsf{ADB}}^{1/2}
 \end{align}
 \noindent for some $0<\mu <\mathsf{poly}(\beta, \beta^{-1}, d) $ and $1 > \lambda> 1/\mathsf{poly}(\beta, \beta^{-1}, d)$.  
\end{thm}

We comment that this theorem justifies the aforementioned ``two time-scales'' picture of the local Markov property for metastable states, now in the quantum setting: a ``leakage''  timescale $\sim 1/\epsilon_{\mathsf{ADB}}^{1/2}$ which captures the lifetime of the metastable state, after which it may have macroscopically changed; and a ``local mixing'' time-scale $\sim e^{|\mathsf{A}|}$ in which local noise are fixed by the dynamics. 

Provided $|\mathsf{A}|$ is not too large relative to $\epsilon_{\mathsf{ADB}}$, one can always achieve a meaningful recovery error $\delta$, at some finite time $t^*\sim \mathsf{poly}(e^{|\mathsf{A}|}, \delta^{-1})$. Since each Linbladian term $\CL_a$ is quasi-local, by Lieb-Robinson bounds, the recovery map $\CR_{\mathsf{A},t}$ is also localized around $\mathsf{A},$ with a truncation error growing with time and falling with the radius (distance defined on the interaction graph of the Hamiltonian). As a consequence, one can directly bound the conditional mutual information for any tripartition $\mathsf{ABC},$ where $\mathsf{B}$ shields $\mathsf{A}$ from $\mathsf{C},$ and the truncated $\CR_{A,t}$ is strictly localized to $\mathsf{AB}.$ For the purposes of proving the area law, we do not require the quasi-locality estimate. For the curious reader, we refer to~\cite[Appendix A]{chen2025GibbsMarkov} for a routine calculation.

\subsection{An Outline of the Argument}
\label{section:outline_adb_markov}

As sketched in~\eqref{eq:quantum_two_time_scales}, our starting point to study the quality of the local Lindbladian dynamics as a local resampling map, is to separate the error into two parts:
\begin{align}
    \lnorm{\vsigma-\CR_{\mathsf{A},t}[\tr_\mathsf{A}[\vsigma]\otimes \vtau_\mathsf{A}]}_1 \le \undersetbrace{\text{Local Mixing}}{\big\|\CR_{\mathsf{A},t}[\vsigma-\tr_\mathsf{A}[\vsigma]\otimes \vtau_A]]\big\|_1 }+ \undersetbrace{\text{Leakage}}{\norm{\vsigma- \CR_{\mathsf{A},t}[\vsigma]}_1}. \label{eq:error-two-parts}
\end{align}

\noindent Perhaps unsurprisingly, the ``leakage'' term is a trivial consequence of our metastability assumption:\footnote{We remark that the metastability error can be related back to the approximate detailed balance assumption using \autoref{thm:adb_to_meta}.} 
    \begin{align}
       \|\vsigma- \CR_{\mathsf{A},t}[\vsigma]&\|\leq \bigg\| \vsigma -  \frac{1}{t}\int^t_0  e^{s\mathcal{L}_\mathsf{A}} [\vsigma] \rd s\bigg\|_1\leq \frac{1}{t}\int^t_0\int^{s_1}_0  \bigg\| e^{s_2\mathcal{L}_\mathsf{A}} \circ \mathcal{L}_{\mathsf{A}}[\vsigma]\bigg\|_1 \rd s_2  \rd s_1  \leq t\cdot \big\| \mathcal{L}_{\mathsf{A}}[\vsigma]\big\|_1 ,\label{eq:metastab_leakage}
    \end{align}

\noindent and it remains only to understand the local mixing contribution. This is the main goal of this section, and our proof will largely follow the template in~\cite{chen2025GibbsMarkov}. We dedicate this subsection to an outline of the argument, and extract from it two central conceptual steps where properties of the genuine Gibbs state are used in \cite{chen2025GibbsMarkov}, where we instead leverage approximate detailed balance. 

\subsubsection{Step I: Locally stationary operators are locally trivial}
\label{sec:STEP1_locallytrivial}

Arguably, the central challenge in establishing local mixing, even in the Gibbs state case, is that the Lindbladian $\CL_{\mathsf{A}}$ associated with a local region is not exactly local and is supported on the entire lattice. Consequently, standard a priori bounds on its mixing time are \textit{not} available. Instead, \cite{chen2025GibbsMarkov} makes use of a much weaker notion of convergence: an algebraic decay of the Dirichlet form known as \textit{local stationarity}. 
\begin{equation}
  \labs{\braket{\CR_{\mathsf{A},t}^{\dagger}[\vO],\CL_a^{\dagger}[\CR_{\mathsf{A},t}^{\dagger}[\vO]]}_{\vrho}} \le \frac{2}{t}\cdot \|\vO\|^2\quad \text{for each}\quad \vO.
\end{equation}

\noindent Roughly speaking, this Dirichlet form captures the rate at which the operator $\CR_{\mathsf{A},t}^{\dagger}[\vO]$ is changing, weighted by the Gibbs state. \cite[Lemma X.4]{chen2025GibbsMarkov} then show that if this rate is small, then these global operators must approximately commute with the local jump operators. That is, for any operator $\vX$, 
\begin{align}
  \text{(Local Stationarity)}\quad   -\braket{\vX,\CL_a^{\dagger}[\vX]}_{\vrho} \approx 0 \quad \text{implies that}\quad \norm{[\vA^a,\vX]}_{\vrho}\approx 0 \quad  \text{(Locally Trivial)}. \label{eq:locally_stationary_locally_trivial}
\end{align}
The RHS has a natural recovery interpretation, since if $\vX$ commutes with all local Pauli operators on the region $\mathsf{A}$, then it must be identity on $\mathsf{A}$, and its expectation can depend only on the complement region $\bar{\mathsf{A}}$. 

In this section, our first goal is to establish the following analog of~\eqref{eq:locally_stationary_locally_trivial} for metastable states. 
\begin{align}
    -\braket{\vX,\CL_a^{\dagger}[\vX]}_{\vsigma} \approx 0 \quad \text{implies that}\quad \norm{[\vA^a,\vX]}_{\vsigma}\approx 0.
\end{align}
We refer to the LHS as the ``auxiliary Dirichlet form''. Oddly, this quantity may not be self-adjoint as a bilinear map in $(\vX,\vY)$ and may not be a physical observable due to the square roots $\sqrt{\vsigma}$ in the KMS inner product over $\vsigma$. Nevertheless, the time-averaging map $\mathcal{R}_{\mathsf{A}, t}$ defined in \eqref{eq:RAt}, is still locally stationary:

\begin{lem}[Local Stationarity of $\vsigma$-weighted Dirichlet form]\label{lem:locally_stationary_diri} For any $\norm{\vO}\le 1,$ and any state $\vsigma$,
    \begin{align}
        \labs{\braket{\CR_{\mathsf{A},t}^{\dagger}[\vO],\CL_a^{\dagger}[\CR_{\mathsf{A},t}^{\dagger}[\vO]]}_{\vsigma}} \le \frac{2}{t}
    \end{align}
\end{lem}
\begin{proof}
    Replaced $\vrho$ by $\vsigma$ in \cite[Corollary VII.1]{chen2025GibbsMarkov}.
\end{proof}

Unlike the original Dirichlet form under the genuine Gibbs state (\autoref{lem:Dicirchlet}), our auxiliary Dirichlet form is not precisely a `gradient-square'. Our proof strategy is to nevertheless relate to an expression of said form: 
\begin{align}
    \CC_a(\vX,\vX)_{\vsigma} &:= \| \nabla_a[\vX]\|_{\mathsf{D}, \vsigma}^2\\
    &=\iint_{-\infty}^{\infty} g(t) h(\omega)\tr\L[\sqrt{\vsigma}[\hat{\vA}^a(\omega,t),\vX ]^{\dagger}\sqrt{\vsigma}[\hat{\vA}^a(\omega,t),\vX ] \R]\rd t \rd \omega \ge 0\label{eq:aux_dirichlet_expand}
\end{align}
where, for purposes of comparison, we do have exactly $\CC_a(\vX,\vX)_{\vrho} = -\braket{\vX,\CL_a^{\dagger}[\vX]}_{\vrho}$ for the true Gibbs state $\vrho$ (\autoref{lem:Dicirchlet}). While $\vsigma$ could be statistically far from $\vrho$, approximate detailed balance is precisely the condition that relates the two bilinear forms.

\begin{lem}[The $\vsigma$-weighted Dirichlet form, from $\mathsf{ADB}$]\label{lem:dirichlet_from_adb} For any state $\vsigma$ and operators $\vA^a, \vX, \vY$ satisfying $\|\vA^a\|, \|\vX\|, \|\vY\|\leq 1$,
\begin{align}
    \labs{-\braket{\vX,\CL_a^{\dagger}[\vY]}_{\vsigma} - \CC_a(\vX,\vY)_{\vsigma} } \lesssim \sqrt{\mathsf{ADB}_a[\vsigma]} + \L(\mathsf{ADB}_a[\vsigma]\beta \norm{[\vA,\vH]}\R)^{1/3}. 
\end{align}    
\end{lem}
A direct computation reveals that the error is entirely due to ``moving $\hat{\vA}(\omega,t),\hat{\vA}(\omega)$ past $\sqrt{\vsigma}$'', which connects closely to (approximate) detailed-balance. Afterwards, the remaining calculation is basically the same as in the original proof of the expression for the Dirichlet form in the $\vrho$-weighted case, as derived in \cite[Lemma C.2]{rouze2024optimal}. We defer a self-contained proof of \autoref{lem:dirichlet_from_adb} to the subsequent Section~\ref{section:aux_dirichlet_from_adb}.

Arguably, the most important consequence of the approximate gradient-norm expression for the Dirichlet form $\CC_a(\vX,\vX)_{\vsigma}$, is that it has a manifestly positive expression, which controls ``bare'' local commutators (i.e., without integrals over $\omega, t$). The following key lemma shows how to bound norms of local commutators using the Dirichlet form, and is a generalization of  \cite[Lemma X.4]{chen2025GibbsMarkov} weighted by $\vsigma$ instead of $\vrho$.

\begin{lem}[Locally Stationary implies Locally Trivial]\label{lem:locally_stationary_locally_trivial} Assume $\vH$ is a local Hamiltonian with bounded degree $d$. Then, for any inverse-temperature $\beta$ (and Gaussian width $\sigma=1/\beta$), there exists explicit constants $c_1(\beta, d) \ge 1/\mathsf{poly}(\beta, d)$ and $c_2(\beta, d) \le e^{\mathsf{poly(\beta, d)}}$ with the following guarantees: 

    For any state $\vsigma$, any operator $\vZ$ s.t. $\|\vZ\|\leq 1$, and any single-qubit Pauli operator $\vA^a\in P^1_\mathsf{A}$, 
    \begin{align}
    \norm{[\vA^a,\vZ]}_{\vsigma}\leq c_2  \cdot \CC_a(\vZ,\vZ)_{\vsigma}^{c_1}.
    \end{align}
\end{lem}

\begin{proof} Replaced $\vrho$ by $\vsigma$ in the weights in \cite[Corollary VII.1]{chen2025GibbsMarkov}.
\end{proof}

By combining the three lemmas above, we intuitively already have that the commutators of time-averaged observables decay algebraically $\norm{[\vA^a,\CR_{\mathsf{A}, t}^{\dagger}[\vO]]}_{\vsigma}\sim 1/\mathsf{poly}(t)$, up to error in approximated detailed-balance. However, it remains to show how to relate the fidelity of the recovery map, to bounds on the KMS-norm of local commutators $\norm{[\vA^a,\CR_{\mathsf{A},t}^{\dagger}[\vO]]}_{\vsigma}$, of time-averaged observables $\vO$. \\

\subsubsection{Step 2: The Recovery Fidelity.}
\label{sec:STEP2_moving}

\noindent The second central challenge lies in relating the KMS norm of commutators (\autoref{lem:locally_stationary_locally_trivial}), back to the recovery fidelity of the time-averaging map. The following simple lemma states that this recovery fidelity is instead captured by a \textit{GNS inner product}, of the commutator $[\vA^a,\CR^{\dagger}_{\mathsf{A},t}[\vO]]$ with local Pauli operators on $\mathsf{A}$.

\begin{lem}[Noise Recovery as a Commutator]\label{lem:recovery_from_erasures}
    Fix any state $\vsigma$ and a region $\mathsf{A}\subset [n]$. Consider the state $\CN_{\mathsf{A}}[\vsigma]$ resulting from a noise channel applied to $\mathsf{A}$. Then, the recovery map $\CR_{\mathsf{A}, t}$ recovers $\vsigma$ up to error:
\begin{align}
    \big\|\CR_{\mathsf{A},t}[\vsigma-\CN_{\mathsf{A}}[\vsigma]]\big\|_1 \le 2^{8\labs{\mathsf{A}}+1} \max_{\|\vO\|\leq 1}\max_{\vP,\vQ\in P_\mathsf{A}\cup \vI} \max_{\vA\in P_\mathsf{A}^1}  \labs{\tr\L[ \vsigma \vP [\vA,\CR^{\dagger}_{\mathsf{A},t}[\vO]]\vQ\R]}
\end{align}
\noindent where the max is over multi-qubit Pauli operators $\vP, \vQ$ and single-qubit Pauli operator $\vA$, within the region $\mathsf{A}$.
\end{lem}

We defer the proof of this lemma to the next subsection (Section~\ref{section:recover_from_erasures}). One of the central insights of \cite{chen2025GibbsMarkov} is precisely to leverage detailed-balance, to relate the GNS inner product of commutators arising from the erasure recovery fidelity, back to the KMS norm of commutators which we know to decay with time (from local stationarity): For any (global) operator $\vX$, and local operators $\vP, \vQ$ on $\mathsf{A}$,
\begin{align}
\norm{\vX}_{\vrho} \approx 0 \quad \text{implies that}\quad \tr[\vX\vQ\vrho\vP]\approx 0.  \label{eq:move_rho_around}   
\end{align}

The key difficulty in establishing the above is the location of $\vrho$s. On the LHS of \eqref{eq:move_rho_around}, $\vrho^{1/2}$ is adjacent to $\vX$; but on the RHS $\vrho$ is adjacent to $\vP, \vQ$, thus preventing naive use of Holder inequalities. The role of detailed balance is precisely to ``pass'' the local operators $\vP, \vQ$ through the Gibbs state. 

Here, our goal is to show an analog of~\eqref{eq:move_rho_around} for metastable states:
\begin{align}
\norm{\vX}_{\vsigma} \approx 0 \quad \text{implies that}\quad \tr[\vX\vQ\vsigma\vP]\approx 0,     
\end{align}

\noindent where again the RHS is a physical observable that depends linearly on $\vsigma$, which will later control the recovery fidelity (once we pick $\vX = [\vA, \mathcal{R}_{\mathsf{A}, t}^\dagger[\vO]]$). In order to do so, ultimately we leverage approximate detailed balance, to similarly place $\sqrt{\vsigma}$ adjacent to $\vX$, roughly
\begin{align}
    \tr[\vX\vQ\vsigma\vP]&\approx \tr\bigg[\vX \vsigma^{1/2} \cdot \big(\vrho^{-1/2}\vQ \vrho\vP \vrho^{-1/2}\big)\cdot \vsigma^{1/2} \bigg],
\end{align}
we can then bound the operator norm of $\vrho^{-1/2}\vQ \vrho\vP \vrho^{-1/2}$, upon suitable regularization. We make this implication precise in the following lemma:

\begin{lem}[Exposing KMS-norm using Approximate Detailed Balance]\label{lem:expose_KMS}
 Consider a local Hamiltonian $\vH$ with bounded degree $d$. Then, for any inverse-temperature $\beta$ (and Gaussian filter width $\sigma = 1/\beta$), there exists explicit constants $\alpha_1(\beta, d) \ge 1/\mathsf{poly}(\beta, d)$, $\alpha_2(\beta, d)\le e^{\mathsf{poly}(\beta, \beta^{-1}, d)}$ satisfying the following guarantees. For any state $\vsigma$ and observable $\vX$ s.t. $\|\vX\|\leq 1$, denote
\begin{equation}
    \delta = \|\vX\|_{\vsigma} + \max_{\vA\in P^1_\mathsf{A}}\sqrt{\mathsf{ADB}_{\vA}[\vsigma]}.
\end{equation}
Then, for any Pauli operators $\vP, \vQ\in P_\mathsf{A}\cup \vI$ on some region $\mathsf{A}$, 
\begin{equation}
    \labs{\tr[\vP\vsigma \vQ \vX]}\leq \alpha_2^{|\mathsf{A}|} \cdot  \delta^{\alpha_1}. 
\end{equation}
\end{lem}
The central challenge in establishing \autoref{lem:expose_KMS} is that we only assume $\mathsf{ADB}$ to hold for single-qubit jumps $\vA$, whereas $\vP, \vQ$ are multi-qubit operators. Nevertheless, we show in Section~\ref{section:exposing_kms_adb} that we are still able to transfer the single-qubit $\mathsf{ADB}$ to this setting, up to a loss which scales exponentially in $|\mathsf{A}|$. A full proof of \autoref{lem:expose_KMS} is deferred to Section~\ref{section:exposing_kms_adb}.

\subsection{The Local Mixing Proof (\autoref{thm:local_recovery})}

We now prove that approximate detailed balance (\autoref{defn:adb}) implies a local Markov property. 
\begin{proof}

[of \autoref{thm:local_recovery_intro} \& the ``local mixing'' part of \autoref{thm:local_recovery}] \\

    \noindent \textbf{Step I:} By combining \autoref{lem:locally_stationary_diri}, \autoref{lem:dirichlet_from_adb}, and \autoref{lem:locally_stationary_locally_trivial}, we have that for any operator $\vO$, $\norm{\vO}\le1$: 
    \begin{align}
       \max_{\vA \in P^1_\mathsf{A}} \lnorm{\undersetbrace{=:\vO'}{\frac{1}{2}[\vA, \mathcal{R}_{\mathsf{A}, t}^\dagger[\vO]]}}_{\vsigma} &\leq c_2'\cdot \bigg(\frac{1}{t} + \max_{\vA\in P^1_\mathsf{A}} \mathsf{ADB}_{\vA}[\vsigma]^{1/2} + \beta d\cdot  \mathsf{ADB}_{\vA}[\vsigma]^{1/3}\bigg)^{c_1'},
       \\ &\leq c_2\cdot \bigg(\undersetbrace{=:\delta_t}{\frac{1}{t} + \max_{\vA\in P^1_\mathsf{A}} \mathsf{ADB}_{\vA}[\vsigma]^{1/2}}\bigg)^{c_1},
    \end{align}
    \noindent for explicit constants $c_2(\beta, d)>1> c_1(\beta, d)>0$, and the normalization $1/2$ ensures that $\norm{\vO'}\le 1$. \\
    
    \noindent \textbf{Step II:} When converting between the KMS norm and the GNS inner product in \autoref{lem:expose_KMS}, we note that the error is mostly dominated by that of \textbf{Step I},
    \begin{align}
        \norm{\vO'}_{\vsigma} + \max_{\vA\in P^1_\mathsf{A}} \sqrt{\mathsf{ADB}_{\vA}[\vsigma]} 
        &\le c_2\cdot \delta_t^{c_1} +  \max_{\vA\in P^1_\mathsf{A}}\sqrt{\mathsf{ADB}_{\vA}[\vsigma]}
        \leq 2 \cdot c_2\cdot \delta_t^{c_1}
    \end{align}
    \noindent and consequently \autoref{lem:expose_KMS} gives us the bound:
    \begin{align}
    \labs{\tr\L[ \vsigma \vP \vO'\vQ\R]} \le (2 c_2 \delta_t^{c_1})^{\alpha_1} \alpha_2^{\labs{\mathsf{A}}}\quad \text{for each}\quad \vP,\vQ\in P_\mathsf{A}\cup \vI\label{eq:POQ}.
    \end{align}
    \noindent We can now conclude by applying \autoref{lem:recovery_from_erasures}, to bound the `mixing' error:
\begin{align}
        \big\|\CR_{\mathsf{A},t}[\vsigma-\CN_{\mathsf{A}}[\vsigma]]\big\|_1 &\leq 2^{8\labs{\mathsf{A}}+2} \max_{\|\vO\|\leq 1}\max_{\vP,\vQ\in P_\mathsf{A}\cup \vI} \max_{\vA\in P_\mathsf{A}^1}  \labs{\tr\L[ \vsigma \vP \frac{1}{2}[\vA^a,\CR^{\dagger}_{\mathsf{A},t}[\vO]]\vQ\R]}\\
        &\le 2^{8\labs{\mathsf{A}}+2} (2c_2\delta_t^{c_1})^{\alpha_1} \alpha_2^{|\mathsf{\mathsf{A}}|} \tag*{(By~\eqref{eq:POQ})}\\
        &\leq \delta_t^{\eta}\cdot e^{\chi|\mathsf{A}|}
    \end{align}

    \noindent where $1\geq \eta(\beta,d)\geq 1/\mathsf{poly}(\beta, d)$ and $\chi(\beta, d)\leq \mathsf{poly}(\beta, 1/\beta, d)$. 
\end{proof}

\begin{proof}

    [of \autoref{thm:local_recovery}] By combining the ``leakage'' error derived in \eqref{eq:metastab_leakage} (combined with \autoref{thm:adb_to_meta}), with the ``local mixing'' error from \autoref{thm:local_recovery_intro}, we arrive at:
    \begin{align}
         \lnorm{\vsigma-\CR_{\mathsf{A},t}\circ\CN_{\mathsf{A}}[\vsigma]}_1 &\le \undersetbrace{\text{Local Mixing}}{\big\|\CR_{\mathsf{A},t}[\vsigma-\CN_{\mathsf{A}}[\vsigma]]\big\|_1 }+ \undersetbrace{\text{Leakage}}{\norm{\vsigma- \CR_{\mathsf{A},t}[\vsigma]}_1} \\
         &\lesssim \bigg(\frac{1}{t} + \epsilon_{\mathsf{ADB}}^{1/2}\bigg)^{\eta} \cdot e^{\chi |\mathsf{A}|} + t\cdot |\mathsf{A}|\cdot \epsilon_{\mathsf{ADB}}^{1/2}
    \end{align}
    To conclude, we note that the result becomes vacuous unless $t|\mathsf{A}|\leq (\epsilon_\mathsf{ADB})^{-1}$ to further simplify the local mixing term (by also adequating $\chi$).
\end{proof}

\subsubsection{The Erasure Channel in Commutator Form (Proof of \autoref{lem:recovery_from_erasures})}
\label{section:recover_from_erasures}

First, we rewrite the noise channel in terms of commutator with local operators acting on the region $\mathsf{A}\subset [n]$. This is a straightforward, crude argument as we tolerate factors exponential in the region size $\labs{\mathsf{A}}.$
\begin{lem}[Exposing commutators]
\label{lem:decompose_replacement}

For any noise channel (completely positive and trace preserving map) $\CN_{\mathsf{A}}$ acting on $\mathsf{A}\subset [n]$, there exist coefficients $\labs{c^a_{\vS,\vS'}}\le 2^{4\labs{\mathsf{A}}}$ such that 
\begin{align}
    \vX - \CN_{\mathsf{A}}^{\dagger}[\vX] = \sum_{a\in P_\mathsf{A}} \sum_{\vS,\vS'\in P_\mathsf{A}\cup \vI} c^a_{\vS,\vS'} \vS[\vA^a,\vX]\vS.'
\end{align}
\end{lem}

The desired result in \autoref{lem:recovery_from_erasures} then follows as a corollary, by expressing the trace norm $\|\cdot\|_1$ as a supremum over observables $\vX$ and subsequently applying the triangle inequality. 

\begin{lem} For any set of operators $\vA^i\in P_\mathsf{A}^1$ and $\vO$, $[\prod^{w}_{j=1} \vA^{i}, \vO] = \sum_{i=1}^w\prod^{j-1}_{i=1}\vA^{i} \cdot [\vA^{j}, \vO] \cdot \prod^{w}_{i=j} \vA^{i}$.
\end{lem}
The precise exponent $4\labs{\mathsf{A}}$ can likely be improved by multiplicative constants.

\begin{proof}

    [of \autoref{lem:decompose_replacement}] Consider the Kraus decomposition for the channel
\begin{align}
    \CN[\vsigma]= \sum_{i=1}^{4^{\labs{\mathsf{A}}}} \vV_i \vsigma\vV_i^{\dagger}\quad \text{such that}\quad \sum_i \vV_i^{\dagger}\vV_i = \vI,
\end{align}
using that any quantum channel can be written in square-dimension ($4^{\labs{\mathsf{A} }}= (2^{\labs{\mathsf{A} }})^2$) many Kraus operators.
Then, in the Heisenberg picture, we can expose commutators
\begin{align}
    \vO- \CN^{\dagger}[\vO] &=\frac{1}{2}\vO \bigg( \sum_{i}\vV_i^{\dagger}\vV_i\bigg) + \frac{1}{2}\bigg(\sum_i   \vV_i^{\dagger}\vV_i\bigg)\vO -\sum_{i} \vV_i^{\dagger} \vO \vV_i \quad \quad \text{(Since $\sum_i \vV_i^{\dagger}\vV_i = \vI$)}\\
    &= \frac{1}{2}\sum_{i} [\vO,\vV_i^{\dagger}] \vV_i + \vV_i^{\dagger} [\vV_i,\vO] . 
\end{align}
Now, consider a decomposition into Pauli strings
$\vV_i = \sum_{\vS \in P_\mathsf{A}\cup \vI} c^{\vS}_{i} \vS,$ and expand the commutator into local ones by $\vS = \prod^{w}_{j=1} \vA^{i}$
\begin{align}
    [\vV_i,\vO] = \sum_{\vS \in P_\mathsf{A}\cup \vI} c^{\vS}_{i} [\vS,\vO] =\sum_{\vS \in P_\mathsf{A}\cup \vI} \sum_{i=1}^w c^{\vS}_{i} \prod^{j-1}_{i=1}\vA^{i} \cdot [\vA^{j}, \vO] \cdot \prod^{w}_{i=j} \vA^{i},
\end{align}
and similarly for the $[\vO,\vV^{\dagger}_i]$ term. Exchange the sum $\sum_{i=1}^w$ by $\sum_{a\in P^1_\mathsf{A}}$, and use the very crude bound $\labs{c_i^{\vS}} \le \sqrt{\sum_{i',\vS'}\labs{c_{i'}^{\vS'}}^2 } =1$, cardinality bound $\labs{P_\mathsf{A}\cup \vI} = 4^{\labs{\mathsf{A}}}$, and the Kraus rank bound $4^{\labs{\mathsf{A}}}$
to bound the final coefficients $\labs{c^a_{\vS,\vS'}}\le 4^{\labs{\mathsf{A}}}\cdot 4^{\labs{\mathsf{A}}}.$ 
\end{proof}

\subsection{The Auxiliary Dirichlet Form (\autoref{lem:dirichlet_from_adb})}
\label{section:aux_dirichlet_from_adb}

We dedicate this subsection to the proof of \autoref{lem:dirichlet_from_adb}, on the auxiliary $\vsigma$ weighted Dirichlet form. To simplify the notation, within the proof we drop the jump label $\vA^{a}=\vA, \CL_a=\CL, \vC_a = \vC$. For convenience, we recollect the explicit form of the Lindbladians $\mathcal{L}$ we consider, in terms of commutators 
\begin{align}
    \CL^{\dagger}[\vY] = \int_{-\infty}^{\infty} \frac{\gamma(\omega)}{2}\L([\hat{\vA}(\omega)^{\dagger},\vY]\hat{\vA}(\omega)+\hat{\vA}(\omega)^{\dagger}[\vY,\hat{\vA}(\omega)] \R) \rd \omega +i [\vC,\vY].
\end{align}
Recall that the coherent term is given by
\begin{align}
    \vC &= \iint_{-\infty}^{\infty} \gamma(\omega) \hat{\vA}(\omega,t)^{\dagger}\hat{\vA}(\omega,t) c(t) \rd t\rd \omega\quad\text{where}\quad c(t) =\frac{1}{\beta\sinh(2\pi t/\beta)}.
\end{align}
Broadly speaking, the goal in the Dirichlet form (refer back to \eqref{eq:aux_dirichlet_expand}) calculation is to exhibit a ``commutator square" expression, where the inner product $\braket{\vX,\CL_a^{\dagger}[\vY]}_{\vsigma}$ has commutators acting on $\vY$ twice. To expose a commutator square, we need to pass $\hat{\vA}(\omega)$ through $\vsigma^{1/2}$ using approximate detailed balance.

\begin{lem}
    [The Dissipative Error] Fix a state $\vsigma$, and operators $\vA, \vX, \vY$ s.t. $\|\vA\|, \|\vX\|, \|\vY\|\leq 1$. Then,
    \begin{align}
        \bigg|\int_{-\infty}^{\infty} \gamma(\omega) \bigg( \ipc{\vX}{\hat{\vA}(\omega)^\dagger [\vY, \hat{\vA}(\omega)]}_{\vsigma} - \ipc{\vrho^{-1/2}\hat{\vA}(\omega)\vrho^{1/2}\vX}{[\vY, \hat{\vA}(\omega)]}_{\vsigma}\bigg) \rd \omega \bigg| \lesssim \sqrt{\mathsf{ADB}_{\vA}[\vsigma]}.
    \end{align}
\end{lem}

\begin{proof}
    The error in the application of approximate detailed balance above follows from a purification trick akin to the proof of~\autoref{thm:adb_to_meta}. To be explicit, define 
    \begin{align}
        &\vU = \int_{-\infty}^{\infty}  \sqrt{\gamma(\omega)}\cdot (\sqrt{\vsigma} \hat{\vA}(\omega)^\dagger - \vrho^{1/2}\hat{\vA}(\omega)^\dagger\vrho^{-1/2} \sqrt{\vsigma}) \otimes \bra{\omega}\ \rd \omega\\
        &\vV = \int_{-\infty}^{\infty} \sqrt{\gamma(\omega)}\cdot   [\vY,\hat{\vA}(\omega)] \otimes \ket{\omega} \rd \omega. 
    \end{align}
    Then the error quantity of interest can therefore be written as 
    \begin{align}
        |\tr[\vX\vU \vV \vsigma^{1/2}] &\leq \|\vX\|\cdot \|\vU\|_2\cdot \|\vV\|\cdot \|\vsigma^{1/2}\|_2 
    \end{align}
    To conclude the proof, we bound the desired norms
    \begin{align}
        &\|\vU\|_2^2 = \int_{-\infty}^{\infty} \gamma(\omega) \lnorm{\sqrt{\vsigma} \hat{\vA}(\omega)^\dagger - \vrho^{1/2}\hat{\vA}(\omega)^\dagger\vrho^{-1/2} \sqrt{\vsigma}}_2^2 \rd\omega\lesssim \mathsf{ADB}_{\vA}[\vsigma],  \\
        &\|\vV\|^2 \leq 4 \|\vY\|^2\cdot \bigg\| \int_{-\infty}^{\infty}  \gamma(\omega) \hat{\vA}^\dagger (\omega)\hat{\vA}(\omega)\rd\omega\bigg\| \leq 4 .
    \end{align}
    \noindent where the last line follows from \autoref{lem:OParseval}.
\end{proof}

Handling the coherent term is slightly more delicate due to the weight $c(t)$, which diverges near $t=0$. We address this divergence via a truncation around $t=0$, which results in some algebraic loss. 
\begin{lem}
    [The Coherent Error] Fix a state $\vsigma$, and operators $\vA, \vX, \vY$ s.t. $\|\vA\|, \|\vX\|, \|\vY\|\leq 1$. Then,
    \begin{align}
        \bigg|\iint_{-\infty}^{\infty} \gamma(\omega)c(t)\cdot \bigg(  &\ipc{\vX}{\hat{\vA}(\omega, t)^\dagger [\vY, \hat{\vA}(\omega, t)]}_{\vsigma} \\&-\ipc{\vrho^{-1/2}\hat{\vA}(\omega, t)\vrho^{1/2}\vX}{[\vY, \hat{\vA}(\omega, t)]}_{\vsigma}\bigg)\rd\omega\rd t \bigg|  \lesssim \bigg( \mathsf{ADB}_{\vA}[\vsigma]\cdot \norm{[\vA,\vH]}\bigg)^{1/3}.
    \end{align}
\end{lem}

\begin{proof}

    Introduce a truncation parameter $\theta >0$ for the time $t$ integral. We begin with \\
    
    \noindent \textbf{The Small $\labs{t}\le \theta$ Regime}. Note that $c(t)$ is an odd function in $t$. Therefore, for any function $r(t)$, 
\begin{align}
    \labs{\int_{t \le \theta} r(t)c(t) \rd t} &= \labs{\int_{t \le \theta} \bigg(r(0)+ \int_{0}^t r'(s) \rd s \bigg) c(t) \rd t} \\&\le \sup_s\labs{r'(s)}\cdot  \labs{\int_{t \le \theta} \labs{t}\cdot \labs{c(t)} \rd t } \\
    &\lesssim \sup_s\labs{r'(s)}\cdot \theta. \quad \quad \tag*{(Since $\sinh x \geq x$)}
\end{align}

In our setting, one can explicitly write out the time derivative of the integrand as follows. Note that there is time-dependence in four different $\hat{\vA}(\omega, t)$ in the expression for the coherent error. The first of the two components is of the following form:
\begin{align}
    \bigg|\int_{-\infty}^{\infty} &\gamma(\omega) \braket{\vX, \hat{\vA}(\omega,s)^{\dagger}[[\hat{\vA}(\omega,s),\vH],\vY]}_{\vsigma}\rd\omega \bigg|, \bigg|\int_{-\infty}^{\infty}  \gamma(\omega) \braket{\vX, [\hat{\vA}(\omega,s)^{\dagger}, \vH][\hat{\vA}(\omega,s),\vY]}_{\vsigma}\rd\omega\bigg| \leq \\
    &\leq \sup_\omega \gamma(\omega)\cdot   \lnorm{\int_{-\infty}^{\infty}   \vA(\omega)^{\dagger}\vA(\omega)\rd \omega}^{1/2} \cdot \lnorm{\int_{-\infty}^{\infty}  [\vA,\vH](\omega)^{\dagger}[\vA,\vH](\omega)\rd \omega}^{1/2}
    \lesssim \norm{[\vA,\vH]},  
\end{align}
\noindent where we similarly applied the purification trick and then \autoref{lem:OParseval}. Similarly, the terms resulting from the imaginary time-conjugation are also bounded by 
\begin{align}
    &\bigg|\int_{-\infty}^{\infty} \gamma(\omega) \ipc{\vX}{ \vrho^{\frac{1}{2}}\hat{\vA}(\omega,s)^{\dagger}\vrho^{-\frac{1}{2}}[[\hat{\vA}(\omega,s),\vH],\vY]}_{\vsigma}\rd\omega \bigg| \\ &= \bigg|\int_{-\infty}^{\infty} \bigg(\gamma(\omega) e^{\beta \omega/2 + \sigma^2\beta^2/4}\bigg) \ipc{\vX}{\hat{\vA}(\omega+\sigma^2\beta,s)^{\dagger}[[\hat{\vA}(\omega,s),\vH],\vY]}_{\vsigma} \rd\omega\bigg| \\
    &\leq \sup_\omega \bigg(\gamma(\omega) e^{\beta \omega/2 + \sigma^2\beta^2/4}\bigg)\cdot   \lnorm{\int_{-\infty}^{\infty}   \vA(\omega)^{\dagger}\vA(\omega)\rd \omega}^{1/2} \times \lnorm{\int_{-\infty}^{\infty}  [\vA,\vH](\omega)^{\dagger}[\vA,\vH](\omega)\rd \omega}^{1/2} \\
    &\lesssim \norm{[\vA,\vH]},
\end{align}
\noindent where in sequence we used \autoref{lem:bounds_imaginary_conjugation}, \autoref{lem:OParseval}, the explicit value of $\gamma$, and assumed $\sigma= 1/\beta$.\\

\noindent \textbf{The Large $t$ Regime.} Here, we simply apply Holder's inequality (again, under the same purification trick), as well as the approximate detailed balance condition:
\begin{align}
        \bigg|\int_{|t|\geq \theta} \frac{c(t)}{\sqrt{g(t)}}\cdot \sqrt{g(t)}\int_{-\infty}^{\infty}\gamma(\omega)\cdot \bigg( & \braket{\vX,\hat{\vA}(\omega, t)^\dagger [\vY, \hat{\vA}(\omega, t)]}_{\vsigma} \\&-\braket{\vrho^{-1/2}\hat{\vA}(\omega, t)\vrho^{1/2}\vX, [\vY, \hat{\vA}(\omega, t)]}_{\vsigma}\bigg)\rd\omega \rd t\bigg| \leq \\
        &\lesssim \sqrt{\mathsf{ADB}_{\vA}[\vsigma]} \cdot \lnorm{\int_{t >\theta} \int_{-\infty}^{\infty}\gamma(\omega)\hat{\vA}(\omega,t)^{\dagger}\hat{\vA}(\omega,t)\frac{c(t)^2}{g(t)} \rd t\rd \omega}^{1/2}  \text{(Holder's)}. \\
        &\leq \sqrt{\mathsf{ADB}_{\vA}[\vsigma]} \cdot  \bigg| \int_{t\ge \theta} \frac{c(t)^2}{g(t)} \rd t\bigg|^{1/2} \quad  \quad \text{(By \autoref{lem:OParseval})} \\
        &\leq \sqrt{\mathsf{ADB}_{\vA}[\vsigma]} \sqrt{\frac{\beta}{\theta}},
\end{align}
\noindent where, in the last inequality, we computed
\begin{equation}
    \int_{t\ge \theta} \frac{c(t)^2}{g(t)} \rd t = \int_{t\ge \theta} \frac{\cosh(2\pi t/\beta)}{\beta \sinh^2(2\pi t/\beta)} \rd t \lesssim \frac{\beta}{2\pi \theta}.
\end{equation}

\noindent\textbf{Combining the Regimes.} We proceed by balancing the errors as a function of $\theta$, up to some algebraic loss:
\begin{align}
    (\text{total error}) \lesssim \sqrt{\mathsf{ADB}_{\vA}[\vsigma]} \sqrt{\frac{\beta}{\theta}} + \norm{[\vA,\vH]}\cdot \theta 
\lesssim \bigg(\beta \cdot \mathsf{ADB}_{\vA}[\vsigma]\cdot \norm{[\vA,\vH]}\bigg)^{1/3}
\end{align}

\noindent with the choice $\theta = (\mathsf{ADB}_{\vA}[\vsigma]\beta)^{1/3}/\norm{[\vA,\vH]}^{2/3}$.
\end{proof}

We now combine the dissipative error and coherent error to conclude this section.
\begin{proof}

    [of \autoref{lem:dirichlet_from_adb}] Rewrite the approximate expressions derived above, as the weighted KMS norm of the commutators $[\hat{\vA}(\omega, t), \vX]$. Then, repeat the existing calculation for Gibbs state $\vrho$ (\cite[Lemma C.2]{rouze2024efficient}, \cite[Lemma X.3]{chen2025GibbsMarkov})
    \begin{align}
        &\int_{-\infty}^{\infty} \gamma(\omega)\bigg(\ipc{\vrho^{-1/2}\hat{\vA}(\omega)\vrho^{1/2}\vX}{\ [\vX, \hat{\vA}(\omega)]}_{\vsigma} + i\int_{-\infty}^{\infty} c(t) \ipc{\vrho^{-1/2}\hat{\vA}(\omega, t)\vrho^{1/2}\vX}{[\vX, \hat{\vA}(\omega, t)]}_{\vsigma}\rd t\bigg)\rd\omega  \\
        =&\iint_{-\infty}^{\infty}  g(t) h(\omega)\cdot  \ipc{[\hat{\vA}(\omega, t), \vX]}{[\hat{\vA}(\omega, t), \vX] }_{\vsigma}\rd t \rd \omega
    \end{align}
    to conclude the proof.
\end{proof}

\subsection{Exposing the KMS norm using Approximate Detailed Balance (Proof of~\autoref{lem:expose_KMS})}
\label{section:exposing_kms_adb}

From~\cite{chen2025GibbsMarkov}, a key feature of the Gibbs state is that one can move local operators ``around'' $\vrho$, without serious divergences (\cite[Lemma IX.5]{chen2025GibbsMarkov}). In this section, we prove \autoref{lem:expose_KMS}, that local operators also interact nicely with the metastable state $\vsigma$. Somehow, approximate detailed balance (\autoref{lem:meta_implies_ADB_without_t}, as a result of metastability), implicitly passes on the regularizing effect of the Gibbs state of a few-body Hamiltonian.

\begin{rmk}
    We are not able to prove \emph{\cite[Lemma IX.5]{chen2025GibbsMarkov}} as stated, as it seems to require approximate detailed balance in the sense of $\norm{\vsigma^{\frac{1}{4}}\hat{\vA}(\omega) - \vrho^{\frac{1}{4}} \hat{\vA}(\omega) \vrho^{-\frac{1}{4}} \vsigma^{\frac{1}{4}}}_4\approx 0,$ which is strictly stronger than with $\sqrt{\vsigma}$ in the 2-norms. Fortunately, the proof can be rerouted to avoid this version of approximate detailed balance. 
\end{rmk}

In~\cite{chen2025GibbsMarkov}, the trick is to rewrite the jump operators in the frequency basis  (\autoref{lem:sumoverenergies}) 
\begin{align}
\vA =  \frac{1}{\sqrt{2\sigma\sqrt{2\pi}}} \int_{-\infty}^{\infty} \hat{\vA}(\omega)\rd \omega = \frac{1}{\sqrt{2\sigma\sqrt{2\pi}}} \bigg(\int_{\labs{\omega}\le \Omega}+\int_{\labs{\omega}>\Omega} \bigg)\hat{\vA}(\omega)\rd \omega.
\end{align}
Then, to control the divergence of the low-frequency part by using \autoref{lem:bounds_imaginary_conjugation}:
\begin{equation}
    \|\vrho^{1/2}\hat{\vA}(\omega)\vrho^{-1/2}\| \propto e^{\beta\omega/2}.
    \end{equation}
Finally, balancing the two parts yielded a system-sized independent bound. To extend this approach, however, the main challenge we encounter is that our approximate detailed balance condition holds only for local (Pauli) operators -- much smaller than the support of the Pauli operators $\vP, \vQ$ on the target region $\mathsf{A}$. To address multi-qubit Pauli's $\vP = \otimes_{i\in \mathsf{A}} \vA^i$, we expand each single-qubit operator using the operator Fourier transform:
\begin{align}
    \vP = \otimes_{i\in \mathsf{A}} \vA^i&= \frac{1}{(2\sigma\sqrt{2\pi})^{|\mathsf{A}|/2}} \int_{\vec{\omega}\in \BR^k}  \prod_i \hat{\vA}^i (\omega_i) \rd \vec{\omega}, \quad \text{denoting}\quad \vec{\omega} = (\omega_1, \cdots, \omega_{|\mathsf{A}|})\in\mathbb{R}^{k}\\
    &= \frac{1}{(2\sigma\sqrt{2\pi})^{|\mathsf{A}|/2}} \L(\int_{\vec{\omega}\in \mathsf{U}_{\Omega}}+\int_{\vec{\omega}\notin \mathsf{U}_{\Omega}}\R) \prod_i \hat{\vA}^i (\omega_i) \rd \vec{\omega} \\
    &=: \vP_{\mathsf{U}_\Omega} + \vP_{\notin\mathsf{U}_\Omega}.
\end{align}
The low/high frequency division is based on the following carefully chosen sets for $k=\labs{\mathsf{A}}$ 
\begin{align}
    \mathsf{U}_{\Omega} &:= \bigg\{ \vec{\omega} : \sum_{j=1}^{k} |\omega_j| \leq \Omega \bigg\}  \subset \mathbb{R}^{k}\label{eq:UOmega}
\end{align}
namely, the sum of the norms of the frequencies is bounded by a tunable cut-off $\Omega >0.$\footnote{It is useful to have the sum altogether constrained; imposing $\labs{\omega_i}\le \Omega$ individually can lead to super-exponential corrections in the region size $|\mathsf{A}|$.} It will also be helpful to consider the complement of $\mathsf{U}_{\Omega}$ and its refinement into sets $\mathsf{V}_i$
\begin{align}
    \overline{\mathsf{U}}_{\Omega} &= \bigg\{\vec{\omega}: \sum_{j=1}|\omega_j| \geq \Omega\bigg\} = \bigcup_{i=1} \mathsf{V_i}, \quad\text{where}\quad \mathsf{V_i} := \bigg\{ \vec{\omega}: \sum_{j=1}^i|\omega_j| \geq \Omega \text{ and }  \sum_{j =1}^{i-1}|\omega_j| \leq \Omega\bigg\},
\end{align}
namely, where the norm bound is first violated by $\omega_i$. The volume of the simplex will also be handy.
\begin{lem}
    [Volume of Simplex]\label{lem:simplex_volume} In the setting of~\eqref{eq:UOmega},    
\begin{align}
    \int_{\vec{\omega} \in \mathsf{U}_{\Omega} } \rd \vec{\omega} =    \frac{2^{k}\Omega^{k}}{k!}.
\end{align}
\end{lem}

\subsubsection{Decay of high-frequency parts for multiple qubits}

\begin{lem} [Norms of High-Frequency Truncations]\label{lem:norms_of_high_freq} Let $\vH$ be a bounded degree $d$ Hamiltonian, and set the Gaussian width $\sigma=1/\beta$. Then, there exists explicit constants $r_1(d) = (8d)^{-1}, e^{\mathsf{poly}(\beta, 1/\beta, d)}\geq r_2(\beta, d)\geq 1$ with the following guarantee. For any $\labs{\mathsf{A}}$-qubit Pauli $\vP\in P_{|\mathsf{A}|}$ and frequency cutoff $\Omega\in \mathbb{R}^+$: 
    \begin{equation}
         \|\vP_{\notin \mathsf{U}_\Omega}\| \leq e^{-r_1\Omega} \cdot r_2^{|\mathsf{A}|}. 
    \end{equation}

\end{lem}
Therefore, we can also bound the norm of the complement by 
\begin{align}
\|\vP_{\mathsf{U}_\Omega}\| \le \norm{\vP} + \|\vP_{\notin \mathsf{U}_\Omega}\|\le 1 +e^{-r_1\Omega} \cdot r_2^{|\mathsf{A}|}.\label{eq:complement_bound}    
\end{align}

\begin{proof} 

[of \autoref{lem:norms_of_high_freq}] For every $\sigma,$ 
    \begin{align}
          \bigg\| \int_{\vec{\omega}\notin \mathsf{U}_\Omega} \prod_i \hat{\vA}^i (\omega_i) \rd \vec{\omega} \bigg\|  &\leq \int_{\vec{\omega}\notin \mathsf{U}_\Omega} \prod_i \bigg\| \hat{\vA}^i (\omega_i) \bigg\| \rd \vec{\omega}\\
          &\lesssim \bigg(\frac{e^{\sigma^2\beta_0^2}}{\sqrt{\sigma}}\bigg)^{|\mathsf{A}|}\cdot  \bigg(\prod_i  2\|e^{-\beta_0\vH}\vA^i e^{\beta_0 \vH}\|\bigg) \cdot  \int_{\vec{\omega}\notin \mathsf{U}_{\Omega}} e^{-\beta_0 \sum_i|\omega_i|} \rd\vec{\omega} \tag*{(\autoref{lem:high_freq_bound})} \\
          &\leq \bigg(\frac{2e^{\sigma^2\beta_0^2}}{\sqrt{\sigma} (1-2d\cdot \beta_0)}\bigg)^{|\mathsf{A}|}\cdot \int_{\vec{\omega}\notin \mathsf{U}_{\Omega}} e^{-\beta_0 \sum_i|\omega_i|} \rd \vec{\omega}. \tag*{(\autoref{lem:norm_bound_pauli})}
    \end{align}
The second line introduces tunable parameter $\beta_0$ and uses that $\|e^{\beta_0\vH}\vA^i e^{-\beta_0 \vH}\| = \|e^{-\beta_0\vH}\vA^i e^{\beta_0 \vH}\|$ for Hermitian $\vA^i.$
    \noindent We refine the complement of $\mathsf{U}_{\Omega}$ into sets $\mathsf{V}_i$ (i.e., the norm bound is first violated by $\omega_i$) and compute the integral separately for each subset. We write the coordinates $\vec{\omega} = (\vec{\omega}_{\le i}, \vec{\omega}_{>i}),$ $\vec{\omega}_{\le i}\subset \mathbb{R}^{i}$. For each $i$, we perform the integral in two batches: the last few frequencies are unconstrained
\begin{align}
    \int_{\vec{\omega}_{> i}} e^{-\beta_0\sum_{j>i}|\omega_j|} \rd \vec{\omega}_{>i} \le \bigg(\frac{2}{\beta_0}\bigg)^{|\mathsf{A}| - i}
\end{align}
using that $\int_{-\infty}^\infty e^{-\beta_0 |x|} \rd x= 2/\beta_0.$ The first few frequencies are constrained in their total norm
    \begin{align}
        \int_{\vec{\omega}_{\leq i}\in \mathsf{V}_{i}} e^{-\beta_0\sum_{j=1}^i|\omega_j|}\rd \vec{\omega}_{\le i} &\le \undersetbrace{ = 2^{i-1}\Omega^{i-1}/(i-1)! }{\int_{\vec{\omega}_{< i}\in \mathsf{V}_{i}} \rd \vec{\omega}_{< i} }\cdot \undersetbrace{=2{e^{-\beta_0\Omega}}/{\beta_0}}{\int_{|s|\geq \Omega}  e^{-\beta_0 |s|} \rd s}
    \end{align}
setting the variable $s=\sum_i |\omega_i|$ and the simplex volume (\autoref{lem:simplex_volume}). We arrive at
  \begin{align}
        \int_{\vec{\omega}\in \mathsf{V}_{i}} e^{-\beta_0\sum_{j=1}|\omega_j|} \rd \vec{\omega}&= \int_{\vec{\omega}_{> i}} e^{-\beta_0\sum_{j>i}|\omega_j|} \rd \vec{\omega}_{>i} \cdot \int_{\vec{\omega}_{\leq i}\in \mathsf{V}_{i}} e^{-\beta_0\sum_{j=1}^i|\omega_j|}\rd \vec{\omega}_{\le i} \quad \\
        &\lesssim  \frac{e^{-\beta_0\Omega}}{\beta_0}\cdot  \bigg(\frac{2}{\beta_0}\bigg)^{\labs{\mathsf{A}}-i} \cdot \frac{(2\Omega)^{\labs{i-1}}}{(i-1)!}\\
        &\lesssim e^{-\beta_0\Omega} \cdot \bigg(\frac{4}{\beta_0}\bigg)^{|\mathsf{A}|} \frac{(\beta_0\Omega/2)^{|i-1|}}{(i-1)!}. 
    \end{align}

    Sum over $ 1\le i\le |\mathsf{A}|,$  expose an exponential $\sum_{i} x^i/i!\leq e^{\labs{x}}$, and include all normalization factors, to conclude 
    \begin{align}
         \|\vP_{\notin \mathsf{U}_\Omega}\| \leq e^{-\beta_0 \Omega/2}  \cdot  \bigg(\frac{8e^{\sigma^2\beta_0^2}}{\beta_0\sigma (8\pi)^{1/4} (1-2d\cdot \beta_0)}\bigg)^{|\mathsf{A}|}.
    \end{align}
    The choice of $\beta_0 = 1/4d$ and the assumption that $\sigma = \frac{1}{\beta}$ then gives $r_1(d) = 1/8d$. 
\end{proof}

Now, we can apply the above to truncate the high-frequency part for the inner product in \autoref{lem:expose_KMS} of interest. 
\begin{lem}[Truncating High-Frequency Components]\label{lem:high_freq_trunc} Let $\vH$ be a bounded degree $d$ Hamiltonian. Then, there exists explicit constants $r_1(d) = 1/8d, r_4(\beta, d)< e^{\mathsf{poly}(\beta, 1/\beta, d)}$ with the following guarantee. Fix a Gaussian width $\sigma=1/\beta$, and a truncation frequency $\Omega\in \mathbb{R}^{+}$. For any state $\vsigma$, Pauli operators $\vP,\vQ \in P_\mathsf{A}\cup \vI$ supported on some region $\mathsf{A}$, and any observable $\vX$ s.t. $\|\vX\|\leq 1$, we have that
    \begin{equation}
        \big|\tr[\vP\vsigma \vQ\vX] - \tr[\vP_{\mathsf{U}_\Omega}\vsigma \vQ_{\mathsf{U}_\Omega}\vX] \big| \leq   e^{-r_1(d)\cdot \Omega} \cdot r_4^{|\mathsf{A}|}.
    \end{equation}
\end{lem}

\begin{proof}

    By the triangle inequality and Holder's inequality, 
    \begin{align}
         \labs{\tr[\vP\vsigma \vQ\vX] - \tr[\vP_{\mathsf{U}_\Omega}\vsigma \vQ_{\mathsf{U}_\Omega}\vX]}
        &\leq \|\vP_{\notin \mathsf{U}_\Omega}\| + \|\vP_{\mathsf{U}_\Omega}\|\cdot \|\vQ_{\notin \mathsf{U}_\Omega}\| \\
        &\le  e^{-r_1(d)\Omega}\cdot r_2^{|\mathsf{A}|} (2+ e^{-r_1(d)\Omega}\cdot r_2^{|\mathsf{A}|}) \\&\lesssim  e^{-r_1(d)\Omega}\cdot r_2^{2|\mathsf{A}|}.
    \end{align}
    \noindent where in the last inequality we leveraged \autoref{lem:norms_of_high_freq} and~\eqref{eq:complement_bound}.
\end{proof}

\subsubsection{The low-frequency parts}

The following lemma regarding the low-frequency parts proceeds by ``peeling off" the single-qubit Pauli operators $\vP^i$, one at a time, using approximate detailed balance.

\begin{lem}
    [Controlling Low-Frequency Components via ADB]\label{lem:low_freq_trunc} Let $\vH$ be a bounded degree $d$ Hamiltonian. Then, there exists an explicit constant $e^{\mathsf{poly}(\beta, 1/\beta, d)}>r_5(\beta, d)>1$ with the following guarantee. Set the Gaussian width $\sigma=1/\beta$, and a tunable truncation frequency $\Omega\in \mathbb{R}^{+}$. For any state $\vsigma$, Pauli operator $\vP=\otimes_i^{\mathsf{A}} \vA^i\in P_\mathsf{A}\cup \vI$ supported on some region $\mathsf{A}$, and any observable $\vX$ s.t. $\|\vX\|\leq 1$.  Then, for each $i\in [|\mathsf{A}|]$,
    \begin{align}\label{eq:one-step-low-freq-adb}
     \frac{1}{(2\sigma\sqrt{2\pi})^{|\mathsf{A}|/2}} &\bigg| \int_{\vec{\omega}\in \mathsf{U}_\Omega}  \tr\bigg[\sqrt{\vsigma}\vrho^{-1/2}\bigg(\prod_{j=1}^{i-1} \hat{\vA}_j(\omega_j)\bigg) \vrho^{1/2}\sqrt{\vsigma} \bigg(\prod_{j=i} \hat{\vA}_j(\omega_j)\bigg) \vX - \\
        &\quad \sqrt{\vsigma}\vrho^{-1/2}\bigg(\prod_{j=1}^{i} \hat{\vA}_j(\omega_j)\bigg) \vrho^{1/2}\sqrt{\vsigma} \bigg(\prod_{j=i+1} \hat{\vA}_j(\omega_j)\bigg) \vX\bigg] \rd \vec{\omega} \bigg| \leq \sqrt{\mathsf{ADB}_{\vA_i}[\vsigma]}\cdot e^{2\beta\Omega} \cdot r_5(\beta, d)^{|\mathsf{A}|}.
    \end{align} 
\end{lem}

\begin{proof}

    To begin, we note the following norm bounds. From~\autoref{lem:high_freq_bound} for $\beta_0$, and any $k\leq |\mathsf{A}|$:
    \begin{align}
        \bigg\|\prod_{j=1}^k \hat{\vA}_j(\omega_j)\bigg\|  \leq \bigg(\frac{e^{\sigma^2\beta_0^2}}{\sqrt{2\pi\sigma} (1-2d\cdot \beta_0)}\bigg)^{k}
    \end{align}
    and from \autoref{lem:bounds_imaginary_conjugation}, with $\sigma = \beta^{-1}$:
\begin{align}
    \bigg\|\vrho^{-1/2}\bigg(\prod_{j=1}^{k} \hat{\vA}_j(\omega_j)\bigg)\vrho^{1/2}\bigg\|\leq \exp\bigg(\beta \sum_{j=1}^{k} |\omega_j|/2\bigg)\cdot ( 2\sigma^{-1/2})^{k}.
\end{align}
Indeed, the imaginary-time conjugated Paulis have norms diverging with the total frequencies, which is precisely why the high-frequency truncation is introduced. To proceed, by Holder's inequality, we can rewrite the desired integral in terms of norms of the various operator fourier transforms' of single-qubit Pauli operators. 
    \begin{align}
        &\int_{\vec{\omega}\in \mathsf{U}_\Omega} \rd \vec{\omega}\cdot \tr\bigg[\sqrt{\vsigma}\vrho^{-1/2}\bigg(\prod_{j=1}^{i-1} \hat{\vA}_j(\omega_j)\bigg)\vrho^{1/2}\bigg(\sqrt{\vsigma} \hat{\vA}_i(\omega_i) - \vrho^{-1/2}\hat{\vA}_i(\omega_i)\vrho^{1/2} \sqrt{\vsigma} \bigg) \bigg(\prod_{j=i+1} \hat{\vA}_j(\omega_j)\bigg)\vX\bigg]  \\
        &\leq  \int_{\vec{\omega}\in \mathsf{U}_\Omega} \rd \vec{\omega} \bigg\|\vrho^{-1/2}\bigg(\prod_{j=1}^{i-1} \hat{\vA}_j(\omega_j)\bigg)\vrho^{1/2}\bigg\|\cdot \bigg\|\sqrt{\vsigma} \hat{\vA}_i(\omega_i) - \vrho^{-1/2}\hat{\vA}_i(\omega_i)\vrho^{1/2} \sqrt{\vsigma}\bigg\|_2\cdot \bigg\|\prod_{j=i+1} \hat{\vA}_j(\omega_j)\bigg\|. \label{eq:low_freq_exp}
    \end{align}
Next, we leverage the upper bound on the total frequency sum:
    \begin{align}
        \eqref{eq:low_freq_exp} &\leq \bigg(\frac{2e^{\sigma^2\beta_0^2}}{\sqrt{\sigma} (1-2d\beta_0)}\bigg)^{|\mathsf{A}|-1}\cdot \int_{\vec{\omega}\in \mathsf{U}_\Omega}  e^{\beta \sum_j^{i-1} |\omega_j|/2} \bigg\|\sqrt{\vsigma} \hat{\vA}_i(\omega_i) - \vrho^{-1/2}\hat{\vA}_i(\omega_i)\vrho^{1/2} \sqrt{\vsigma}\bigg\|_2 \rd \vec{\omega} \\
        &\leq e^{\beta\Omega/2} \cdot  \bigg(\frac{2e^{\sigma^2\beta_0^2}}{\sqrt{\sigma} (1-2d\beta_0)}\bigg)^{|\mathsf{A}|-1}\int_{\vec{\omega}\in \mathsf{U}_\Omega}  \bigg\|\sqrt{\vsigma} \hat{\vA}_i(\omega_i) - \vrho^{-1/2}\hat{\vA}_i(\omega_i)\vrho^{1/2} \sqrt{\vsigma}\bigg\|_2 \rd \vec{\omega}.\label{eq:abd-single-qubit}
    \end{align}
    The above integral over all frequencies in the truncation $\mathsf{U}_{\Omega}$, can be written in terms of an integral over only the $i$th frequency. Indeed, $\vec{\omega}\in \mathsf{U}_{\Omega}$ entails all the prefix-sums are bounded by $\Omega$; and the frequencies $\vec{\omega}_{\setminus i}$ are restricted to $2^{|\mathsf{A}|-1}$ different $(|\mathsf{A}|-1)$-dimensional simplices, scaled by $\Omega$. By further leveraging the bound on the simplex volume bound (\autoref{lem:simplex_volume}):
    \begin{align}
        &\int_{\vec{\omega}\in \mathsf{U}_\Omega}  \bigg\|\sqrt{\vsigma} \hat{\vA}_i(\omega_i) - \vrho^{-1/2}\hat{\vA}_i(\omega_i)\vrho^{1/2} \sqrt{\vsigma}\bigg\|_2 \rd \vec{\omega}   \\ \leq& \int_{\vec{\omega}_{\setminus i}\in \mathsf{U}_\Omega} \rd \vec{\omega}_{\setminus i}\cdot \int_{|\omega_i|\leq \Omega}  \bigg\|\sqrt{\vsigma} \hat{\vA}_i(\omega_i) - \vrho^{-1/2}\hat{\vA}_i(\omega_i)\vrho^{1/2} \sqrt{\vsigma}\bigg\|_2 \rd \omega_i  \\
        \leq &\frac{(2\Omega)^{|\mathsf{A}|-1}}{(|\mathsf{A}|-1)!}\int_{|\omega_i|\leq \Omega}  \bigg\|\sqrt{\vsigma} \hat{\vA}_i(\omega_i) - \vrho^{-1/2}\hat{\vA}_i(\omega_i)\vrho^{1/2} \sqrt{\vsigma}\bigg\|_2 \rd \omega_i.
    \end{align}
    \noindent It remains to bound the "ABD-like" error expression in the above. Shorthand $\hat{\vA} := \hat{\vA}_i$, 
    \begin{align}
       \int_{|\omega|\leq \Omega} \bigg\|\sqrt{\vsigma} \hat{\vA}(\omega) - \vrho^{-1/2}\hat{\vA}(\omega)\vrho^{1/2} \sqrt{\vsigma}\bigg\| _2\rd \omega 
         &\le \int_{-\infty}^{\infty}\frac{\indicator(\labs{\omega}\le \Omega)}{\sqrt{\gamma(\omega)}} \cdot\sqrt{\gamma(\omega)} \lnorm{\sqrt{\vsigma} \hat{\vA}(\omega)-\vrho^{\frac{1}{2}}\hat{\vA}(\omega) \vrho^{-\frac{1}{2}}\sqrt{\vsigma} }_2\rd \omega \\
    &\le \sqrt{\int_{\labs{\omega}\le \Omega}\frac{1}{\gamma(\omega)} \rd \omega} \cdot \sqrt{\mathsf{ADB}_{\vA}[\vsigma]} \quad  \text{(Cauchy-Schwarz over $\omega$)}\\
    &\lesssim \frac{e^{\beta\Omega/2}}{\beta^{1/2}} \sqrt{\mathsf{ADB}_{\vA}[\vsigma]}. \quad \text{(Using $\sigma = 1/\beta$)}
    \end{align}
    \noindent Put together, we conclude
    \begin{align}
      \eqref{eq:abd-single-qubit}  &\leq \sqrt{\mathsf{ADB}_{\vA_j}[\vsigma]}\cdot e^{\beta\Omega/2}\cdot   \bigg(\frac{2e^{\sigma^2\beta_0^2}}{\sqrt{\sigma} (1-2d\beta_0)}\bigg)^{|\mathsf{A}|-1} \cdot \frac{e^{\beta\Omega/2}}{\sqrt{\beta}}\cdot \undersetbrace{\le e^{\beta \Omega} (2/\beta)^{\labs{\mathsf{A}}-1}}{ \frac{(2\Omega)^{|\mathsf{A}|-1}}{(|\mathsf{A}|-1)!}} \\ &\leq \sqrt{\mathsf{ADB}_{\vA_j}[\vsigma]}\cdot \frac{e^{2\beta\Omega}}{\beta^{1/2}} \cdot \bigg(\frac{4e^{\sigma^2\beta_0^2}}{\sqrt{\sigma}\beta (1-2d\beta_0)}\bigg)^{|\mathsf{A}|-1}, 
    \end{align}
    where we used that for $x>0, a\in \mathbb{Z}^+$: $x^a/a!\leq e^{x}$. When combining with the relevant normalization from \eqref{eq:one-step-low-freq-adb}, and recalling the choice of $\sigma=1/\beta$ and $\beta_0 = 1/4d$, we arrive at the bound 
    \begin{equation}
        \eqref{eq:one-step-low-freq-adb} \leq \sqrt{\mathsf{ADB}_{\vA_j}[\vsigma]}\cdot e^{2\beta\Omega} \cdot (8e^{\beta_0^2/\beta^2})^{|\mathsf{A}|}.
    \end{equation}
\end{proof}

\subsubsection{Proof of \autoref{lem:expose_KMS}}

\begin{proof}

We invoke a bound on the below inner product in terms of the $\vsigma$-weighted KMS-norm of $\vX$.
    \begin{align}
       | \tr[\vsigma^{1/2}\vrho^{-1/2}\vP_{\mathsf{U}_\Omega}\vrho \vQ_{\mathsf{U}_\Omega}\vrho^{-1/2}\vsigma^{1/2}\vX] |&\leq \|\vX\|_{\vsigma}  \cdot \bigg\|\vrho^{-1/2}\vP_{\mathsf{U}_\Omega}\vrho \vQ_{\mathsf{U}_\Omega}\vrho^{-1/2}\bigg\|_{\vsigma} \\
       &\leq \|\vX\|_{\vsigma}  \cdot \bigg\|\vrho^{-1/2}\vP_{\mathsf{U}_\Omega}\vrho^{1/2} \bigg\|\cdot \bigg\|\vrho^{1/2} \vQ_{\mathsf{U}_\Omega}\vrho^{-1/2}\bigg\|,
    \end{align}
    \noindent where the imaginary time-conjugated low frequency term is bounded by the same argument in \eqref{eq:abd-single-qubit}
\begin{align}
    \bigg\|\vrho^{1/2} \vQ_{\mathsf{U}_\Omega}\vrho^{-1/2}\bigg\| \le \frac{1}{(2\sigma\sqrt{2\pi})^{|\mathsf{A}|/2}} \cdot \bigg(\frac{2}{\sqrt{\sigma}}\bigg)^{\labs{\mathsf{A}}} \cdot \frac{(2\Omega)^{|\mathsf{A}|}}{|\mathsf{A}|!}\cdot e^{\beta\Omega} \le  \bigg(\frac{2\sqrt{2}}{\sigma\beta }\bigg)^{\labs{\mathsf{A}}} e^{2\beta\Omega}. 
\end{align}
    \noindent Put together with \autoref{lem:high_freq_trunc}, \autoref{lem:low_freq_trunc} one can express, by the triangle inequality, 
    \begin{align}
    \labs{\tr[\vP\vsigma \vQ \vX]} \leq \bigg[\bigg(\|\vX\|_{\vsigma} + \max_{\vA\in P^1_\mathsf{A}}\sqrt{\mathsf{ADB}_{\vA}[\vsigma]} \bigg)  e^{4\beta \Omega} + e^{-r_1(d)\Omega} \bigg] \cdot \exp\bigg[|\mathsf{A}|\cdot \mathsf{poly}(\beta, 1/\beta, d)\bigg]
\end{align}

\noindent A careful choice of the parameter $\Omega>0$,
\begin{equation}
   e^{\Omega} =  \max\bigg[\bigg(\under{\|\vX\|_{\vsigma} + \max_{\vA\in P^1_\mathsf{A}}\sqrt{\mathsf{ADB}_{\vA}[\vsigma]}}{\delta} \bigg)^{-\frac{1}{4\beta+r_1(d)}}, 1\bigg],
\end{equation}

\noindent then gives the desired bound. The fact that the bound becomes trivial if $\delta \geq 1$ concludes the proof. 
    
\end{proof}

\section{Metastable States Arise at Typical Times}
\label{sec:metastable_time_averaging}

This appendix contains two simple lemmas, proving that in time $T$ the time-evolution of a generic Lindbladian $\mathcal{L}$ dynamics (probably) reaches an approximately stationary state, up to an error of inverse $T$. What is nice about these simple facts is that they are agnostic to any structure of $\mathcal{L}$ such as spectral gaps. That is to say, the dynamics may not mix quickly, but nevertheless, it quickly approaches stationarity. 

We begin with the lemma advertized in the main text (\autoref{lem:stab_typical_times}): any initial state $\vsigma$ under evolution $e^{t\mathcal{L}}$ at random times is approximately stationary. 

\begin{proof}

[Proof of~\autoref{lem:stab_typical_times}]
    We begin by expressing the equation of motion:
    \begin{equation}
        \frac{\rd}{\rd t} \vsigma_t = \mathcal{L}[ \vsigma_t],
    \end{equation}

    \noindent which, after integration, entails:
    \begin{equation}
      \vsigma_T - \vsigma =  \int^T_0 \rd t \cdot \mathcal{L}[ \vsigma_t] = \mathcal{L}\bigg[ \int^T_0 \rd t \vsigma_t\bigg] = T\cdot \mathcal{L}\big[ \overline{\vsigma}_T].
    \end{equation}
    \noindent Take the trace norm on both sides to yield the desired result. 
\end{proof}

What is somewhat unsatisfactory about \autoref{lem:stab_typical_times} is that it is not a ``with-high-probability" statement about a specific state encountered during the evolution; instead, it is just about the regularized, time-averaged state. In the following lemma, we reason that the expected value of the entropy production rate / Fisher information (see~\autoref{defn:fisher} and \autoref{thm:integral_square_logs}) also decays with the evolution time $t$, which is a slightly sharper statement. The following calculation is the direct quantum translation of \cite[Lemma 3.1]{liu2024locally}.

\begin{lem}
    [The Fisher Information at Typical Times] Let $\mathcal{L}$ denote a Lindbladian with fixed point $\vrho$, acting on some initial state $\vsigma$, and fix a cutoff time $T>0$. Then, for $t$ chosen uniformly at random in $[0, T]:$
    \begin{equation}
        \mathbb{E}_{t\sim [T]} \mathsf{FI}[\vsigma_t||\vrho] = \mathbb{E}_{t\sim [T]} \mathsf{EP}_{\mathcal{L}}[\vsigma_t] \leq \frac{\|\log \vrho\|}{T},
    \end{equation}

    \noindent with $\vsigma_t = e^{t\mathcal{L}}[\vsigma]$ the time-evolved state.
\end{lem}

Markov's inequality then trivially implies that with probability $1-\sqrt{\epsilon}$ over the choice of $t\in [0, T]$, the entropy production rate of $\vsigma_t$ is $\leq \sqrt{\epsilon}$ with $\epsilon = \|\log \vrho\|/ T$.

\begin{proof}
    We begin by expressing the definition of the entropy production rate:
    \begin{equation}
        \frac{\rd}{\rd t} D(\vsigma_t||\vrho) := - \mathsf{EP}_\CL[\vsigma_t],
    \end{equation}
    \noindent which, after integration, entails:
    \begin{equation}
        \mathbb{E}_{t\sim T} \mathsf{EP}_\CL[\vsigma_t] = \frac{1}{T}\int_0^T \rd t\cdot  \mathsf{EP}_\CL[\vsigma_t] = \frac{1}{T}\cdot \bigg(D(\vsigma_0||\vrho) - D(\vsigma_T||\vrho)\bigg) \leq \frac{\|\log\vrho\|}{T}
    \end{equation}
    \noindent as desired. 
\end{proof}

\section{Proofs of Main Results}
\label{section:smoothing}

We dedicate this section to complete the proofs of the Main Results (Section~\ref{section:results}). Through our individual lemmas, we have liberally invoked $\norm{\log (\vsigma)} $, which, strictly speaking, diverges logarithmically with the smallest eigenvalue of $\vsigma.$  To draw formal conclusions for metastable states that may have zero eigenvalues, we introduce a convenient \textit{regularization} trick that removes the logarithmic divergences. Given an $\epsilon$-metastable state $\vsigma$, we define the regularized metastable state by a convex combination of $\vsigma$ with the Gibbs state $\vrho$:
\begin{equation}
    \vsigma \quad \rightarrow  \quad \vsigma_\delta = (1-\delta) \cdot \vsigma  + \delta \cdot \vrho\quad \text{for any}\quad \delta \in [0, 1].
\end{equation}

In sequence, we prove that the regularized states have a bounded ``effective Hamiltonian'' $\log(\vsigma_{\delta}),$ entropy production, and approximate detailed balance error. This gives a local Markov property for the regularized metastable state, which can be translated via continuity to a Markov property for the original metastable state, i.e. \autoref{thm:main_meta_implies_markov}. By combining with our implication that the local Markov property implies an area law (\autoref{thm:arealaw_metastable}), we conclude the proof of the main theorem in \autoref{thm:main_meta_implies_area}.

\begin{lem}
    [Regularized Effective Hamiltonian]\label{lem:smoothed_h} For any $\delta\in [0, \frac{1}{2}]$, any state $\vsigma$ and full rank state $\vrho$, $\vsigma_\delta = (1-\delta) \cdot \vsigma  + \delta \cdot \vrho$ satisfies:
    \begin{equation}
         \|\log \vsigma_\delta\| \lesssim  \log \frac{1}{\delta} + \|\log \vrho\| \equiv \Delta_\delta
    \end{equation}
\end{lem}

\begin{proof}
    Regularization ensures the lowest eigenvalue of $\vsigma_\delta$ is lower-bounded by $\lambda_{min}(\vsigma_\delta) \geq \delta\cdot \lambda_{min}(\vrho) $, which gives the desired bound. 
\end{proof}

In the above, we've defined $\Delta_\delta$ to curtail future notation.

\begin{lem}
    [Regularized Entropy Production] For any $\epsilon, \delta\in [0, \frac{1}{2}]$, and any state $\vsigma$ which is $\epsilon$-metastable w.r.t. a Lindbladian $\CL$ with full-rank fixed point $\vrho$: 
    \begin{equation}
        \mathsf{EP}[\vsigma_\delta] \lesssim \epsilon\cdot \Delta_\delta.
    \end{equation}
\end{lem}
\begin{proof}
    We begin by noting that $\vsigma_\delta$, trivially remains metastable:
    \begin{equation}
         \|\CL[\vsigma_\delta]\|_1 = (1-\delta)\cdot \|\CL[\vsigma]\|_1 = (1-\delta)\cdot \epsilon.
    \end{equation}
    The desired claim then follows from Holder's inequality and \autoref{lem:smoothed_h}:
    \begin{align}
        \mathsf{EP}[\vsigma_\delta] \leq   \|\CL[\vsigma_\delta]\|_1\cdot \|\log \vsigma_\delta -\log\vrho\| \lesssim \epsilon\cdot \bigg(\log \frac{1}{\delta} + \|\log \vrho\|\bigg) = \epsilon\cdot \Delta_\delta
    \end{align}
\end{proof}

\begin{lem}
    [Regularized Approximate Detailed Balance]\label{lem:smoothed_adb} For any $\epsilon, \delta\in [0, \frac{1}{2}]$, any state $\vsigma$ which is $\epsilon$-metastable w.r.t. the Lindbladian $\CL$ \eqref{eq:lindblad_def}  with full rank fixed point $\vrho$, and any jump operator $\vA^a$ s.t. $\|\vA^a\|\leq 1$:
    \begin{equation}
        \mathsf{ADB}_a[\vsigma_\delta] \lesssim \epsilon  \cdot  \Delta_\delta\cdot  \bigg(\log\frac{1}{\epsilon} + \log \Delta_\delta\bigg).
    \end{equation}
\end{lem}

\begin{proof}
    From linearity of $\mathsf{EP}$ and the derived relationship between entropy production and approximate detailed balance (\autoref{thm:meta_implie_ADB}),
    \begin{align}
        \mathsf{ADB}_a[\vsigma_\delta] &\lesssim \mathsf{EP}[\vsigma_\delta]\cdot \bigg(1+ \log \frac{\|\log \vsigma_\delta -\log\vrho\|^2}{\mathsf{EP}[\vsigma_\delta]}\bigg) \\
        &\leq \epsilon \cdot \|\log \vsigma_\delta -\log\vrho\|\cdot \bigg(1 + \log \frac{\|\log \vsigma_\delta -\log\vrho\|}{\epsilon}\bigg) \quad \text{(Since $-x\log x, x\leq \frac{1}{4}$ is monotonic)} \\
        &\lesssim \epsilon  \cdot  \Delta_\delta\cdot  \bigg(1 + \log\frac{1}{\epsilon} + \log \Delta_\delta\bigg) \quad \quad \quad \text{(From \autoref{lem:smoothed_h})}.
    \end{align}
    \noindent Use that $\delta, \epsilon \leq \frac{1}{2}$ to conclude the proof. 
\end{proof}

\begin{lem}
    [Regularized Markov Property]\label{lem:smoothed_markov} Fix $\epsilon, \delta\in [0, \frac{1}{4}]$, and any $\epsilon$-metastable state $\vsigma$ w.r.t. the Lindbladian $\CL$ \eqref{eq:lindblad_def} with full rank fixed point $\vrho$. In the context of \autoref{thm:local_recovery}, there exists a recovery map $\CR$ s.t. the regularized $\epsilon$-metastable state $\vsigma_\delta$ satisfies 
    \begin{equation}
       \norm{\vsigma_\delta-\CR[\CN_{\mathsf{A}}[\vsigma_{\delta} ]}_1 \leq \e^{\mu|\mathsf{A}|} \cdot \epsilon^\lambda\cdot \Delta_\delta.
    \end{equation}
    for constants $\mu>0, 1>\lambda>0$ which depend only polynomially on $\beta, \beta^{-1}, d$. 
\end{lem}

\begin{proof}
     From \autoref{lem:smoothed_adb}, we have that for any $\epsilon, \delta\in [0, 1/4]$ and $\epsilon$-metastable state $\vsigma$ w.r.t the Lindbladian $\CL$ of single site Pauli jumps, 
    \begin{align}
       \epsilon_{\mathsf{ADB}} &= \max_{\vP\in P_{[n]}^1}\mathsf{ADB}_a[\vsigma_\delta] \lesssim \epsilon  \cdot  \Delta_\delta\cdot  \bigg(\log\frac{1}{\epsilon} + \log \Delta_\delta\bigg) \lesssim \epsilon^{1/2} \cdot \Delta_\delta\cdot\log \Delta_\delta  \label{eq:smoothed_adb_error}
    \end{align}
    \noindent where we used $\epsilon\in [0, \frac{1}{2}]:\epsilon^{1/2}\log \epsilon \leq 1$ and $\delta\in [0, 1/4]:\Delta_\delta \geq 2$. 
    
    Now, from the relationship between approximate detailed balance and the local Markov property \autoref{thm:local_recovery} (with constants $\mu_1>0, 1>\lambda_1>0$), and an appropriate choice of $t$, 
    \begin{equation}
        t = (e^{\mu_1 |\mathsf{A}|}/\epsilon_{\mathsf{ADB}})^{1/(1+\lambda_1)}
    \end{equation}
    we then have 
    \begin{equation}
        \norm{\vsigma_\delta-\CR_{\mathsf{A},t}[\CN_{\mathsf{A}}[\vsigma_{\delta} ]}_1 \leq |\mathsf{A}|\cdot e^{ |\mathsf{A}|\mu_1/(1+\lambda_1)}\cdot \epsilon_{\mathsf{ADB}}^{\lambda_1/(1+\lambda_1)} \leq e^{ |\mathsf{A}|\mu_2}\cdot \epsilon_{\mathsf{ADB}}^{\lambda_2}\label{eq:smoothed_markov_step1}
    \end{equation}
    for an appropriate choice of $\mu_2>0, 1>\lambda_2>0$. Now, leveraging our bound on the regularized $\mathsf{ADB}$ error in \eqref{eq:smoothed_adb_error}, we have
    \begin{align}
    \eqref{eq:smoothed_markov_step1} \lesssim \e^{\mu_2|\mathsf{A}|}\cdot \epsilon^{\lambda_2/2}\cdot \Delta_\delta^{\lambda_2}\log^{\lambda_2} \Delta_\delta
 \end{align}
 \noindent Finally, to simplify the result, we use the fact $\lambda_2 < 1$ and $\Delta_\delta \geq 2$ to ``absorb" the residual log factors:
 \begin{align}
     x\geq 2\quad \Rightarrow \quad \quad x^\lambda\log x \leq C_\lambda x \quad \text{ with }\quad C_\lambda \leq \frac{1}{e(1-\lambda)} \label{eq:absorb_log}
 \end{align} 
 The claimed bound then follows absorbing $C_\lambda$ into the definition of $\mu$ and $\lambda = \lambda_2/2$.
\end{proof}

 We are now in a position to prove the local Markov property for the original metastable state.

 \begin{proof}

     [of the Local Markov Property, \autoref{thm:main_meta_implies_markov}] Note that the claim becomes trivial unless $\epsilon\leq \frac{1}{4}$. By continuity, we can thus relate the local recovery back to that of $\vsigma$, through \autoref{lem:smoothed_markov}:
     \begin{align}
     \norm{\vsigma-\CR[\CN_{\mathsf{A}}[\vsigma] ]}_1 \lesssim \delta +\e^{\mu|\mathsf{A}|}\cdot \epsilon^{\lambda}\cdot \Delta_\delta\label{eq:smoothed_markov}
 \end{align}
Finally, we make the explicit choice of $\delta = \e^{\mu|\mathsf{A}|}\cdot \epsilon^{\lambda}$ to arrive at:
 \begin{equation}
      \norm{\vsigma-\CR[\CN_{\mathsf{A}}[\vsigma] ]}_1 \lesssim \e^{\mu_3|\mathsf{A}|}\cdot \epsilon^{\lambda} \cdot \max \bigg( n, \log \frac{1}{\epsilon},  \|\log \vrho\|\bigg)
 \end{equation}

\noindent under an appropriate update to $\mu_3$ to capture the $\beta, d$ dependent pre-factors. Next, we use the bound on the norms of $\vrho, \vH$: $\|\log \vrho\|\lesssim n+\beta\|\vH\|\leq (1+\beta \cdot \mathsf{poly}(d))n$. Finally, $y\in [0, \frac{1}{4}]: y^{\lambda}\log \lambda \leq C_\lambda' y^{\lambda/2}$, to arrive at 
\begin{equation}
     \norm{\vsigma-\CR[\CN_{\mathsf{A}}[\vsigma] ]}_1 \lesssim n\cdot e^{\mu_4|\mathsf{A}|} \epsilon^{\lambda_4}
\end{equation}
for appropriate constants $\mu_4> 0$ and $1>\lambda_4>0$ as desired.    
\end{proof}

Finally, by combining the above with our statement in \autoref{thm:arealaw_metastable} that the local Markov property implies the area law, we conclude with the proof of the main theorem of this paper. 

\begin{proof}

    [of the Area Law, \autoref{thm:main_meta_implies_area}] With the local Markov error $\epsilon_{\mathsf{Markov}}$ as previously derived in \autoref{thm:main_meta_implies_markov}, we have that the correction to the area law as captured in \autoref{thm:arealaw_metastable}  is given by
    \begin{equation}
      \lesssim  \epsilon_{\mathsf{Markov}}\cdot \max\bigg(n, \beta\|\vH\|, \log\frac{1}{\epsilon_{\mathsf{Markov}}}\bigg) \lesssim n\cdot e^{\mu'|\mathsf{A}|} \epsilon^{\lambda}\cdot \bigg(n+\log \frac{1}{\epsilon}\bigg) \leq n^2\cdot e^{\mu''|\mathsf{A}|} \epsilon^{\lambda/2}.
    \end{equation}
    Where we once again set an adequate $\mu$ to absorb temperature-dependent pre-factors. This concludes the proof. 
\end{proof}

\end{document}